\documentclass[11pt]{article}
\usepackage[T1]{fontenc}
\usepackage[margin=1in]{geometry}

\usepackage{microtype}
\usepackage{graphicx}
\usepackage{subfig}
\usepackage{caption}
\usepackage{booktabs} %

\usepackage{amsmath}

\usepackage[colorlinks=true,allcolors=blue,hyperfootnotes=false]{hyperref}
\usepackage[capitalise]{cleveref}

\title{Parameter Estimation in Epidemic Spread Networks Using Limited Measurements\thanks{
{Research supported in part by the National Science Foundation, grants NSF-CMMI 1635014 and NSF-ECCS 2032258.}}}

\author{%
Lintao Ye\thanks{Department of Electrical Engineering, University of Notre Dame, Notre Dame, IN, USA (lye2@nd.edu)}
\and
Philip E. Par{\'e}\thanks{School of Electrical and Computer Engineering, Purdue University, West Lafayette, IN, USA
  (\{philpare,sundara2\}@purdue.edu)}
\and
Shreyas Sundaram\footnotemark[3]}

\RequirePackage{mathtools}
\RequirePackage{dsfont}
\RequirePackage{delimset}

\usepackage{amsthm}
\usepackage{amssymb}
\usepackage{thmtools}

\newtheorem{theorem}{Theorem}[section]
\newtheorem{problem}[theorem]{Problem}

\newtheorem{lemma}[theorem]{Lemma}
\newtheorem{remark}[theorem]{Remark}
\newtheorem{definition}[theorem]{Definition}
\newtheorem{proposition}[theorem]{Proposition}

\newtheorem{assumption}[theorem]{Assumption}
\newtheorem{fact}{Fact}

\usepackage{enumitem}
\usepackage[numbers]{natbib}
\usepackage{algorithm}
\usepackage[noend]{algpseudocode}

\algnewcommand{\IfThenElse}[3]{%
  \State \algorithmicif\ #1\ \algorithmicthen\ #2\ \algorithmicelse\ #3}
  
 \DeclareMathOperator{\Tr}{tr}
\DeclareMathOperator{\rank}{rank}

\begin{document}

\maketitle

\begin{abstract}
We study the problem of estimating the parameters (i.e., infection rate and recovery rate) governing the spread of epidemics in networks.  Such parameters are typically estimated by measuring various characteristics (such as the number of infected and recovered individuals) of the infected populations over time.  However, these measurements also incur certain costs, depending on the population being tested and the times at which the tests are administered.  We thus formulate the epidemic parameter estimation problem as an optimization problem, where the goal is to either minimize the total cost spent on collecting measurements, or to optimize the parameter estimates while remaining within a measurement budget.  We show that these problems are NP-hard to solve in general, and then propose approximation algorithms with performance guarantees.  We validate our algorithms using numerical examples.

\end{abstract}

\section{Introduction}
Models of spreading processes over networks have been widely studied by researchers from different fields (e.g., \cite{kempe2003maximizing,pastor2015epidemic,chakrabarti2008epidemic,ahn2013global,preciado2014optimal,newman2002spread}).  The case of epidemics spreading through networked populations has received a particularly significant amount of attention, especially in light of the ongoing COVID-19 pandemic (e.g., \cite{newman2002spread,nowzari2016analysis}).  A canonical example is the networked SIR model, where each node in the network represents a subpopulation or an individual, and can be in one of three states: susceptible (S), infected (I), or recovered (R)  \cite{mei2017dynamics}.  There are two key parameters that govern such models: the infection rate of a given node, and the recovery rate of that node.  In the case of a novel virus, these parameters may not be known a priori, and must be identified or estimated from gathered data, including for instance the number of infected and recovered individuals in the network at certain points of time. For instance, in the COVID-19 pandemic, when collecting the data on the number of infected individuals or the number of recovered individuals in the network, one possibility is to perform virus or antibody tests on the individuals, with each test incurring a cost. Therefore, in the problem of parameter estimation in epidemic spread networks, it is important and of practical interest to take the costs of collecting the data (i.e., measurements) into account, which leads to the problem formulations considered in this paper. The goal is to exactly identify (when possible)  or estimate the parameters in the networked SIR model using a limited number of measurements. Specifically, we divide our analysis into two scenarios: 1) when the measurements (e.g., the number of infected individuals) can be collected exactly without error; and 2) when only a stochastic measurement can be obtained. 
 
Under the setting when exact measurements of the infected and recovered proportions of the population at certain nodes in the network can be obtained, we formulate the Parameter Identification Measurement Selection (PIMS) problem as minimizing the cost spent on collecting the measurements, while ensuring that the parameters of the SIR model can be uniquely identified (within a certain time interval in the epidemic dynamics). In settings where the measurements are stochastic (thereby precluding exact identification of the parameters), we formulate the Parameter Estimation Measurement Selection (PEMS) problem.  The goal is to optimize certain estimation metrics based on the collected measurements, while satisfying the budget on collecting the measurements.

\subsection*{Related Work}
The authors in \cite{pare2018analysis,vrabac2020overcoming} studied the parameter estimation problem in epidemic spread networks using a "Susceptible-Infected-Susceptible (SIS)" model of epidemics. model. When exact measurements of the infected proportion of the population at each node of the network can be obtained, the authors proposed a sufficient and necessary condition on the set of the collected measurements such that the parameters of the SIS model (i.e., the infection rate and the recovery rate) can be uniquely identified. However, this condition does not pose any constraint on the number of measurements that can be collected. 

In \cite{pezzutto2020smart}, the authors considered a measurement selection problem in the SIR model. Their problem is to perform a limited number of virus tests among the population such that the probability of undetected asymptotic cases is minimized. The transmission of the disease in the SIR model considered in \cite{pezzutto2020smart} is characterized by a Bernoulli random variable which leads to a Hidden Markov Model for the SIR dynamics. 

Finally, our work is also closely related to the sensor placement problem that has been studied for control systems (e.g., \cite{mo2011sensor,ye2018complexity,ye2020resilient}), signal processing (e.g., \cite{chepuri2014sparsity,ye2019sensor}), and machine learning (e.g., \cite{krause2008near}). The goal of these problems is to optimize certain (problem-specific) performance metrics of the estimate based on the measurements of the placed sensors, while satisfying the sensor placement budget constraints. 

\subsection*{Contributions}
First, we show that the PIMS problem is NP-hard, which precludes polynomial-time algorithms for the PIMS problem (if P $\neq$ NP). By exploring structural properties of the PIMS problem, we provide a polynomial-time \textit{approximation} algorithm which returns a solution that is within a certain approximation ratio of the optimal. The approximation ratio depends on the cost structure of the measurements and on the graph structure of the epidemic spread network.  Next, we show that the PEMS problem is also NP-hard. In order to provide a polynomial-time approximation algorithm that solves the PEMS problem with performance guarantees, we first show that the PEMS problem can be transformed into the problem of maximizing a set function subject to a knapsack constraint. We then apply a greedy algorithm to the (transformed) PEMS problem, and provide approximation guarantees for the greedy algorithm. Our analysis for the greedy algorithm also generalizes the results from the literature for maximizing a submodular set function under a knapsack constraint to nonsubmodular settings.  We use numerical examples to validate the obtained performance bounds of the greedy algorithm, and show that the greedy algorithm performs well in practice.

\subsection*{Notation and Terminology}
The sets of integers and real numbers are denoted as $\mathbb{Z}$ and $\mathbb{R}$, respectively. For a set $\mathcal{S}$, let $|\mathcal{S}|$ denote its cardinality. For any $n\in\mathbb{Z}_{\ge1}$, let $[n]\triangleq\{1,2,\dots,n\}$. Let $\mathbf{0}_{m\times n}$ denote a zero matrix of dimension $m\times n$; the subscript is dropped if the dimension can be inferred from the context. For a matrix $P\in\mathbb{R}^{n\times n}$, let $P^{\top}$, $\Tr(P)$ and $\det(P)$ be its transpose, trace and determinant, respectively. The eigenvalues of $P$ are ordered such that $|\lambda_1(P)|\ge\cdots\ge|\lambda_n(P)|$. Let $P_{ij}$ (or $(P)_{ij}$) denote the element in the  $i$th row and $j$th column of $P$, and let $(P)_i$ denote the $i$th row of $P$. A positive semidefinite matrix $P\in\mathbb{R}^{n\times n}$ is denoted by $P\succeq\mathbf{0}$. 

\section{Model of Epidemic Spread Network}
Suppose a disease (or virus) is spreading over a directed graph $\mathcal{G}=\{\mathcal{V},\mathcal{E}\}$, where $\mathcal{V}\triangleq[n]$ is the set of $n$ nodes, and $\mathcal{E}$ is the set of directed edges (and self loops) that captures the interactions among the nodes in $\mathcal{V}$. Here, each node $i\in\mathcal{V}$ is considered to be a group (or population) of individuals (e.g., a city or a country). A directed edge from node $i$ to node $j$, where $i\neq j$, is denoted by $(i,j)$. For all $i\in\mathcal{V}$, denote $\mathcal{N}_i\triangleq\{j:(j,i)\in\mathcal{E}\}$ and $\bar{\mathcal{N}}_i\triangleq\{j:(j,i)\in\mathcal{E}\}\cup\{i\}$. For all $i\in\mathcal{V}$ and for all $k\in\mathbb{Z}_{\ge0}$, we let $s_i[k]$, $x_i[k]$ and $r_i[k]$ represent the proportions of population of node $i\in\mathcal{V}$ that are susceptible, infected and recovered  at time $k$, respectively. To describe the dynamics of the spread of the disease in $\mathcal{G}$, we will use the following discrete-time SIR model (e.g., \cite{hota2020closed}):
\begin{subequations}
\label{eqn:SIR single node}
\begin{align}
&s_i[k+1]=s_i[k]-hs_i[k]\beta\sum_{j\in\bar{\mathcal{N}}_i}a_{ij}x_j[k],\label{eqn:SIR S}\\
&x_i[k+1]=(1-h\delta)x_i[k]+hs_i[k]\beta\sum_{j\in\bar{\mathcal{N}}_i}a_{ij}x_j[k],\label{eqn:SIR I}\\
&r_i[k+1]=r_i[k]+h\delta x_i[k],\label{eqn:SIR R}
\end{align}
\end{subequations}
where $\beta\in\mathbb{R}_{\ge0}$ is the infection rate of the disease, $\delta\in\mathbb{R}_{\ge0}$ is the recovery rate of the disease, $h\in\mathbb{R}_{\ge0}$ is the sampling parameter, and $a_{ij}\in\mathbb{R}_{\ge0}$ is the weight associated with edge $(j,i)$. Let $A\in\mathbb{R}^{n\times n}$ be a weight matrix, where $A_{ij}=a_{ij}$ for all $i,j\in\mathcal{V}$ such that $j\in\bar{\mathcal{N}}_i$, and $A_{ij}=0$ otherwise. Denote $s[k]\triangleq\begin{bmatrix}s_1[k] & \cdots & s_n[k]\end{bmatrix}^T\in\mathbb{R}^n$, $x[k]\triangleq\begin{bmatrix}x_1[k] & \cdots & x_n[k]\end{bmatrix}^T\in\mathbb{R}^n$, and $r[k]\triangleq\begin{bmatrix}r_1[k] & \cdots & r_n[k]\end{bmatrix}^T\in\mathbb{R}^n$, for all $k\in\mathbb{Z}_{\ge0}$. Throughout this paper, we assume that the weight matrix $A\in\mathbb{R}^{n\times n}$ and the sampling parameter $h\in\mathbb{R}_{\ge0}$ are given. 

\section{Preliminaries}
\label{sec:preliminaries}
We make the following assumptions on the initial conditions $s[0]$, $x[0]$ and $r[0]$, and the parameters of the SIR model in Eq.~\eqref{eqn:SIR single node} (e.g., \cite{pare2018analysis,hota2020closed}).

\begin{assumption}
\label{ass:initial condition}
For all $i\in\mathcal{V}$, $s_i[0]\in(0,1]$, $x_i[0]\in[0,1)$, $r_i[0]=0$, and $s_i[0]+x_i[0]=1$. 
\end{assumption}

\begin{assumption}
\label{ass:model parameters}
Assume that $h,\beta,\delta\in\mathbb{R}_{>0}$ with $h\delta<1$. For all $i,j\in\mathcal{V}$ with $(j,i)\in\mathcal{E}$ and $i\neq j$, assume that $a_{ij}\in\mathbb{R}_{>0}$. For all $i\in\mathcal{V}$, $h\beta\sum_{j\in\bar{\mathcal{N}}_i}a_{ij}<1$.
\end{assumption}

\begin{definition}
\label{def:distance}
Consider a directed graph $\mathcal{G}=\{\mathcal{V},\mathcal{E}\}$ with $\mathcal{V}=[n]$. A directed path of length $t$ from node $i_0$ to node $i_t$ in $\mathcal{G}$ is a sequence of $t$ directed edges $(i_0,i_1),\dots,(i_{t-1},i_t)$. For any distinct pair of nodes $i,j\in\mathcal{V}$ such that there exists a directed path from $i$ to $j$, the distance from node $i$ to node $j$, denoted as $d_{ij}$, is defined as the shortest length over all such paths.
\end{definition}

\begin{definition}
\label{def:sets of initial states}
Define $\mathcal{S}_I\triangleq\{i:x_i[0]>0,i\in\mathcal{V}\}$ and $\mathcal{S}_H\triangleq\{i:x_i[0]=0,i\in\mathcal{V}\}$. For all $i\in\mathcal{S}_H$, define $d_i\triangleq\mathop{\min}_{j\in\mathcal{S}_I}d_{ji}$ where $d_i\ge1$, and $d_i\triangleq+\infty$ if there is no path from $j$ to $i$ for any $j\in\mathcal{S}_I$. For all $i\in\mathcal{S}_I$, define $d_i\triangleq0$.
\end{definition}

Using similar arguments to those in \cite{hota2020closed}, one can show that $s_i[k],x_i[k],r_i[k]\in[0,1]$ with $s_i[k]+x_i[k]+r_i[k]=1$ for all $i\in\mathcal{V}$ and for all $k\in\mathbb{Z}_{\ge0}$ under Assumptions~$\ref{ass:initial condition}$-$\ref{ass:model parameters}$. Thus, given $x_i[k]$ and $r_i[k]$, we can obtain $s_i[k]=1-x_i[k]-r_i[k]$ for all $i\in\mathcal{V}$ and for all $k\in\mathbb{Z}_{\ge0}$. We also have the following result that characterizes properties of $x_i[k]$ and $r_i[k]$ in the SIR model over $\mathcal{G}$ given by Eq.~\eqref{eqn:SIR single node}; the proof can be found in Section~$\ref{sec:proof of results about SIR}$ in the Appendix.
\begin{lemma}
\label{lemma:propagate of initial condition}
Consider a directed graph $\mathcal{G}=\{\mathcal{V},\mathcal{E}\}$ with $\mathcal{V}=[n]$ and the SIR dynamics given by Eq.~\eqref{eqn:SIR single node}. Suppose Assumptions~$\ref{ass:initial condition}$-$\ref{ass:model parameters}$ hold. Then, the following results hold for all $i\in\mathcal{V}$, where $k\in\mathbb{Z}_{\ge0}$, and $\mathcal{S}_H$ and $d_i$ are defined in Definition~$\ref{def:sets of initial states}$.\\
$(a)$ $s_i[k]>0$ for all $k\ge0$.\\
$(b)$ If $d_i\neq+\infty$, then $x_i[k]=0$ for all $k<d_i$, and $x_i[k]\in(0,1)$ for all $k\ge d_i$.\footnote{Note that for the case when $d_i=0$, i.e., $i\in\mathcal{S}_I$, part $(b)$ implies $x_i[k]>0$ for all $k\ge0$.}\\
$(c)$ If $d_i\neq+\infty$, then $r_i[k]=0$ for all $k\le d_i$, and $r_i[k]\in(0,1)$ for all $k>d_i$.\\
$(d)$ If $i\in\mathcal{S}_H$ with $d_i=+\infty$, then $x_i[k]=0$ and $r_i[k]=0$ for all $k\ge0$.
\end{lemma}

\section{Measurement Selection Problem in Exact Measurement Setting}\label{sec:perfect measurement}
Throughout this section, we assume that $\mathcal{S}_I,\mathcal{S}_H\subseteq\mathcal{V}$ defined in Definition~$\ref{def:sets of initial states}$ are known, i.e., we know the set of nodes in $\mathcal{V}$ that have infected individuals initially. 

\subsection{Problem Formulation}
Given exact measurements of $x_i[k]$ and $r_i[k]$ for a subset of nodes, our goal is to estimate (or uniquely identify, if possible) the unknown parameters $\beta$ and $\delta$. Here, we consider the scenario where collecting the measurement of $x_i[k]$ (resp., $r_i[k]$) at any node $i\in\mathcal{V}$ and at any time step $k\in\mathbb{Z}_{\ge0}$ incurs a cost, denoted as $c_{k,i}\in\mathbb{R}_{\ge0}$ (resp., $b_{k,i}\in\mathbb{R}_{\ge0}$). Moreover, we can only collect the measurements of $x_i[k]$ and $r_i[k]$ for $k\in\{t_1,t_1+1,\dots,t_2\}$, where $t_1,t_2\in\mathbb{Z}_{\ge0}$ are given with $t_2> t_1$. Noting that Lemma~$\ref{lemma:propagate of initial condition}$ provides a (sufficient and necessary) condition under which $x_i[k]=0$ (resp., $r_i[k]=0$) holds, we see that one does not need to collect a measurement of $x_i[k]$ (resp., $r_i[k]$) if $x_i[k]=0$ (resp., $r_i[k]=0$) from Lemma~$\ref{lemma:propagate of initial condition}$. Given time steps $t_1,t_2\in\mathbb{Z}_{\ge0}$ with $t_2> t_1$, we define a set 
\begin{equation}
\label{eqn:all candidate measurements}
\mathcal{I}_{t_1:t_2}\triangleq\{\lambda_i[k]:k\in\{t_1,\dots,t_2\},i\in\mathcal{V},\lambda_i[k]>0,\lambda\in\{x,r\}\},
\end{equation}
which represents the set of all candidate measurements from time step $t_1$ to time step $t_2$. To proceed, we first use  Eq.~\eqref{eqn:SIR I}-\eqref{eqn:SIR R} to obtain
\begin{equation}
\label{eqn:rewrite SIR matrix}
\begin{bmatrix}x[t_1+1]-x[t_1]\\ \vdots\\ x[t_2]-x[t_2-1]\\ r[t_1+1]-r[t_1]\\ \vdots\\ r[t_2]-r[t_2-1]\end{bmatrix}=h\begin{bmatrix}\Phi_{t_1:t_2}^x\\ \Phi_{t_1:t_2}^r\end{bmatrix}\begin{bmatrix}\beta\\ \delta\end{bmatrix},
\end{equation}
where $\Phi_{t_1:t_2-1}^x\triangleq\begin{bmatrix}(\Phi_{t_1}^x)^T & \cdots & (\Phi_{t_2-1}^x)^T\end{bmatrix}^T$ with
\begin{equation}
\label{eqn:def of Phi_i,t1:t2}
\Phi_k^x\triangleq\begin{bmatrix}s_1[k]\sum_{j\in\bar{\mathcal{N}}_1}a_{1j}x_j[k] & -x_1[k]\\ \vdots & \vdots\\ s_n[k]\sum_{j\in\bar{\mathcal{N}}_n}a_{nj}x_j[k] & -x_n[k]\end{bmatrix},\ \forall k\in\{t_1,\dots,t_2-1\},
\end{equation}
and $\Phi_{t_1:t_2-1}^r\triangleq\begin{bmatrix}(\Phi_{t_1}^r)^T & \cdots & (\Phi_{t_2-1}^r)^T\end{bmatrix}^T$ with
\begin{equation}
\label{eqn:def of Psi_i,t1:t2}
\Phi_k^r\triangleq\begin{bmatrix}0 & x_1[k]\\ \vdots & \vdots\\ 0 & x_n[k]\end{bmatrix},\ \forall k\in\{t_1,\dots,t_2-1\}.
\end{equation}
We can then view Eq.~\eqref{eqn:rewrite SIR matrix} as a set of $2(t_2-t_1)n$ equations in $\beta$ and $\delta$. Noting that $s_i[k]$ for all $i\in\mathcal{V}$ can be obtained from $s_i[k]=1-x_i[k]-r_i[k]$ as argued in Section~$\ref{sec:preliminaries}$, we see that the coefficients in the set of equations in $\beta$ and $\delta$ given by Eq.~\eqref{eqn:rewrite SIR matrix}, i.e., the terms in Eq.~\eqref{eqn:rewrite SIR matrix} other than $\beta$ and $\delta$, can be determined given that $x[k]$ and $r[k]$ are known for all $k\in\{t_1,\dots,t_2\}$. Also note that given $x[k]$ and $r[k]$ for all $k\in\{t_1,\dots,t_2\}$, we can uniquely identify $\beta$ and $\delta$ using Eq.~\eqref{eqn:rewrite SIR matrix} if and only if $\rank(\begin{bmatrix}(\Phi_{t_1:t_2-1}^x)^T & (\Phi_{t_1:t_2-1}^r)^T\end{bmatrix})=2$.

Next, let $\mathcal{I}\subseteq\mathcal{I}_{t_1:t_2}$ denote a measurement selection strategy, where $\mathcal{I}_{t_1:t_2}$ is given by Eq.~\eqref{eqn:all candidate measurements}. We will then consider identifying $\beta$ and $\delta$ using measurements contained in $\mathcal{I}\subseteq\mathcal{I}_{t_1:t_2}$. To illustrate our analysis, given any $i\in\mathcal{V}$ and any $k\in\{t_1,\dots,t_2-1\}$, we first consider the following equation from Eq.~\eqref{eqn:rewrite SIR matrix}:
\begin{equation}
\label{eqn:single eq k i x}
x_i[k+1]-x_i[k]=h\begin{bmatrix}s_i[k]\sum_{w\in\bar{\mathcal{N}}_i}a_{iw}x_w[k] & -x_i[k] \end{bmatrix}\begin{bmatrix}\beta\\ \delta\end{bmatrix},
\end{equation}
where $s_i[k]=1-x_i[k]-r_i[k]$, and we index the equation in Eq.~\eqref{eqn:rewrite SIR matrix} corresponding to Eq.~\eqref{eqn:single eq k i x} as $(k,i,x)$. Note that in order to use Eq.~\eqref{eqn:single eq k i x} in identifying $\beta$ and $\delta$, one needs to determine the coefficients (i.e., the terms other than $\beta$ and $\delta$) in the equation. Also note that in order to determine the coefficients in equation $(k,i,x)$, one can use the measurements contained in $\mathcal{I}\subseteq\mathcal{I}_{t_1:t_2}$, and use Lemma~$\ref{lemma:propagate of initial condition}$ to determine if $x_i[k]=0$ (resp., $r_i[k]=0$) holds. Supposing $x_i[k+1]=0$, we see from Lemma~$\ref{lemma:propagate of initial condition}$ and Eq.~\eqref{eqn:SIR I} that $x_i[k]=0$ and $s_i[k]\sum_{w\in\bar{\mathcal{N}}_i}a_{iw}x_w[k]=0$, which makes equation $(k,i,x)$ useless in identifying $\beta$ and $\delta$. Thus, in order to use equation $(k,i,x)$ in identifying $\beta$ and $\delta$, we need $x_i[k+1]\in\mathcal{I}$ with $x_i[k+1]>0$. Similarly, given any $i\in\mathcal{V}$ and any $k\in\{t_1,\dots,t_2-1\}$, we consider the following equation from Eq.~\eqref{eqn:rewrite SIR matrix}:
\begin{equation}
r_i[k+1]-r_i[k]=h\begin{bmatrix}0 & x_i[k] \end{bmatrix}\begin{bmatrix}\beta\\ \delta\end{bmatrix},\label{eqn:eq for r i k}
\end{equation}
where we index the above equation as $(k,i,r)$. Supposing $r_i[k+1]=0$, we see from Lemma~$\ref{lemma:propagate of initial condition}$ and Eq.~\eqref{eqn:SIR R} that $r_i[k]=x_i[k]=0$, which makes equation $(k,i,r)$ useless in identifying $\beta$ and $\delta$. Hence, in order to use equation $(k,i,r)$ in identifying $\beta$ and $\delta$, we need to have $\{x_i[k],r_i[k+1]\}\subseteq\mathcal{I}$ with $x_i[k]>0$ and $r_i[k+1]>0$. More precisely, we observe that equation $(k,i,r)$ can be used in identifying $\beta$ and $\delta$ if and only if $\{x_i[k],r_i[k+1]\}\subseteq\mathcal{I}$, and $r_i[k]\in\mathcal{I}$ or $r_i[k]=0$ (from Lemma~$\ref{lemma:propagate of initial condition}$).

In general, let us denote the following two coefficient matrices corresponding to equations $(k,i,x)$ and $(k,i,r)$ in Eq.~\eqref{eqn:rewrite SIR matrix}, respectively:
\begin{subequations}
\label{eqn:def of rows}
\begin{align}
&\Phi_{k,i}^x\triangleq\begin{bmatrix}s_i[k]\sum_{j\in\bar{\mathcal{N}}_i}a_{ij}x_j[k] & -x_i[k] \end{bmatrix},\label{eqn:def of Phi x}\\
&\Phi_{k,i}^r\triangleq\begin{bmatrix}0 & x_i[k] \end{bmatrix},\label{eqn:def of Phi r}
\end{align}
\end{subequations}
for all $k\in\{t_1,\dots,t_2-1\}$ and for all $i\in\mathcal{V}$. Moreover, given any measurement selection strategy $\mathcal{I}\subseteq\mathcal{I}_{t_1:t_2}$, we let
\begin{multline}
\label{eqn:def of I bar}
\bar{\mathcal{I}}\triangleq\{(k,i,x):x_i[k+1]\in\mathcal{I},x_i[k]=0\}\cup\{(k,i,x):\{x_i[k+1],x_i[k]\}\subseteq\mathcal{I}\}\\\cup\{(k,i,r):\{r_i[k+1],x_i[k]\}\subseteq\mathcal{I},r_i[k]=0\}\cup\{(k,i,r):\{r_i[k+1],r_i[k],x_i[k]\}\subseteq\mathcal{I}\}
\end{multline}
be the set that contains indices of the equations from Eq.~\eqref{eqn:rewrite SIR matrix} that can be {\it potentially} used in identifying $\beta$ and $\delta$, based on the measurements contained in $\mathcal{I}$. In other words, the coefficients on the left-hand side of equation $(k,i,x)$ (resp., ($k,i,r$)) can be determined using the measurements from $\mathcal{I}$ and using Lemma~$\ref{lemma:propagate of initial condition}$, for all $(k,i,x)\in\bar{\mathcal{I}}$ (resp., $(k,i,r)\in\bar{\mathcal{I}}$). Let us now consider the coefficient matrix $\Phi_{k,i}^x$ (resp., $\Phi_{k,i}^r$) corresponding to $(k,i,x)\in\bar{\mathcal{I}}$ (resp., $(k,i,r)\in\bar{\mathcal{I}}$). One can then show that it is possible that there exist equations in $\bar{\mathcal{I}}$ whose coefficients cannot be (directly) determined using the measurements contained in $\mathcal{I}$ or using Lemma~$\ref{lemma:propagate of initial condition}$, where the undetermined coefficients come from the first element in $\Phi_{k,i}^x$ given by Eq.~\eqref{eqn:def of Phi x}. Nevertheless, it is also possible that one can perform algebraic operations among the equations in $\bar{\mathcal{I}}$ such that the undetermined coefficients get cancelled. Formally, we define the following.

\begin{definition}
\label{def:Phi I}
Consider a measurement selection strategy $\mathcal{I}\subseteq\mathcal{I}_{t_1:t_2}$, where $\mathcal{I}_{t_1:t_2}$ is given by Eq.~\eqref{eqn:all candidate measurements}. Stack coefficient matrices $\Phi_{k,i}^x\in\mathbb{R}^{1\times 2}$ for all $(k,i,x)\in\bar{\mathcal{I}}$ and $\Phi_{k,i}^r\in\mathbb{R}^{1\times2}$ for all $(k,i,r)\in\bar{\mathcal{I}}$ into a single matrix, where $\Phi_{k,i}^x$ and $\Phi_{k,i}^r$ are given by \eqref{eqn:def of rows} and $\bar{\mathcal{I}}$ is given by Eq.~\eqref{eqn:def of I bar}. The resulting matrix is denoted as $\Phi(\mathcal{I})\in\mathbb{R}^{|\bar{\mathcal{I}}|\times2}$. Moreover, define $\tilde{\Phi}(\mathcal{I})$ to be the set that contains all the matrices $\Phi\in\mathbb{R}^{2\times2}$ such that $(\Phi)_1$ and $(\Phi)_2$ can be obtained via algebraic operations among the rows in $\Phi(\mathcal{I})$, and the elements in $(\Phi)_1$ and $(\Phi)_2$ can be fully determined using the measurements from $\mathcal{I}\subseteq\mathcal{I}_{t_1:t_2}$ and using Lemma~$\ref{lemma:propagate of initial condition}$.
\end{definition}

In other words, $\Phi\in\tilde{\Phi}(\mathcal{I})$ corresponds to two equations (in $\beta$ and $\delta$) obtained from Eq.~\eqref{eqn:rewrite SIR matrix} such that the coefficients on the right-hand side of the two equations can be determined using the measurements contained in $\mathcal{I}$ and using Lemma~$\ref{lemma:propagate of initial condition}$ (if the coefficients contain $x_i[k]=0$ or $r_i[k]=0$). Moreover, one can show that the coefficients on the left-hand side of the two equations obtained from Eq.~\eqref{eqn:rewrite SIR matrix} corresponding to $\Phi$ can also be determined using measurements from $\mathcal{I}$ and using Lemma~$\ref{lemma:propagate of initial condition}$. Putting the above arguments together, we see that given a measurement selection strategy $\mathcal{I}\subseteq\mathcal{I}_{t_1:t_2}$,  $\beta$ and $\delta$ can be uniquely identified if and only if there exists $\Phi\in\tilde{\Phi}(\mathcal{I})$ such that $\rank(\Phi)=2$. Equivalently, denoting
\begin{equation}
\label{eqn:max rank of Phi}
r_{\mathop{\max}}(\mathcal{I})\triangleq\mathop{\max}_{\Phi\in\tilde{\Phi}(\mathcal{I})}\rank(\Phi),
\end{equation}
where $r_{\mathop{\max}}(\mathcal{I})\triangleq0$ if $\tilde{\Phi}(\mathcal{I})=\emptyset$, we see that $\beta$ and $\delta$ can be uniquely identified using the measurements from $\mathcal{I}\subseteq\mathcal{I}_{t_1:t_2}$ if and only if $r_{\mathop{\max}}(\mathcal{I})=2$.

\begin{remark}
\label{remark:Phi I at least two rows}
Note that if a measurement selection strategy $\mathcal{I}\subseteq\mathcal{I}_{t_1:t_2}$ satisfies $r_{\mathop{\max}}(\mathcal{I})=2$, it follows from the above arguments that $|\bar{\mathcal{I}}|\ge2$, i.e., $\Phi(\mathcal{I})\in\mathbb{R}^{|\bar{\mathcal{I}}|\times2}$ has at least two rows, where $\bar{\mathcal{I}}$ is defined in Eq.~\eqref{eqn:def of I bar}.
\end{remark}

Recall that collecting the measurement of $x_i[k]$ (resp., $r_i[k]$) at any node $i\in\mathcal{V}$ and at any time step $k\in\mathbb{Z}_{\ge0}$ incurs cost $c_{k,i}\in\mathbb{R}_{\ge0}$ (resp., $b_{k,i}\in\mathbb{R}_{\ge0}$). Given any measurement selection strategy $\mathcal{I}\subseteq\mathcal{I}_{t_1:t_2}$, we denote the cost associated with $\mathcal{I}$ as
\begin{equation}
\label{eqn:cost of I}
c(\mathcal{I})\triangleq\sum_{x_i[k]\in\mathcal{I}}c_{k,i}+\sum_{r_i[k]\in\mathcal{I}}b_{k,i}.
\end{equation}
We then define the Parameter Identification Measurement Selection (PIMS) problem in the perfect measurement setting as follows.
\begin{problem}
\label{pro:PIMS}
Consider a discrete-time SIR model given by Eq.~\eqref{eqn:SIR single node} with a directed graph $\mathcal{G}=\{\mathcal{V},\mathcal{E}\}$, a weight matrix $A\in\mathbb{R}^{n\times n}$, a sampling parameter $h\in\mathbb{R}_{\ge0}$, and sets $\mathcal{S}_I,\mathcal{S}_H\subseteq\mathcal{V}$ defined in Definition~$\ref{def:sets of initial states}$. Moreover, consider time steps $t_1,t_2\in\mathbb{Z}_{\ge0}$ with $t_1< t_2$, and a cost $c_{k,i}\in\mathbb{R}_{\ge0}$ of measuring $x_i[k]$ and a cost $b_{k,i}\in\mathbb{R}_{\ge0}$ of measuring $r_i[k]$ for all $i\in\mathcal{V}$ and for all $k\in\{t_1,\dots,t_2\}$. The PIMS problem is to find $\mathcal{I}\subseteq\mathcal{I}_{t_1:t_2}$ that solves
\begin{equation}
\label{eqn:PIMS obj}
\begin{split}
&\mathop{\min}_{\mathcal{I}\subseteq\mathcal{I}_{t_1:t_2}}c(\mathcal{I})\\
s.t.\ &r_{\mathop{\max}}(\mathcal{I})=2,
\end{split}
\end{equation}
where $\mathcal{I}_{t_1:t_2}$ is defined in Eq.~\eqref{eqn:all candidate measurements}, $c(\mathcal{I})$ is defined in Eq.~\eqref{eqn:cost of I}, and $r_{\mathop{\max}}(\mathcal{I})$ is defined in Eq.~\eqref{eqn:max rank of Phi}.
\end{problem}

We  have the following result; the proof is included in Section~$\ref{sec:proof of PIMS NP-hard}$ in the Appendix.
\begin{theorem}
\label{thm:PIMS NP-hard}
The PIMS problem is NP-hard.
\end{theorem}

Theorem~$\ref{thm:PIMS NP-hard}$ indicates that there is no polynomial-time algorithm that solves all instances of the PIMS problem optimally (if P $\neq$ NP). Moreover, we note from the formulation of the PIMS problem given by Problem~$\ref{pro:PIMS}$ that for a measurement selection strategy $\mathcal{I}\subseteq\mathcal{I}_{t_1:t_2}$, one needs to check if $\mathop{\max}_{\Phi\in\tilde{\Phi}(\mathcal{I})}\rank(\Phi)=2$ holds, {\it before} the corresponding measurements are collected. However, in general, it is not possible to calculate $\rank(\Phi)$ when no measurements are collected. In order to bypass these issues, we will explore additional properties of the PIMS problem in the following.

\subsection{Solving the PIMS Problem}\label{sec:solve PIMS}
We start with the following result.
\begin{lemma}
\label{lemma:rank cond noiseless measurement}
Consider a discrete time SIR model given by Eq.~\eqref{eqn:SIR single node}. Suppose Assumptions~$\ref{ass:initial condition}$-$\ref{ass:model parameters}$ hold. Then, the following results hold, where $\Phi_{k_1,i_1}^x\in\mathbb{R}^{1\times2}$ and $\Phi_{k_2,i_2}^r\in\mathbb{R}^{1\times2}$ are defined in \eqref{eqn:def of rows}, $\mathcal{S}_I^{\prime}\triangleq\{i\in\mathcal{S}_I: a_{ii}>0\}$, $\mathcal{S}^{\prime}\triangleq\{i\in\mathcal{V}\setminus\mathcal{S}_I^{\prime}:\mathcal{N}_i\neq\emptyset,\mathop{\min}\{d_j:j\in\mathcal{N}_i\}\neq\infty\}$, and $\mathcal{S}_I$ and $d_i$ are defined in Definition~$\ref{def:sets of initial states}$ for all $i\in\mathcal{V}$.\\
$(a)$ For any $i_1\in\mathcal{S}_I^{\prime}$ and for any $i_2\in\mathcal{V}$ with $d_{i_2}\neq\infty$, $\text{rank}\big(\begin{bmatrix}(\Phi_{k_1,i_1}^x)^T & (\Phi_{k_2,i_2}^r)^T\end{bmatrix}\big)=2$ for all $k_1\ge0$ and for all $k_2\ge d_{i_2}$, where $k_1,k_2\in\mathbb{Z}_{\ge0}$.\\
$(b)$ For any $i_1\in\mathcal{S}^{\prime}$ and for any $i_2\in\mathcal{V}$ with $d_{i_2}\neq\infty$, $\text{rank}\big(\begin{bmatrix}(\Phi_{k_1,i_1}^x)^T & (\Phi_{k_2,i_2}^r)^T\end{bmatrix}\big)=2$ for all $k_1\ge\mathop{\min}\{d_j:j\in\mathcal{N}_{i_1}\}$, and for all $k_2\ge d_{i_2}$, where $k_1,k_2\in\mathbb{Z}_{\ge0}$. \end{lemma}
\begin{proof}
Noting from \eqref{eqn:def of rows}, we have
\begin{equation}
\label{eqn:sufficient rank cond}
\begin{bmatrix}\Phi_{k_1,i_1}^x\\ \Phi_{k_2,i_2}^r\end{bmatrix}=\begin{bmatrix}s_{i_1}[k_1]\sum_{j\in\bar{\mathcal{N}}_{i_1}}a_{i_1j}x_j[k_1] & -x_{i_1}[k_1]\\ 0 & x_{i_2}[k_2] \end{bmatrix}.
\end{equation}
To prove part~$(a)$, consider any $i_1\in\mathcal{S}_I^{\prime}$ and any $i_2\in\mathcal{V}$ with $d_{i_2}\neq\infty$, where we note $x_{i_1}[0]>0$ and $a_{i_1i_1}>0$ from the definition of $\mathcal{S}_I^{\prime}$. We then see from Lemma~$\ref{lemma:propagate of initial condition}(a)$-$(b)$ that $s_{i_1}[k_1]>0$ and $x_{i_1}[k_1]>0$ for all $k_1\ge0$. It follows that $s_{i_1}[k_1]\sum_{j\in\bar{\mathcal{N}}_{i_1}}a_{i_1j}x_j[k_1]>0$ for all $k_1\ge0$. Also, we obtain from Lemma~$\ref{lemma:propagate of initial condition}(b)$ $x_{i_2}[k_2]>0$ for all $k_2\ge d_{i_2}$, which proves part~$(a)$. 

We then prove part~$(b)$. Considering any $i_1\in\mathcal{S}^{\prime}$ and any $i_2\in\mathcal{V}$ with $d_2\neq\infty$, we see from the definition of $\mathcal{S}^{\prime}$ that $\mathcal{N}_{i_1}\neq\emptyset$ and there exists $j\in\mathcal{N}_{i_1}$ such that $d_j\neq\infty$. Letting $j_1$ be a node in $\mathcal{N}_{i_1}$ such that $d_{j_1}=\mathop{\min}\{d_j:j\in\mathcal{N}_{i_1}\}\neq\infty$, we note from Lemma~$\ref{lemma:propagate of initial condition}(a)$ that $x_{j_1}[k_1]>0$ for all $k_1\ge\mathop{\min}\{d_j:j\in\mathcal{N}_{i_1}\}$. Also note that $a_{i_1j_1}>0$ from Assumption~$\ref{ass:model parameters}$. The rest of the proof of part~$(b)$ is then identical to that of part~$(a)$.
\end{proof}


Recalling the way we index the equations in Eq.~\eqref{eqn:rewrite SIR matrix} (see Eqs.~\eqref{eqn:single eq k i x}-\eqref{eqn:eq for r i k} for examples), we define the set that contains all the indices of the equations in Eq.~\eqref{eqn:rewrite SIR matrix}:
\begin{equation}
\label{eqn:def of Q}
\mathcal{Q}\triangleq\{(k,i,\lambda):k\in\{t_1,\dots,t_2-1\},i\in\mathcal{V},\lambda\in\{x,r\}\}.
\end{equation}
Following the arguments in Lemma~$\ref{lemma:rank cond noiseless measurement}$, we denote
\begin{align}
&\mathcal{Q}_1\triangleq\{(k,i,x)\in\mathcal{Q}:i\in\mathcal{S}_I^{\prime}\}\cup\{(k,i,x)\in\mathcal{Q}:k\ge\mathop{\min}\{d_j:j\in\mathcal{N}_i\},i\in\mathcal{S}^{\prime}\},\label{eqn:def of Q_1}\\
&\mathcal{Q}_2\triangleq\{(k,i,r)\in\mathcal{Q}:k\ge d_i,i\in\mathcal{V},d_i\neq\infty\},\label{eqn:def of Q_2}
\end{align}
where $\mathcal{S}_I^{\prime}$ and $\mathcal{S}^{\prime}$ are defined in Lemma~$\ref{lemma:rank cond noiseless measurement}$, and $d_i$ is defined in Definition~$\ref{def:sets of initial states}$. Next, for all $(k,i,x)\in\mathcal{Q}$, we define the set of measurements that are needed to determine the coefficients in equation $(k,i,x)$ (when no other equations are used) to be
\begin{equation*}
\label{eqn:set of measurements for x}
\mathcal{I}(k,i,x)\triangleq\big(\{x_i[k+1],r_i[k]\}\cup\{x_j[k]:j\in\bar{\mathcal{N}}_i\}\big)\cap\mathcal{I}_{t_1:t_2},
\end{equation*}
where $\mathcal{I}_{t_1:t_2}$ is defined in Eq.~\eqref{eqn:all candidate measurements}. Similarly, for all $(k,i,r)\in\mathcal{Q}$, we define
\begin{equation*}
\label{eqn:set of measurements for r}
\mathcal{I}(k,i,r)\triangleq\big(\{r_i[k+1],r_i[k],x_i[k]\}\big)\cap\mathcal{I}_{t_1:t_2}.
\end{equation*}
Moreover, let us denote
\begin{equation}
\label{eqn:measurements of two eqs}
\mathcal{I}((k_1,i_1,\lambda_1),(k_2,i_2,\lambda_2))\triangleq\mathcal{I}(k_1,i_1,\lambda_1)\cup\mathcal{I}(k_2,i_2,\lambda_2)
\end{equation}
for all $(k_1,i_1,\lambda_1),(k_2,i_2,\lambda_2)\in\mathcal{Q}$. Similarly to Eq.~\eqref{eqn:cost of I}, denote the sum of the costs of the measurements from $\mathcal{I}((k_1,i_1,\lambda_1),(k_2,i_2,\lambda_2))$ as $c(\mathcal{I}((k_1,i_1,\lambda_1),(k_2,i_2,\lambda_2)))$.

\begin{algorithm}
\caption{Algorithm for PIMS}
\label{alg:perfect measurement}
\begin{algorithmic}[1]
\State{\textbf{Input}: An instance of PIMS}
\State{Find $(k_1,i_1,x)\in\mathcal{Q}_1$, $(k_2,i_2,r)\in\mathcal{Q}_2$ s.t. $c(\mathcal{I}((k_1,i_1,x),(k_2,i_2,r)))$ is minimized}
\Return{$\mathcal{I}((k_1,i_1,x),(k_2,i_2,r))$}
\end{algorithmic}
\end{algorithm}

Based on the above arguments, we propose an algorithm defined in Algorithm~$\ref{alg:perfect measurement}$ for the PIMS problem. Note that Algorithm~$\ref{alg:perfect measurement}$ finds an equation from $\mathcal{Q}_1$ and an equation from $\mathcal{Q}_2$ such that the sum of the costs of the two equations is minimized, where $\mathcal{Q}_1$ and $\mathcal{Q}_2$ are defined in Eq.~\eqref{eqn:def of Q_1} and Eq.~\eqref{eqn:def of Q_2}, respectively.
\begin{proposition}
\label{thm:performance of algorithm 1}
Consider an instance of the PIMS problem under Assumptions $\ref{ass:initial condition}$-$\ref{ass:model parameters}$. Algorithm~$\ref{alg:perfect measurement}$ returns a solution $\mathcal{I}((k_1,i_1,x),(k_2,i_2,r))$ to the PIMS problem that satisfies the constraint in \eqref{eqn:PIMS obj}, and the following:
\begin{equation}
\label{eqn:bound 2 for alg 1}
\frac{c(\mathcal{I}((k_1,i_1,x),(k_2,i_2,r)))}{c(\mathcal{I}^{\star})}\le\frac{\mathop{\min}_{(k,i,x)\in\mathcal{Q}_1}(b_{k+1,i}+b_{k,i}+c_{k+1,i}+\sum_{j\in\bar{\mathcal{N}}_i}c_{k,j})}{3c_{\mathop{\min}}},
\end{equation}
where $\mathcal{I}^{\star}$ is an optimal solution to the PIMS problem, $\mathcal{Q}_1$ is defined in Eq.~\eqref{eqn:def of Q_1}, and $c_{\mathop{\min}}\triangleq\mathop{\min}\{\mathop{\min}_{x_i[k]\in\mathcal{I}_{t_1:t_2}}c_{k,i},\mathop{\min}_{r_i[k]\in\mathcal{I}_{t_1:t_2}}b_{k,i}\}>0$ with $\mathcal{I}_{t_1:t_2}$ given by Eq.~\eqref{eqn:all candidate measurements}.
\end{proposition}
\begin{proof}
The feasibility of $\mathcal{I}((k_1,i_1,x),(k_2,i_2,r))$ follows directly from the definition of Algorithm~$\ref{alg:perfect measurement}$ and Lemma~$\ref{lemma:rank cond noiseless measurement}$. We now prove \eqref{eqn:bound 2 for alg 1}.  Consider any equations $(k,i,x)\in\mathcal{Q}_1$ and $(k,i,r)\in\mathcal{Q}_2$. We have from Eq.~\eqref{eqn:measurements of two eqs} the following:
\begin{equation*}
\mathcal{I}((k,i,x),(k,i,r))=\big(\{x_i[k+1],r_i[k]\}\cup\{x_j[k]:j\in\bar{\mathcal{N}}_i\}\cup\{r_i[k+1],r_i[k],x_i[k]\}\big)\cap\mathcal{I}_{t_1:t_2},
\end{equation*}
which implies 
\begin{equation*}
c(\mathcal{I}((k_1,i_1,x),(k_2,i_2,r)))\le\mathop{\min}_{(k,i,x)\in\mathcal{Q}_1}(b_{k+1,i}+b_{k,i}+c_{k+1,i}+\sum_{j\in\bar{\mathcal{N}}_i}c_{k,j}).
\end{equation*}
Next, since $\mathcal{I}^{\star}$ satisfies $r_{\mathop{\max}}(\mathcal{I}^{\star})=2$, we recall from Remark~$\ref{remark:Phi I at least two rows}$ $|\bar{\mathcal{I}}^{\star}|\ge2$, where
\begin{multline*}
\label{eqn:I bar star}
\bar{\mathcal{I}}^{\star}=\{(k,i,x):x_i[k+1]\in\mathcal{I}^{\star},x_i[k]=0\}\cup\{(k,i,x):\{x_i[k+1],x_i[k]\}\subseteq\mathcal{I}^{\star}\}\\\cup\{(k,i,r):\{r_i[k+1],x_i[k]\}\subseteq\mathcal{I}^{\star},r_i[k]=0\}\cup\{(k,i,r):\{r_i[k+1],r_i[k],x_i[k]\}\subseteq\mathcal{I}^{\star}\},
\end{multline*}
which implies $|\mathcal{I}^{\star}|\ge2$. In fact, suppose $\mathcal{I}^{\star}=\{x_i[k+1],x_j[k+1]\}$, where $i\in\mathcal{V}$ and $k\in\{t_1-1,\dots,t_2-1\}$. Since the elements in $\Phi_{k,i}^x$ and $\Phi_{k,j}^x$ (defined in \eqref{eqn:def of rows}) do not contain $x_w[0]$, $r_w[0]$ or $s_w[0]$ for any $w\in\mathcal{V}$, and cannot all be zero, we see that there exists $x_{w^{\prime}}[k]\in\mathcal{I}^{\star}$ (with $x_{w^{\prime}}[k]>0$), where $w^{\prime}\in\mathcal{V}$. This further implies $|\mathcal{I}^{\star}|\ge3$. Using similar arguments, one can show that $|\mathcal{I}^{\star}|\ge3$ holds in general, which implies $c(\mathcal{I}^{\star})\ge3c_{\mathop{\min}}$. Combining the above arguments leads to \eqref{eqn:bound 2 for alg 1}.
\end{proof}

Finally, note that $\mathcal{Q}_2$ and $\mathcal{I}_{t_1:t_2}$ can be obtained by calling the Breadth-First-Search (BFS) algorithm (e.g., \cite{cormen2009introduction}) $|\mathcal{S}_I|$ times with $O(|\mathcal{S}_I|(n+|\mathcal{E}|))$ total time complexity. Also note that the time complexity of line $2$ in Algorithm~$\ref{alg:perfect measurement}$ is $O(n^2(t_2-t_1+1)^2)$. Thus, the overall time complexity of Algorithm~$\ref{alg:perfect measurement}$ is $O(n^2(t_2-t_1+1)^2+|\mathcal{S}_I||\mathcal{E}|)=O(|\mathcal{Q}|^2+|\mathcal{S}_I||\mathcal{E}|)$.

\section{Measurement Selection Problem in Random Measurement Setting}\label{sec:noisy measurement}
In this section, we assume that the initial condition $l\triangleq[(s[0])^T\ (x[0])^T\ (r[0])^T]^T$ is known. Nevertheless, our analysis can potentially be extended to cases where the initial condition $l$ is given by a probability distribution.

\subsection{Problem Formulation}
\label{sec:PEMS problem formulation}
Here, we consider the scenario where the measurement of $x_i[k]$ (resp., $r_i[k]$), denoted as $\hat{x}_i[k]$ (resp., $\hat{r}_i[k]$), is given by a pmf $p(\hat{x}_i[k]|x_i[k])$ (resp., $p(\hat{r}_i[k]|r_i[k])$). Note one can express $x_i[k]$ in terms of $l$ and $\theta\triangleq[\beta\ \delta]^T$ using \eqref{eqn:SIR I}. Hence, given $l$ and $\theta$, we can alternatively write $p(\hat{x}_i[k]|x_i[k])$ as $p(\hat{x}_i[k]|l,\theta)$ for all $i\in\mathcal{V}$ and for all $k\in\mathbb{Z}_{\ge1}$. Since the initial conditions are assumed to be known, we drop the dependency of $p(\hat{x}_i[k]|l,\theta)$ on $l$, and denote the pmf of $\hat{x}_i[k]$ as $p(\hat{x}_i[k]|\theta)$ for all $i\in\mathcal{V}$ and for all $k\in\mathbb{Z}_{\ge1}$. Similarly, given $l$ and $\theta$, we denote the pmf of $\hat{r}_i[k]$ as  $p(\hat{r}_i[k]|\theta)$ for all $i\in\mathcal{V}$ and for all $k\in\mathbb{Z}_{\ge1}$. Note that when collecting measurement $\hat{x}_i[k]$ (resp., $\hat{r}_i[k]$) under a limited budget,  one possibility is to give virus (resp., antibody) tests to a group of randomly and uniformly sampled individuals of the population at node $i\in\mathcal{V}$ and at time $k\in\mathbb{Z}_{\ge1}$ (e.g., \cite{protect}), where a positive testing result indicates that the tested individual is infected (resp., recovered) at time $k$  (e.g., \cite{cdcviraltest}). Thus, the obtained random measurements $\hat{x}_i[k]$ and $\hat{r}_i[k]$ and the corresponding pmfs $p(\hat{x}_i[k]|\theta)$ and $p(\hat{r}_i[k]|\theta)$ depend on the total number of conducted virus tests and antibody tests at node $i$ and at time $k$, respectively.  Consider any node $i\in\mathcal{V}$ and any time step $k\in\mathbb{Z}_{\ge1}$, where the total population of $i$ is denoted by $N_i\in\mathbb{Z}_{\ge1}$ and is assumed to be fixed over time. Suppose we are allowed to choose the number of virus (resp., antibody) tests that will be performed on the (randomly sampled) individuals at node $i$ and at time $k$. Assume that the cost of performing the virus (resp., antibody) tests is proportional to the number of the tests. For any $i\in\mathcal{V}$ and for any $k\in\{t_1,\dots,t_2\}$, let 
\begin{equation}
\label{eqn:set C_k,i}
\mathcal{C}_{k,i}\triangleq\{\zeta c_{k,i}:\zeta\in(\{0\}\cup[\zeta_i])\}
\end{equation}
be the set of all possible costs that we can spend on collecting the measurement $\hat{x}_i[k]$, where $c_{k,i}\in\mathbb{R}_{\ge0}$ and $\zeta_i\in\mathbb{Z}_{\ge1}$. Similarly, for any $i\in\mathcal{V}$ and any $k\in\{t_1,\dots,t_2\}$, let
\begin{equation}
\label{eqn:set B_k,i}
\mathcal{B}_{k,i}\triangleq\{\eta b_{k,i}:\eta\in(\{0\}\cup[\eta_i])\}
\end{equation}
denote the set of all possible costs that we can spend on collecting the measurement $\hat{r}_i[k]$, where $b_{k,i}\in\mathbb{R}_{\ge0}$ and $\eta_i\in\mathbb{Z}_{\ge1}$. For instance, $\zeta c_{k,i}$ can be viewed as the cost of performing virus tests on $\zeta N_i^x$ (randomly sampled) individuals in the population at node $i$, where $N_i^x\in\mathbb{Z}_{\ge1}$ and $\zeta_i N_i^x\le N_i$. To reflect the dependency of the pmf $p(\hat{x}_i[k]|\theta)$ (resp., $p(\hat{r}_i[k]|\theta)$) of measurement $\hat{x}_i[k]$ (resp., $\hat{r}_i[k]$) on the cost spent on collecting the measurement of $x_i[k]$ (resp., $r_i[k]$),  we further denote the pmf of $\hat{x}_i[k]$ (resp., $\hat{r}_i[k]$) as $p(\hat{x}_i[k]|\theta,\varphi_{k,i})$ (resp., $p(\hat{r}_i[k]|\theta,\omega_{k,i})$), where $\varphi_{k,i}\in \mathcal{C}_{k,i}$ (resp., $\omega_{k,i}\in \mathcal{B}_{k,i}$) with $\mathcal{C}_{k.i}$ (resp., $\mathcal{B}_{k,i}$) given by Eq.~\eqref{eqn:set C_k,i} (resp., Eq.~\eqref{eqn:set B_k,i}). Note that $\varphi_{k,i}$ (resp., $\omega_{k,i}$) is the cost that we spend on collecting measurement $\hat{x}_i[k]$ (resp., $\hat{r}_i[k]$), and $\varphi_{k,i}=0$ (resp., $\omega_{k,i}=0$) indicates that measurement $\hat{x}_i[k]$ (resp., $\hat{r}_i[k]$) is not collected.

In contrast with the exact measurement case studied in Section $\ref{sec:perfect measurement}$, it is not possible to uniquely identify $\beta$ and $\delta$ using measurements $\hat{x}_i[k]$ and $\hat{r}_i[k]$ which are now random variables. Thus, we will consider estimators of $\beta$ and $\delta$ based on the measurements indicated by a measurement selection strategy. Similarly to Section~$\ref{sec:perfect measurement}$, given time steps $t_1,t_2\in\mathbb{Z}_{\ge1}$ with $t_2\ge t_1$, define the set of all candidate measurements as 
\begin{equation}
\label{eqn:def of U t_1 t_2}
\mathcal{U}_{t_1:t_2}\triangleq\{\hat{x}_i[k]:i\in\mathcal{V},k\in\{t_1,\dots,t_2\}\}\cup\{\hat{r}_i[k]:i\in\mathcal{V},k\in\{t_1,\dots,t_2\}\}.
\end{equation}
Recalling $\mathcal{C}_{k,i}$ and $\mathcal{B}_{k,i}$ defined in Eq.~\eqref{eqn:set C_k,i} and Eq.~\eqref{eqn:set B_k,i}, respectively, we let
$\mu\in\mathbb{Z}_{\ge0}^{\mathcal{U}_{t_1:t_2}}$ be a measurement selection that specifies the costs spent on collecting measurements $\hat{x}_i[k]$ and $\hat{r}_i[k]$ for all $i\in\mathcal{V}$ and for all $k\in\{t_1,\dots,t_2\}$. Moreover, we define the set of all candidate measurement selections as
\begin{equation}
\label{eqn:def of M}
\mathcal{M}\triangleq\{\mu\in\mathbb{Z}_{\ge0}^{\mathcal{U}_{t_1:t_2}}:\mu(\hat{x}_i[k])\in(\{0\}\cup[\zeta_i]),\mu(\hat{r}_i[k])\in(\{0\}\cup[\eta_i])\},
\end{equation}
where $\zeta_i,\eta_i\in\mathbb{Z}_{\ge1}$ for all $i\in\mathcal{V}$.
In other words, a measurement selection $\mu$ is defined over the integer lattice $\mathbb{Z}_{\ge0}^{\mathcal{U}_{t_1:t_2}}$ so that $\mu$ is a vector of dimension $|\mathcal{U}_{t_1:t_2}|$, where each element of $\mu$ corresponds to an element in $\mathcal{U}_{t_1:t_2}$, and is denoted as $\mu(\hat{x}_i[k])$ (or $\mu(\hat{r}_i[k])$). The set $\mathcal{M}$ contains all $\mu\in\mathbb{Z}_{\ge0}^{\mathcal{U}_{t_1:t_2}}$ such that $\mu(\hat{x}_i[k])\in(\{0\}\cup[\zeta_i])$ and $\mu(\hat{r}_i[k])\in(\{0\}\cup[\eta_i])$ for all $i\in\mathcal{V}$ and for all $k\in\{t_1,\dots,t_2\}$. Thus, for any $\varphi_{k,i}\in \mathcal{C}_{k,i}$ and for any $\omega_{k,i}\in \mathcal{B}_{k,i}$, there exists $\mu\in\mathcal{M}_{\ge0}^{\mathcal{U}_{t_1:t_2}}$ such that $\mu(\hat{x}_i[k])c_{k,i}=\varphi_{k,i}$ and $\mu(\hat{r}_i[k])b_{k,i}=\omega_{k,i}$. In other words, $\mu(\hat{x}_i[k])c_{k,i}$ (resp., $\mu(\hat{r}_i[k])b_{k,i}$) indicates the cost spent on collecting the measurement of $x_i[k]$ (resp., $r_i[k]$). Given a measurement selection $\mu\in\mathbb{Z}_{\ge0}^{t_1:t_2}$, we can also denote the pmfs of $\hat{x}_i[k]$ and $\hat{r}_i[k]$ as $p(\hat{x}_i[k]|\theta,\mu(\hat{x}_i[k]))$ and $p(\hat{r}_i[k]|\theta,\mu(\hat{r}_i[k]))$, respectively, where we drop the dependencies of the pmfs on $c_{k,i}$ and $b_{k,i}$ for notational simplicity.

To proceed, we consider the scenario where measurements can only be collected under a budget constraint given by $B\in\mathbb{R}_{\ge0}$. Using the above notations, the budget constraint can be expressed as
\begin{equation}
\label{eqn:budget constraint PEMS}
\sum_{\hat{x}_i[k]\in\mathcal{U}_{t_1:t_2}}c_{k,i}\mu(\hat{x}_i[k])+\sum_{\hat{r}_i[k]\in\mathcal{U}_{t_1:t_2}}b_{k,i}\mu(\hat{r}_i[k])\le B.
\end{equation}
We then consider estimators of $\theta=[\beta\ \delta]^T$ based on any given measurement selection $\mu\in\mathcal{M}$. Considering any $\mu\in\mathcal{M}$, we denote 
\begin{equation}
\label{eqn:def of Uxi Uri}
\mathcal{U}^{\lambda}_i\triangleq\{k:\mu(\hat{\lambda}_i[k])>0,k\in\{t_1,\dots,t_2\}\},
\end{equation}
for all $i\in\mathcal{V}$ and for all $\lambda\in\{x,r\}$. For all $i\in\mathcal{V}$ and for all $\lambda\in\{x,r\}$ with $\mathcal{U}_i^{\lambda}\neq\emptyset$, denote $y(\mathcal{U}_i^{\lambda})\triangleq\begin{bmatrix}\hat{\lambda}_i[k_1] & \cdots \hat{\lambda}_i[k_{|\mathcal{U}_i^{\lambda}|}]\end{bmatrix}^T$, where $\mathcal{U}_i^{\lambda}=\{k_1,\dots,k_{|\mathcal{U}_i^{\lambda}|}\}$. Letting 
\begin{equation*}
\label{eqn:def of Ux Ur}
\mathcal{U}_{\lambda}\triangleq\{i:\mathcal{U}_i^{\lambda}\neq\emptyset,i\in\mathcal{V}\}\ \forall\lambda\in\{x,r\},
\end{equation*}
we denote the measurement vector indicated by $\mu\in\mathcal{M}$ as
\begin{equation}
\label{eqn:def of y mu}
y(\mu)\triangleq\begin{bmatrix}(y(\mathcal{U}_{i_1}^x))^T & \cdots & (y(\mathcal{U}_{i_{|\mathcal{U}_x|}}^x))^T & (y(\mathcal{U}_{j_1}^r))^T & \cdots & (y(\mathcal{U}_{j_{|\mathcal{U}_r|}}^r))^T \end{bmatrix}^T,
\end{equation}
where $\mathcal{U}_x=\{i_1,\dots,i_{|\mathcal{U}_x|}\}$ and $\mathcal{U}_r=\{j_1,\dots,j_{|\mathcal{U}_r|}\}$. Note that $\hat{x}_i[k]$ and $\hat{r}_i[k]$ are (discrete) random variables with pmfs $p(\hat{x}_i[k]|\theta,\mu(\hat{x}_i[k]))$ and $p(\hat{r}_i[k]|\theta,\mu(\hat{r}_i[k]))$, respectively. We then see from Eq.~\eqref{eqn:def of y mu} that $y(\mu)$ is a random vector whose pmf is denoted as $p(y(\mu)|\theta,\mu)$. Similarly, the pmf of $y(\mathcal{U}_i^x)$ (resp., $y(\mathcal{U}_i^r)$) is denoted as $p(y(\mathcal{U}_i^x)|\theta,\mu)$ (resp., $p(y(\mathcal{U}_i^r)|\theta,\mu)$). Given $t_1,t_2\in\mathbb{Z}_{\ge1}$ with $t_2\ge t_1$, we make the following assumption on measurements $\hat{x}_i[k]$ and $\hat{r}_i[k]$.

\begin{assumption}
\label{ass:white and independent noise}
For any $i\in\mathcal{V}$ and for any $k_1,k_2\in\{t_1,\dots,t_2\}$ ($k_1\neq k_2$), $\hat{x}_i[k_1]$, $\hat{x}_i[k_2]$, $\hat{r}_i[k_1]$ and $\hat{r}_i[k_2]$ are independent of each other. Moreover, for any $i,j\in\mathcal{V}$ ($i\neq j$) and for any $k_1,k_2\in\{t_1,\dots,t_2\}$, $\hat{x}_i[k_1]$ and $\hat{x}_j[k_2]$ are independent, and $\hat{x}_i[k_1]$ and $\hat{r}_j[k_2]$ are independent.
\end{assumption}

The above assumption ensures that measurements from different nodes or from different time steps are independent, and the measurements of $x_i[k]$ and $r_i[k]$ are also independent. It then follows from Eq.~\eqref{eqn:def of y mu} that the pmf of $y(\mu)$ can be written as 
\begin{equation}
\label{eqn:pmf of y mu}
p(y(\mu)|\theta,\mu)=\prod_{i\in\mathcal{U}_x}p(y(\mathcal{U}_i^x)|\theta,\mu)\cdot\prod_{j\in\mathcal{U}_r}p(y(\mathcal{U}_j^r)|\theta,\mu),
\end{equation}
where we can further write $p(y(\mathcal{U}_i^x)|\theta,\mu)=\prod_{k\in\mathcal{U}_i^x}p(\hat{x}_i[k]|\theta,\mu(\hat{x}_i[k]))$ for all $i\in\mathcal{U}_x$, and $p(y(\mathcal{U}_j^r)|\theta,\mu)=\prod_{k\in\mathcal{U}_j^r}p(\hat{r}_j[k]|\theta,\mu(\hat{r}_j[k]))$ for all $j\in\mathcal{U}_r$.

In order to quantify the performance (e.g., precision) of estimators of $\theta$ based on $\mu$, we use the Bayesian Cramer-Rao Lower Bound (BCRLB) (e.g., \cite{van2004optimum}) associated with $\mu$.  In the following, we introduce the BCRLB, and explain why we choose it as a performance metric. First, given any measurement $\mu\in\mathcal{M}$, let $F_{\theta}(\mu)$ be the corresponding Fisher information matrix defined as
\begin{equation}
\label{eqn:FIM}
F_{\theta}(\mu)\triangleq-\mathbb{E}\begin{bmatrix}\frac{\partial^2 \ln p(y(\mu)|\theta,\mu)}{{\partial \beta^2}} & \frac{\partial^2 \ln p(y(\mu)|\theta,\mu)}{{\partial\beta}{\partial\delta}}\\ \frac{\partial^2 \ln p(y(\mu)|\theta,\mu)}{{\partial\delta}{\partial\beta}} & \frac{\partial^2 \ln p(y(\mu)|\theta,\mu)}{{\partial\delta^2}}\end{bmatrix}
\end{equation}
with the expectation $\mathbb{E}[\cdot]$ taken with respect to $p(y(\mu)|\theta,\mu)$. Under Assumption~$\ref{ass:white and independent noise}$ and some regularity conditions on the pmfs of $\hat{x}_i[k]$ and $\hat{r}_i[k]$, Eq.~\eqref{eqn:FIM} can be written as the following (e.g., \cite{kay1993fundamentals}):
\begin{equation}
\label{eqn:FIM 2nd exp}
F_{\theta}(\mu)=\sum_{\lambda\in\{x,r\}}\sum_{i\in\mathcal{U}_{\lambda}}\sum_{k\in\mathcal{U}_i^{\lambda}}\mathbb{E}\Big[\frac{\partial \ln p(\hat{\lambda}_i[k]|\theta,\mu(\hat{\lambda}_i[k]))}{{\partial \theta}}\big(\frac{\partial \ln p(\hat{\lambda}_i[k]|\theta,\mu(\hat{\lambda}_i[k]))}{{\partial \theta}}\big)^T\Big].
\end{equation}
Consider any estimator $\hat{\theta}(\mu)$ of $\theta$ based on a measurement selection $\mu\in\mathcal{M}$, and assume that we have a prior pdf of $\theta=[\beta\ \delta]^T$, denoted as $p(\theta)$. Under some regularity conditions on the pmfs of $\hat{x}_i[k]$ and $\hat{r}_i[k]$, and $p(\theta)$, we have (e.g., \cite{van2004detection,van2004optimum}):
\begin{equation}
\label{eqn:ineq of BCRLB}
R_{\hat{\theta}(\mu)}=\mathbb{E}[(\hat{\theta}(\mu)-\theta)(\hat{\theta}(\mu)-\theta)^T]\succeq \bar{C}(\mu),
\end{equation}
where $R_{\hat{\theta}(\mu)}\in\mathbb{R}^{2\times 2}$ is the error covariance of the estimator $\hat{\theta}(\mu)$, the expectation $\mathbb{E}[\cdot]$ is taken with respect to $p(y(\mu)|\theta,\mu)p(\theta)$, and $\bar{C}(\mu)\in\mathbb{R}^{2\times 2}$ is the BCRLB associated with the measurement selection $\mu$. The BCRLB is defined as (e.g., \cite{van2004detection,van2004optimum})
\begin{equation}
\label{eqn:def of C mu}
\bar{C}(\mu)\triangleq (\mathbb{E}_{\theta}[F_{\theta}(\mu)]+F_p)^{-1},
\end{equation}
where $\mathbb{E}_{\theta}[\cdot]$ denotes the expectation taken with respect to $p(\theta)$, $F_{\theta}(\mu)$ is given by Eq.~\eqref{eqn:FIM}, and $F_p\in\mathbb{R}^{2\times 2}$ encodes the prior knowledge of $\theta$ as
\begin{equation}
\label{eqn:def of F_p}
F_p=-\mathbb{E}_{\theta}\begin{bmatrix}\frac{\partial^2 \ln p(\theta)}{{\partial \beta^2}} & \frac{\partial^2 \ln p(\theta)}{{\partial\beta}{\partial\delta}}\\ \frac{\partial^2 \ln p(\theta)}{{\partial\delta}{\partial\beta}} & \frac{\partial^2 \ln p(\theta)}{{\partial\delta^2}}\end{bmatrix}=\mathbb{E}_{\theta}\Big[\frac{\partial \ln p(\theta)}{\partial \theta}\big(\frac{\partial \ln p(\theta)}{\partial \theta}\big)^T\Big]\succeq\mathbf{0},
\end{equation}
where the second equality holds under some regularity conditions on $p(\theta)$ \cite{van2004detection}. 

Thus, the above arguments motivate us to consider (functions of) $\bar{C}(\cdot)$ as optimization metrics in the measurement selection problem studied in this section, in order to characterize the estimation performance corresponding to a measurement selection $\mu\in\mathcal{M}$. In particular, we will consider $\Tr(\bar{C}(\cdot))$ and $\ln\det(\bar{C}(\cdot))$, which are widely used criteria in parameter estimation (e.g., \cite{joshi2008sensor}), and are also known as the Bayesian A-optimality and D-optimality criteria respectively in the context of experimental design (e.g., \cite{pukelsheim2006optimal}). First, considering the optimization metric $\Tr(\bar{C}(\cdot))$, we see from the above arguments that \eqref{eqn:ineq of BCRLB} directly implies $\Tr(R_{\hat{\theta}(\mu)})\ge\Tr(\bar{C}(\mu))$ for all estimators $\hat{\theta}(\mu)$ of $\theta$ and for all $\mu\in\mathcal{M}$. Therefore, a measurement selection $\mu^{\star}$ that minimizes $\Tr(\bar{C}(\mu))$ potentially yields a lower value of $\Tr(R_{\hat{\theta}(\mu)})$ for an estimator $\hat{\theta}(\mu)$ of $\theta$. Furthermore, there may exist an estimator $\hat{\theta}(\mu)$ that achieves the BCRLB (e.g., \cite{van2004detection}), i.e., $\Tr(\bar{C}(\mu))$ provides the minimum value of $\Tr(R_{\hat{\theta}(\mu)})$ that can be possibly achieved by any estimator $\hat{\theta}(\mu)$ of $\theta$, given a measurement selection $\mu$. Similar arguments hold for $\ln\det(\bar{C}(\cdot))$. To proceed, denoting 
\begin{equation}
\label{eqn:def of fa and fd}
f_a(\mu)\triangleq\Tr(\bar{C}(\mu))\ \text{and}\ f_d(\mu)\triangleq\ln\det(\bar{C}(\mu))\ \forall\mu\in\mathcal{M},
\end{equation}
we define the Parameter Estimation Measurement Selection (PEMS) problem.
\begin{problem}
\label{pro:PEMS}
Consider a discrete-time SIR model given by Eq.~\eqref{eqn:SIR single node} with a directed graph $\mathcal{G}=\{\mathcal{V},\mathcal{E}\}$, a weight matrix $A\in\mathbb{R}^{n\times n}$, a sampling parameter $h\in\mathbb{R}_{\ge0}$, and an initial condition $l=[((s[0])^T\ (x[0])^T\ (r[0])^T]^T$. Moreover, consider time steps $t_1,t_2\in\mathbb{Z}_{\ge1}$ with $t_2\ge t_1$; a set $\mathcal{C}_{k,i}=\{\zeta c_{k,i}:\zeta\in(\{0\}\cup[\zeta_i])\}$ with $c_{k,i}\in\mathbb{R}_{\ge0}$ and $\zeta_i\in\mathbb{Z}_{\ge1}$, for all $i\in\mathcal{V}$ and for all $k\in\{t_1,\dots,t_2\}$; a set $\mathcal{B}_{k,i}=\{\eta b_{k,i}:\eta\in(\{0\}\cup[\eta_i])\}$ with $b_{k,i}\in\mathbb{R}_{\ge0}$ and $\eta_i\in\mathbb{Z}_{\ge1}$, for all $i\in\mathcal{V}$ and for all $k\in\{t_1,\dots,t_2\}$; a budget $B\in\mathbb{R}_{\ge0}$; and a prior pdf $p(\theta)$. Suppose $\hat{x}_i[k]$ (resp., $\hat{r}_i[k]$)  is given by a pmf $p(\hat{x}_i[k]|\theta,\varphi_{k,i})$ (resp., $p(\hat{r}_i[k]|\theta,\omega_{k,i})$), where $\varphi_{k,i}\in\mathcal{C}_{k,i}$ (resp., $\omega_{k,i}\in\mathcal{B}_{k,i}$). The PEMS problem is to find a measurement selection $\mu$ that solves
\begin{equation}
\label{eqn:PEMS obj}
\begin{split}
&\mathop{\min}_{\mu\in\mathcal{M}}f(\mu)\\
s.t.&\ \sum_{\hat{x}_i[k]\in\mathcal{U}_{t_1:t_2}}c_{k,i}\mu(\hat{x}_i[k])+\sum_{\hat{r}_i[k]\in\mathcal{U}_{t_1:t_2}}b_{k,i}\mu(\hat{r}_i[k])\le B,
\end{split}
\end{equation}
where $\mathcal{M}$ is defined in  Eq.~\eqref{eqn:def of M}, $f(\cdot)$ can be either of $f_a(\cdot)$ or $f_d(\cdot)$ with $f_a(\cdot)$ and $f_d(\cdot)$ defined in Eq.~\eqref{eqn:def of fa and fd}, $\mathcal{U}_{t_1:t_2}$ is defined in Eq.~\eqref{eqn:def of U t_1 t_2}, and $\bar{C}(\mu)$ is given by Eq.~\eqref{eqn:def of C mu}.
\end{problem}

Note that $F_p\succeq\mathbf{0}$ from \eqref{eqn:def of F_p}, and $f_a(\mathbf{0})=\Tr(\bar{C}(\mathbf{0}))=\Tr((F_p)^{-1})$ and $f_d(\mathbf{0})=\ln\det(\bar{C}(\mathbf{0}))=\ln\det((F_p)^{-1})$ from Eq.~\eqref{eqn:def of C mu}. We further assume that $F_p\succ\mathbf{0}$ in the sequel, which implies $f(\mu)>0$ for all $\mu\in\mathcal{M}$.

\subsection{Solving the PEMS Problem}\label{sec:PEMS}
In this section, we consider a measurement model with specific pmfs of $\hat{x}_i[k]$ and $\hat{r}_i[k]$ (e.g., \cite{bendavid2020covid} and \cite{hota2020closed}). Nonetheless, our analysis can potentially be extended to other measurement models.
\subsubsection{Pmfs of Measurements $\hat{x}_i[k]$ and $\hat{r}_i[k]$}\label{sec:pmfs of measurements}
Consider any $i\in\mathcal{V}$ and any $k\in\{t_1,\dots,t_2\}$. Assume that the total population of node $i$ is fixed over time and is denoted as $N_i\in\mathbb{Z}_{\ge1}$. Given any measurement selection $\mu\in\mathcal{M}$ with $\mathcal{M}$ defined in Eq.~\eqref{eqn:def of M}, we recall from Section~$\ref{sec:PEMS problem formulation}$ that $\mu(\hat{x}_i[k])c_{k,i}$ can be viewed as the cost of performing virus tests on $\mu(\hat{x}_i[k])N_i^x$ randomly and uniformly sampled individuals in the population of node $i\in\mathcal{V}$, where $\mu(\hat{x}_i[k])\in(\{0\}\cup[\zeta_i])$ (with $\zeta_i\in\mathbb{Z}_{\ge1}$), $c_{k,i}\in\mathbb{R}_{\ge0}$ and $N_i^x\in\mathbb{Z}_{\ge1}$ with $\zeta_i N_i^x\le N_i$. Note that $x_i[k]$ is the proportion of population at node $i$ and at time $k$ that is infected, and $x_i[k]\in[0,1)$ under Assumptions~$\ref{ass:initial condition}$-$\ref{ass:model parameters}$ as shown by Lemma~$\ref{lemma:propagate of initial condition}$. Thus, a  randomly and uniformly sampled individual in the population at node $i$ and at time $k$ will be an infected individual (at time $k$) with probability $x_i[k]$, and will be a non-infected (i.e., susceptible or recovered) individual with probability $1-x_i[k]$. Supposing the tests are accurate,\footnote{Here, ``accurate" means that an infected individual (at time $k$) will be tested positive with probability one, and an individual that is not infected will be tested negative with probability one.} we see from the above arguments that the obtained number of individuals that are tested positive, i.e., $N_i\hat{x}_i[k]$, is a binomial random variable with parameters $N_i^x\mu(\hat{x}_i[k])\in\mathbb{Z}_{\ge1}$ and $x_i[k]\in[0,1)$. Thus, for any $i\in\mathcal{V}$ and for any $k\in\{t_1,\dots,t_2\}$, the pmf of $\hat{x}_i[k]$ is
\begin{equation}
\label{eqn:pmf of xi}
p(\hat{x}_i[k]=x|\theta,\mu(\hat{x}_i[k]))={N_i^x\mu(\hat{x}_i[k])\choose N_ix}(x_i[k])^{N_i x}(1-x_i[k])^{N_i^x\mu(\hat{x}_i[k])-N_ix},
\end{equation}
where $x\in\{0,\frac{1}{N_i},\frac{2}{N_i},\dots,\frac{N_i^x\mu(\hat{x}_i[k])}{N_i}\}$ with $x\in[0,1]$ since $N_i^x\zeta_i\le N_i$. Note that we do not define the pmf of measurement $\hat{x}_i[k]$ when $N_i^x\mu(\hat{x}_i[k])=0$, i.e., when $\mu(\hat{x}_i[k])=0$, since $\mu(\hat{x}_i[k])=0$ indicates no measurement is collected for state $x_i[k]$. Also note that when $x_i[k]=0$, the pmf of $\hat{x}_i[k]$ given in Eq.~\eqref{eqn:pmf of xi} reduces to $p(\hat{x}_i[k]=0|\theta,\mu(\hat{x}_i[k]))=1$. Moreover, since the weight matrix $A\in\mathbb{R}^{n\times n}$ and the sampling parameter $h\in\mathbb{R}_{\ge0}$ are assumed to be given, we see that given $\theta=[\beta\ \delta]^T$ and initial condition $l=[(s[0])^T\ (x[0])^T\ (r[0])^T]^T$, $x_i[k]$ can be obtained using Eq.~\eqref{eqn:SIR I} for all $i\in\mathcal{V}$ and for all $k\in\{t_1,\dots,t_2\}$, where we can view $x_i[k]$ as a function in the unknown parameter $\theta$. In other words, given $l$, $\theta$, $\mu(\hat{x}_i[k])$, $N_i^x$ and $N_i$, one can obtain the right-hand side of Eq.~\eqref{eqn:pmf of xi}. Again, we only explicitly express the dependency of the pmf of $\hat{x}_i[k]$ on $\theta$ and $\mu(\hat{x}_i[k])$ in Eq.~\eqref{eqn:pmf of xi}.  Following similar arguments to those above, we assume that for any $i\in\mathcal{V}$ and for any $k\in\{t_1,\dots,t_2\}$, measurement $\hat{r}_i[k]$ has the following pmf:
\begin{equation}
\label{eqn:pmf of ri}
p(\hat{r}_i[k]=r|\theta,\mu(\hat{r}_i[k]))={N_i^r\mu(\hat{r}_i[k])\choose N_i r}(r_i[k])^{N_i r}(1-r_i[k])^{N_i^r\mu(\hat{r}_i[k])-N_i r},
\end{equation}
where $r\in\{0,\frac{1}{N_i},\frac{2}{N_i},\dots,\frac{N_i^r\mu(\hat{r}_i[k])}{N_i}\}$ with $r\in[0,1]$, $\mu(\hat{r}_i[k])\in\{0,\dots,\eta_i\}$, $N_i^r\in\mathbb{Z}_{\ge1}$ and $N_i^r\mu(\hat{r}_i[k])\le N_i$. Similarly, we note that the pmf of $\hat{r}_i[k]$ given in Eq.~\eqref{eqn:pmf of ri} reduces to $p(\hat{r}_i[k]=0|\theta,\mu(\hat{r}_i[k]))=1$ when $r_i[k]=0$. Considering any measurement selection $\mu\in\mathcal{M}$ and any measurement $\hat{\lambda}_i[k]\in\mathcal{U}_{t_1:t_2}$, where $\lambda\in\{x,r\}$ and $\mathcal{U}_{t_1:t_2}$ is defined in Eq.~\eqref{eqn:def of U t_1 t_2}, we have the following:
\begin{align}\nonumber
&\mathbb{E}\Big[\frac{\partial \ln p(\hat{\lambda}_i[k]|\theta,\mu(\hat{\lambda}_i[k]))}{{\partial \theta}}\big(\frac{\partial \ln p(\hat{\lambda}_i[k]|\theta,\mu(\lambda_i[k]))}{{\partial \theta}}\big)^T\Big]\\
=&\mathbb{E}\Big[\big(\frac{\partial \ln p(\hat{\lambda}_i[k]|\theta,\mu(\hat{\lambda}_i[k]))}{{\partial {\lambda}_i[k]}}\big)^2\cdot \frac{\partial {\lambda}_i[k]}{\partial \theta}\big(\frac{\partial {\lambda}_i[k]}{\partial \theta}\big)^T\Big]\label{eqn:ln pmf xi 1}\\
=&\frac{N_i^{\lambda}\mu(\hat{\lambda}_i[k])}{{\lambda}_i[k](1-{\lambda}_i[k])}\cdot\frac{\partial \lambda_i[k]}{\partial \theta}\big(\frac{\partial \lambda_i[k]}{\partial \theta}\big)^T,\label{eqn:ln pmf xi 2}
\end{align}
where the expectation $\mathbb{E}[\cdot]$ is taken with respect to $p(\hat{\lambda}_i[k]|\theta,\mu(\hat{\lambda}_i[k]))$, and ${\lambda}_i[k]\in[0,1)$. To obtain \eqref{eqn:ln pmf xi 1}, we note the form of $\ln p(\hat{\lambda}_i[k]|\theta,\mu(\hat{\lambda}_i[k]))$ in Eq.~\eqref{eqn:pmf of xi} (or Eq.~\eqref{eqn:pmf of ri}), and use the chain rule. Moreover, one can obtain \eqref{eqn:ln pmf xi 2} from the fact that $\hat{\lambda}_i[k]$ is a binomial random variable. Noting that the pmf of $\hat{\lambda}_i[k]$ reduces to $p(\hat{\lambda}_i[k]=0|\theta,\mu(\hat{\lambda}_i[k]))=1$ if ${\lambda}_i[k]=0$ as argued above, we let the right-hand side of \eqref{eqn:ln pmf xi 2} be zero if ${\lambda}_i[k]=0$.

\subsubsection{Complexity of the PEMS Problem}\label{sec:complexity of pems}
Under the measurement model described above, we show that the PEMS problem is also NP-hard, i.e., there exist instances of the PEMS problem that cannot be solved optimally by any polynomial-time algorithm (if P $\neq$ NP). The proof of the following result is included in Section~$\ref{sec:proof of PEMS NP-hard}$ in the Appendix.
\begin{theorem}
\label{thm:PEMS NP-hard}
The PEMS problem is NP-hard.
\end{theorem}

\subsubsection{An Equivalent Formulation for the PEMS Problem}
Theorem~$\ref{thm:PEMS NP-hard}$ motivates us to consider approximation algorithms for solving the PEMS problem. To begin with, we note that the objective function in the PEMS problem can be viewed as a function defined over an integer lattice. We then have $f_a:\mathcal{M}\to\mathbb{R}_{\ge0}$ and $f_d:\mathcal{M}\to\mathbb{R}_{\ge0}$, where $\mathcal{M}$ is defined in Eq.~\eqref{eqn:def of M}. First, considering $f_a:\mathcal{M}\to\mathbb{R}_{\ge0}$, we will define a set function $f_{Pa}:2^{\bar{M}}\to\mathbb{R}_{\ge0}$, where $\bar{\mathcal{M}}$ is a set constructed as 
\begin{equation}
\label{eqn:def of M bar}
\bar{\mathcal{M}}\triangleq\{(\hat{x}_i[k],l_1):i\in\mathcal{V},k\in\{t_1,\dots,t_2\},l_1\in[\zeta_i]\}\cup\{(\hat{r}_i[k],l_2):i\in\mathcal{V},k\in\{t_1,\dots,t_2\},l_2\in[\eta_i]\}.
\end{equation}
In other words, for any $i\in\mathcal{V}$ and for any $k\in\{t_1,\dots,t_2\}$, we associate elements $(\hat{x}_i[k],1),\dots,(\hat{x}_i[k],\zeta_i)$ (resp., $(\hat{r}_i[k],1),\dots,(\hat{r}_i[k],\eta_i)$) in set $\bar{\mathcal{M}}$ to measurement $\hat{x}_i[k]$ (resp., $\hat{r}_i[k]$). The set function $f_{Pa}(\cdot)$ is then defined as
\begin{equation}
f_{Pa}(\mathcal{Y})\triangleq f_a(\mathbf{0})-f_a(\mu_{\mathcal{Y}})=\Tr(\bar{C}(\mathbf{0}))-\Tr(\bar{C}(\mu_{\mathcal{Y}}))\ \forall \mathcal{Y}\subseteq\bar{\mathcal{M}},\label{eqn:def of f_Pa}
\end{equation}
where for any $\mathcal{Y}\subseteq\bar{\mathcal{M}}$, we define $\mu_{\mathcal{Y}}\in\mathcal{M}$ such that $\mu_{\mathcal{Y}}(\hat{x}_i[k])=|\{(\hat{x}_i[k],l_1):(\hat{x}_i[k],l_1)\in\mathcal{Y}\}|$ and $\mu_{\mathcal{Y}}(\hat{r}_i[k])=|\{(\hat{r}_i[k],l_2):(\hat{r}_i[k],l_2)\in\mathcal{Y}\}|$ for all $i\in\mathcal{V}$ and for all $k\in\{t_1,\dots,t_2\}$. In other words, $\mu_{\mathcal{Y}}(\hat{x}_i[k])$ (resp., $\mu_{\mathcal{Y}}(\hat{r}_i[k])$) is set to be the number of elements in $\mathcal{Y}$ that correspond to the measurement $\hat{x}_i[k]$ (resp., $\hat{r}_i[k]$). Also note that $f_{Pa}(\emptyset)=0$. Following the arguments leading to \eqref{eqn:ln pmf xi 2}, we define
\begin{equation}
\label{eqn:def of H_y}
H_y\triangleq
\begin{cases}
&\mathbb{E}_{\theta}\big[\frac{N_i^x}{x_i[k](1-x_i[k])}\frac{\partial x_i[k]}{\partial \theta}\big(\frac{\partial x_i[k]}{\partial \theta}\big)^T\big]\ \text{if}\ y=(\hat{x}_i[k],l_1)\\
&\mathbb{E}_{\theta}\big[\frac{N_i^r}{r_i[k](1-r_i[k])}\frac{\partial r_i[k]}{\partial \theta}\big(\frac{\partial r_i[k]}{\partial \theta}\big)^T\big]\ \text{if}\ y=(\hat{r}_i[k],l_2)
\end{cases}\ \forall y\in\bar{\mathcal{M}},
\end{equation}
where $x_i[k],r_i[k]\in[0,1)$, $i\in\mathcal{V}$, $k\in\{t_1,\dots.t_2\}$, $l_1\in[\zeta_i]$, $l_2\in[\eta_i]$, and the expectation $\mathbb{E}_{\theta}[\cdot]$ is taken with respect to the prior pdf $p(\theta)$. Given any $\theta=[\beta\ \delta]^T$, we see from the arguments for \eqref{eqn:ln pmf xi 2} that $\frac{N_i^x}{x_i[k](1-x_i[k])}\frac{\partial x_i[k]}{\partial \theta}\big(\frac{\partial x_i[k]}{\partial \theta}\big)^T\succeq\mathbf{0}$. Moreover, one can show that  $\mathbb{E}_{\theta}\big[\frac{N_i^x}{x_i[k](1-x_i[k])}\frac{\partial x_i[k]}{\partial \theta}\big(\frac{\partial x_i[k]}{\partial \theta}\big)^T\big]\succeq\mathbf{0}$. Similarly, one can obtain $\mathbb{E}_{\theta}\big[\frac{N_i^r}{r_i[k](1-r_i[k])}\frac{\partial r_i[k]}{\partial \theta}\big(\frac{\partial r_i[k]}{\partial \theta}\big)^T\big]\succeq\mathbf{0}$, which implies $H_y\succeq\mathbf{0}$ for all $y\in\bar{\mathcal{M}}$. Now, suppose the pmfs of $\hat{x}_i[k]$ and $\hat{r}_i[k]$ are given by Eq.~\eqref{eqn:pmf of xi} and Eq.~\eqref{eqn:pmf of ri}, respectively. Recall from Eq.~\eqref{eqn:def of C mu} that $\Tr(\bar{C}(\mu))=\Tr((\mathbb{E}_{\theta}[F_{\theta}(\mu)]+F_p)^{-1})$ for all $\mu\in\mathcal{M}$, where $F_p$ and $F_{\theta}(\mu)$ are given by \eqref{eqn:def of F_p} and \eqref{eqn:FIM 2nd exp}, respectively. Supposing Assumption~$\ref{ass:white and independent noise}$ holds, for all $\mathcal{Y}\subseteq\bar{\mathcal{M}}$, one can first express $F_\theta(\mu_{\mathcal{Y}})$ using \eqref{eqn:ln pmf xi 2},  and then use Eq.~\eqref{eqn:def of H_y} to obtain $\mathbb{E}_{\theta}[F_{\theta}(\mu_{\mathcal{Y}})]=\sum_{y\in\mathcal{Y}}H_y\triangleq H(\mathcal{Y})$, where $\mu_{\mathcal{Y}}$ is defined above given $\mathcal{Y}\subseteq\bar{\mathcal{M}}$. Putting the above arguments together, we have from Eq.~\eqref{eqn:def of f_Pa} the following:
\begin{equation}
\label{eqn:def of f_Pa 2nd}
f_{Pa}(\mathcal{Y})=\Tr\big((F_p)^{-1}\big)-\Tr\big((F_p+H(\mathcal{Y}))^{-1}\big)\quad \forall\mathcal{Y}\subseteq\bar{\mathcal{M}}.
\end{equation}


Next, let the cost of $(\hat{x}_i[k],l_1)$ be $c_{k,i}$, denoted as $c(\hat{x}_i[k],l_1)$, for all $(\hat{x}_i[k],l_1)\in\bar{\mathcal{M}}$, and let the cost of $(\hat{r}_i[k],l_2)$ be $b_{k,i}$, denoted as $c(\hat{r}_i[k],l_2)$, for all $(\hat{r}_i[k],l_2)\in\bar{\mathcal{M}}$, where $c_{k,i}\in\mathbb{R}_{>0}$ and $b_{k,i}\in\mathbb{R}_{>0}$ are given in the instance of the PEMS problem. Setting the cost structure of the elements in $\bar{\mathcal{M}}$ in this way, we establish an equivalence between the cost of a subset $\mathcal{Y}\subseteq\bar{\mathcal{M}}$  and the cost of $\mu_{\mathcal{Y}}\in\mathcal{M}$, where $\mu_{\mathcal{Y}}$ is defined above. Similarly, considering the objective function $f_d:\mathcal{M}\to\mathbb{R}_{\ge0}$ in the PEMS problem, we define a set function $f_{Pd}:2^{\bar{\mathcal{M}}}\to\mathbb{R}_{\ge0}$ as
\begin{equation}
f_{Pd}(\mathcal{Y})\triangleq f_d(\mathbf{0})-f_d(\mu_{\mathcal{Y}})=\ln\det(F_p+H(\mathcal{Y}))-\ln\det(F_p)\quad \forall\mathcal{Y}\subseteq\bar{\mathcal{M}},\label{eqn:def of f_Pd}
\end{equation}
where we define $\mu_{\mathcal{Y}}\in\mathcal{M}$ such that $\mu_{\mathcal{Y}}(\hat{x}_i[k])=|\{(\hat{x}_i[k],l_1):(\hat{x}_i[k],l_1)\in\mathcal{Y}\}|$ and $\mu_{\mathcal{Y}}(\hat{r}_i[k])=|\{(\hat{r}_i[k],l_2):(\hat{r}_i[k],l_2)\in\mathcal{Y}\}|$ for all $i\in\mathcal{V}$ and for all $k\in\{t_1,\dots,t_2\}$. Note that given an instance of the PEMS problem in Problem~$\ref{pro:PEMS}$, we can construct the set $\bar{\mathcal{M}}$ with the associated costs of the elements in $\bar{\mathcal{M}}$ in $O(n(t_2-t_1+1)(\zeta+\eta))$ time, where $n$ is the number of nodes in graph $\mathcal{G}=\{\mathcal{V},\mathcal{E}\}$, and $\zeta,\eta\in\mathbb{Z}_{\ge1}$ with $\zeta_i\le\zeta$ and $\eta_i\le\eta$ for all $i\in\mathcal{V}$. Assuming that $\zeta$ and $\eta$ are (fixed) constants, the construction of the set $\bar{\mathcal{M}}$ with the associated costs takes $O(n(t_2-t_1+1))$ time, which is polynomial in the parameters of the PEMS problem (Problem~\ref{pro:PEMS}). Based on the above arguments, we further consider the following problem:
\begin{equation}\tag{P}
\label{eqn:equivalent obj PEMS}
\begin{split}
&\mathop{\max}_{\mathcal{Y}\subseteq\bar{\mathcal{M}}} f_P(\mathcal{Y})\\
s.t.&\ c(\mathcal{Y})\le B,
\end{split}
\end{equation}
where $f_P(\cdot)$ can be either of $f_{Pa}(\cdot)$ or $f_{Pd}(\cdot)$ with $f_{Pa}(\cdot)$ and $f_{Pd}(\cdot)$ given by in \eqref{eqn:def of f_Pa 2nd} and \eqref{eqn:def of f_Pd}, respectively, and $c(\mathcal{Y})\triangleq\sum_{y\in\mathcal{Y}}c(y)$ for all $\mathcal{Y}\subseteq\bar{\mathcal{M}}$. By the way we construct $f_P(\cdot)$ and the costs of elements in $\bar{\mathcal{M}}$, one can verify that $\mathcal{Y}_a^{\star}\subseteq\bar{\mathcal{M}}$ (resp., $\mathcal{Y}_d^{\star}\subseteq\bar{\mathcal{M}}$) is an optimal solution to Problem~\eqref{eqn:equivalent obj PEMS} with $f_P(\cdot)=f_{Pa}(\cdot)$ (resp., $f_P(\cdot)=f_{Pd}(\cdot)$) if and only if $\mu_{\mathcal{Y}_a^{\star}}$ (resp., $\mu_{\mathcal{Y}_d^{\star}}$) defined above is an optimal solution to \eqref{eqn:PEMS obj} in Problem~$\ref{pro:PEMS}$ with $f(\cdot)=f_a(\cdot)$ (resp., $f(\cdot)=f_d(\cdot)$). Thus, given a PEMS instance we can first construct $\bar{\mathcal{M}}$ with the associated cost for each element in $\bar{\mathcal{M}}$, and then solve Problem~\eqref{eqn:equivalent obj PEMS}. 

\subsection{Greedy Algorithm for the PEMS Problem}\label{sec:greedy algorithm for PEMS}
Note that Problem~\eqref{eqn:equivalent obj PEMS} can be viewed as a problem of maximizing a set function subject to a knapsack constraint, and greedy algorithms have been proposed to solve this problem with performance guarantees when the objective function is monotone nondecreasing and submodular\footnote{A set function $g:2^{\mathcal{V}}\to\mathbb{R}$, where $\mathcal{V}=[n]$ is the ground set, is said to be monotone nondecreasing if $g(\mathcal{A})\le g(\mathcal{B})$ for all $\mathcal{A}\subseteq\mathcal{B}\subseteq\mathcal{V}$. $g(\cdot)$ is called submodular if $g(\{y\}\cup\mathcal{A})-g(\mathcal{A})\ge g(\{y\}\cup\mathcal{B})-g(\mathcal{B})$ for all $\mathcal{A}\subseteq\mathcal{B}\subseteq\mathcal{V}$ and for all $y\in\mathcal{V}\setminus\mathcal{B}$.}  (e.g., \cite{krause2005note} and \cite{streeter2009online}). Before we formally introduce the greedy algorithm for the PEMS problem, we first note from \eqref{eqn:def of H_y}-\eqref{eqn:def of f_Pd} that given a prior pdf of $\theta$ and any $\mathcal{Y}\subseteq\bar{\mathcal{M}}$, one has to take the expectation $\mathbb{E}_{\theta}[\cdot]$ in order to obtain the value of $f_P(\mathcal{Y})$. However, it is in general intractable to explicitly calculate the integration corresponding to $\mathbb{E}_{\theta}[\cdot]$. Hence, one may alternatively evaluate the value of $f_P(\mathcal{Y})$ using numerical integration with respect to $\theta=[\beta\ \delta]^T$ (e.g., \cite{stoer2013introduction}). Specifically, a typical numerical integration method (e.g., the trapezoid rule) approximates the integral of a function (over an interval) based on a summation of (weighted) function values evaluated at certain points within the integration interval, which incurs an approximation error (see e.g., \cite{stoer2013introduction}, for more details). We then see from \eqref{eqn:def of H_y}-\eqref{eqn:def of f_Pd} that in order to apply the numerical integration method described above to $f_P(\mathcal{Y})$, one has to obtain the values of $x_i[k]$, $r_i[k]$, $\frac{\partial x_i[k]}{\partial \theta}$, and $\frac{\partial r_i[k]}{\partial \theta}$ for a given $\theta$ (within the integration interval), where $i\in\mathcal{V}$ and $t_1\le k\le t_2$ with $t_1,t_2$ given in an instance of the PEMS problem. Recall that the initial conditions $s[0]$, $x[0]$ and $r[0]$ are assumed to be known. We first observe that for any given $\theta$, the values of $x_i[k]$ and $r_i[k]$ for all $i\in\mathcal{V}$ and for all $k\in\{t_1,\dots,t_2\}$ can be obtained using the recursions in \eqref{eqn:SIR single node} in $O((t_2-t_1+1)n^2)$ time. Next, noting that $\frac{\partial x_i[k]}{\partial \theta}=[\frac{\partial x_i[k]}{\partial \beta}\ \frac{\partial x_i[k]}{\partial \delta}]^T$ and $\frac{\partial r_i[k]}{\partial \theta}=[\frac{\partial r_i[k]}{\partial \beta}\ \frac{\partial r_i[k]}{\partial \delta}]^T$, we take the derivative with respect to $\beta$ on both sides of the equations in \eqref{eqn:SIR single node} and obtain
\begin{equation}
\label{eqn:SIR single node derivative}
\begin{split}
&\frac{\partial s_i[k+1]}{\partial \beta}=\frac{\partial s_i[k]}{\partial \beta}-h(\frac{\partial s_i[k]}{\partial \beta}\beta+s_i[k])(\sum_{j\in\bar{\mathcal{N}}_i}a_{ij}x_j[k])-hs_i[k]\beta(\sum_{j\in\bar{\mathcal{N}}_i}a_{ij}\frac{\partial x_j[k]}{\partial \beta}),\\
&\frac{\partial x_i[k+1]}{\partial \beta}=(1-h\delta)\frac{\partial x_i[k]}{\partial \beta}+h(\frac{\partial s_i[k]}{\partial \beta}\beta+s_i[k])(\sum_{j\in\bar{\mathcal{N}}_i}a_{ij}x_j[k])+hs_i[k]\beta(\sum_{j\in\bar{\mathcal{N}}_i}a_{ij}\frac{\partial x_j[k]}{\partial \beta}),\\
&\frac{\partial r_i[k+1]}{\partial \beta}=\frac{\partial r_i[k]}{\partial \beta}+h\delta \frac{\partial x_i[k]}{\partial \beta}.
\end{split}
\end{equation}
Considering any given $\beta$, we can then use the recursion in \eqref{eqn:SIR single node} together with the recursion in \eqref{eqn:SIR single node derivative} to obtain the values of $\frac{\partial x_i[k]}{\partial \beta}$ and $\frac{\partial r_i[k]}{\partial \beta}$ for all $i\in\mathcal{V}$ and for all  $k\in\{t_1,\dots,t_2\}$ in $O((t_2-t_1+1)n^2)$ time. Similarly, considering any given $\delta$, one can obtain the values of $\frac{\partial x_i[k]}{\partial \delta}$ and $\frac{\partial r_i[k]}{\partial \delta}$ for all $i\in\mathcal{V}$ and for all $k\in\{t_1,\dots,t_2\}$ in $O((t_2-t_1+1)n^2)$ time.

Putting the above arguments together and considering the prior pdf of $\theta$, i.e., $p(\theta)$, we see from \eqref{eqn:def of H_y}-\eqref{eqn:def of f_Pd} that for all $\mathcal{Y}\subseteq\bar{\mathcal{M}}$, an approximation of $f_P(\mathcal{Y})$, denoted as $\hat{f}_P(\mathcal{Y})$, can be obtained in $O(n_I(t_2-t_1+1)n^2)$ time, where $n_I\in\mathbb{Z}_{\ge1}$ is the number of points used for the numerical integration method with respective to $\theta$, as we described above.\footnote{We assume that $n_I$ is polynomial in the parameters of the PEMS instance.} Furthermore, in the sequel, we may assume that $\hat{f}_P(\cdot)$ satisfies $|\hat{f}_P(\mathcal{Y})-f_P(\mathcal{Y})|\le\varepsilon/2$ for all $\mathcal{Y}\subseteq\bar{\mathcal{M}}$ (with $\hat{f}_P(\emptyset)=0$), where $\varepsilon\in\mathbb{R}_{\ge0}$.\footnote{Note that $\varepsilon$ is related to the approximation error of the numerical integration method, and $\varepsilon$ will decrease if $n_I$ increases; see, e.g., \cite{stoer2013introduction}, for more details about the numerical integration method.} We are now ready to introduce the greedy algorithm given in Algorithm~$\ref{alg:greedy}$  to solve the PEMS problem, where $\hat{f}_P(\cdot)\in\{\hat{f}_{Pa}(\cdot),\hat{f}_{Pd}(\cdot)\}$ and $\hat{f}_P(\cdot)$ is the approximation of $f_P(\cdot)$ as we described above. From the definition of Algorithm~$\ref{alg:greedy}$, we see that the number of function calls of $\hat{f}_P(\cdot)$ required in the algorithm is $O(|\bar{\mathcal{M}}|^2)$, and thus the overall time complexity of Algorithm~\ref{alg:greedy} is given by $O(n_I(t_2-t_1+1)n^2|\bar{\mathcal{M}}|^2)$.

We proceed to analyze the performance of Algorithm~$\ref{alg:greedy}$ when applied to the PEMS problem. First, one can observe that $f_{Pd}(\mathcal{Y})=\ln\det(F_p+H(\mathcal{Y}))-\ln\det(F_p)$ in Problem \eqref{eqn:equivalent obj PEMS} shares a similar form to the one in \cite{summers2015submodularity}. Thus, using similar arguments to those in \cite{summers2015submodularity}, one can show that $f_{Pd}(\cdot)$ is monotone nondecreasing and submodular with $f_{Pd}(\emptyset)=0$. Noting the assumption that $|\hat{f}_{Pd}(\mathcal{Y})-f_{Pd}(\mathcal{Y})|\le\varepsilon/2$ for all $\mathcal{Y}\subseteq\bar{\mathcal{M}}$, one can show that $y^{\star}$ given by line~6 of Algorithm~\ref{alg:greedy} satisfies that $\frac{\hat{f}_{Pd}(\{y^{\star}\}\cup\mathcal{Y}_2)-\hat{f}_{Pd}(\mathcal{Y}_2)+\varepsilon}{c(y^{\star})}\ge\frac{f_{Pd}(\{y\}\cup\mathcal{Y}_2)-f_{Pd}(\mathcal{Y}_2)-\varepsilon}{c(y)}$ for all $y\in\mathcal{C}$. Similarly, one can show that $\mathop{\max}_{y\in\bar{\mathcal{M}}}f_{Pd}(y)\le f_{Pd}(\mathcal{Y}_1)+\varepsilon$, where $\mathcal{Y}_1$ is given by line~3 in Algorithm~\ref{alg:greedy}. One can then use similar arguments to those for Theorem~1 in \cite{krause2005note} and obtain the following result; the detailed proof is omitted for conciseness.
\begin{theorem}
\label{thm:greedy guarantee f_Pd}
Consider Problem \eqref{eqn:equivalent obj PEMS} with the objective function $f_{Pd}:2^{\bar{\mathcal{M}}}\to\mathbb{R}_{\ge0}$ given by \eqref{eqn:def of f_Pd}. Then Algorithm~$\ref{alg:greedy}$ yields a solution, denoted as $\mathcal{Y}_d^g$, to Problem \eqref{eqn:equivalent obj PEMS} that satisfies
\begin{equation}
\label{eqn:guarantee for greedy f_Pd}
f_{Pd}(\mathcal{Y}_d^g)\ge \frac{1}{2}(1-e^{-1})f_{Pd}(\mathcal{Y}_d^{\star})-(\frac{B}{c_{\mathop{\min}}}+\frac{3}{2})\varepsilon,
\end{equation}
where $\mathcal{Y}_d^{\star}\subseteq\bar{\mathcal{M}}$ is an optimal solution to Problem ~\eqref{eqn:equivalent obj PEMS}, $c_{\mathop{\min}}=\mathop{\min}_{y\in\bar{\mathcal{M}}}c(y)$,\footnote{Note that we can assume without loss of generality that $c(y)\le B$ for all $y\in\bar{\mathcal{M}}$.} and $\varepsilon\in\mathbb{R}_{\ge0}$ satisfies $|\hat{f}_{Pd}(\mathcal{Y})-f_{Pd}(\mathcal{Y})|\le\varepsilon/2$ for all $\mathcal{Y}\subseteq\bar{\mathcal{M}}$.
\end{theorem}

\begin{algorithm}
\caption{Greedy algorithm for PEMS}
\label{alg:greedy}
\begin{algorithmic}[1]
\State{\textbf{Input}: An instance of PEMS transformed into the form in \eqref{eqn:equivalent obj PEMS}}
\State{\textbf{Output}: $\mathcal{Y}_g$}
\State{Find $\mathcal{Y}_1\in\mathop{\arg}{\max}_{y\in\bar{\mathcal{M}}}\hat{f}_P(y)$}
\State{Initialize $\mathcal{Y}_2=\emptyset$ and $\mathcal{C}=\bar{\mathcal{M}}$}
\While{$\mathcal{C}\neq\emptyset$}
\State{Find $y^{\star}\in\mathop{\arg}{\max}_{y\in\mathcal{C}}\frac{\hat{f}_P(\{y\}\cup\mathcal{Y}_2)-\hat{f}_P(\mathcal{Y}_2)}{c(y)}$}
\If{$c(y^{\star})+c(\mathcal{Y}_2)\le B$}
\State{$\mathcal{Y}_2=\{y^{\star}\}\cup\mathcal{Y}_2$}
\EndIf
\State{$\mathcal{C}=\mathcal{C}\setminus\{y^{\star}\}$}
\EndWhile
\State{$\mathcal{Y}_g=\mathop{\arg}{\max}_{\mathcal{Y}\in\{\mathcal{Y}_1,\mathcal{Y}_2\}}\{\hat{f}_P(\mathcal{Y}_1),\hat{f}_P(\mathcal{Y}_2)\}$}
\end{algorithmic}
\end{algorithm}

In contrast to $f_{Pd}(\cdot)$, the objective function $f_{Pa}(\cdot)$ is not submodular in general (e.g., \cite{krause2008near}).  In fact, one can construct examples where the objective function $f_{Pa}(\mathcal{Y})=\Tr((F_p)^{-1}\big)-\Tr\big((F_p+H(\mathcal{Y}))^{-1})$ in the PEMS problem is not submodular. Hence, in order to provide performance guarantees of the greedy algorithm when applied to Problem \eqref{eqn:equivalent obj PEMS} with $f(\cdot)=f_{Pa}(\cdot)$, we will extend the analysis in \cite{krause2005note} to nonsubmodular settings. To proceed, note that for all $\mathcal{A}\subseteq\mathcal{B}\subseteq\bar{\mathcal{M}}$, we have $F_p+H(\mathcal{A})\preceq F_p+H(\mathcal{B})$, which implies $(F_p+H(\mathcal{A}))^{-1}\succeq (F_p+H(\mathcal{B}))^{-1}$ and $\Tr(F_p+H(\mathcal{A}))^{-1})\ge \Tr(F_p+H(\mathcal{B}))^{-1})$ \cite{horn2012matrix}. Therefore, the objective function $f_{Pa}(\cdot)$ is monotone nondecreasing with $f_{Pa}(\emptyset)=0$. We then characterize how close $f_{Pa}(\cdot)$ is to being submodular by introducing the following definition. 

\begin{definition}
\label{def:submodularity ratios}
Consider Problem \eqref{eqn:equivalent obj PEMS} with $f_P(\cdot)=f_{Pa}(\cdot)$, where $f_{Pa}:2^{\bar{\mathcal{M}}}\to\mathbb{R}_{\ge0}$ is defined in \eqref{eqn:def of f_Pa}. Suppose Algorithm~$\ref{alg:greedy}$ is applied to solve Problem \eqref{eqn:equivalent obj PEMS}. For all $j\in\{1,\dots,|\mathcal{Y}_2|\}$, let $\mathcal{Y}_2^j=\{y_1,\dots,y_j\}$ denote the set that contains the first $j$ elements added to set $\mathcal{Y}_2$ in Algorithm~$\ref{alg:greedy}$, and let $\mathcal{Y}_2^0=\emptyset$. The type-1 greedy submodularity ratio of $f_{Pa}(\cdot)$ is defined to be the largest $\gamma_1\in\mathbb{R}$ that satisfies
\begin{equation}
\label{eqn:sbmr 1}
\sum_{y\in \mathcal{A}\setminus\mathcal{Y}_2^j}\big(f_{Pa}(\{y\}\cup\mathcal{Y}_2^j)-f_{Pa}(\mathcal{Y}_2^j)\big)\ge\gamma_1\big(f_{Pa}(\mathcal{A}\cup\mathcal{Y}_2^j)-f_{Pa}(\mathcal{Y}_2^j)\big),
\end{equation}
for all $\mathcal{A}\subseteq\bar{\mathcal{M}}$ and for all $j\in\{0,\dots,|\mathcal{Y}_2|\}$. The type-2 greedy submodularity ratio of $f_{Pa}(\cdot)$ is defined to be the largest $\gamma_2\in\mathbb{R}$ that satisfies
\begin{equation}
\label{eqn:sbmr 2}
f_{Pa}(\mathcal{Y}_1)-f_{Pa}(\emptyset)\ge\gamma_2\big(f_{Pa}(\{y\}\cup\mathcal{Y}_2^j)-f_{Pa}(\mathcal{Y}_2^j)\big),
\end{equation}
for all $j\in\{0,\dots,|\mathcal{Y}_2|\}$ and for all $y\in\bar{\mathcal{M}}\setminus\mathcal{Y}_2^j$ such that $c(y)+c(\mathcal{Y}_2^j)>B$, where $\mathcal{Y}_1\in\mathop{\arg}{\max}_{y\in\bar{\mathcal{M}}}\hat{f}_{Pa}(y)$.
\end{definition}

\begin{remark}
\label{remark:submodularity ratios}
Note that $f_{Pa}(\cdot)$ is monotone nondecreasing as argued above. Noting the definition of $\gamma_1$ in \eqref{eqn:sbmr 1}, one can use similar arguments to those in \cite{bian2017guarantees} and show that $\gamma_1\in[0,1]$; if $f_{Pa}(\cdot)$ is submodular, then $\gamma_1=1$. Similarly, one can show that $\gamma_2\ge 0$. Supposing that $\mathcal{Y}_1\in\mathop{\arg}{\max}_{y\in\bar{\mathcal{M}}}f_{Pa}(y)$, one can further show that if $f_{Pa}(\cdot)$ is submodular, then $\gamma_2\ge1$.
\end{remark}

Note that since we approximate $f_{Pa}(\cdot)$ using $\hat{f}_{Pa}(\cdot)$ as we argued above, we may not be able to obtain the exact values of $\gamma_1$ and $\gamma_2$ from Definition~\ref{def:submodularity ratios}. Moreover, finding $\gamma_1$ may require an exponential number of function calls of $f_{Pa}(\cdot)$ (or $\hat{f}_{Pa}(\cdot)$). Nonetheless, it will be clear from our analysis below that obtaining lower bounds on $\gamma_1$ and $\gamma_2$ suffices. Here, we describe how we obtain a lower bound on $\gamma_2$ using $\hat{f}_{Pa}(\cdot)$, and defer our analysis for lower bounding $\gamma_1$ to the end of this section, which requires more careful analysis. Similarly to \eqref{eqn:sbmr 2}, let $\hat{\gamma}_2$ denote the largest real number that satisfies
\begin{equation}
\label{eqn:gamma_2 lower bound}
\hat{f}_{Pa}(\mathcal{Y}_1)-\frac{\varepsilon}{2}\ge\hat{\gamma}_2\big(\hat{f}_{Pa}(\{y\}\cup\mathcal{Y}_2^j)-\hat{f}_{Pa}(\mathcal{Y}_2^j)+\varepsilon), 
\end{equation}
for all $j\in\{0,\dots,|\mathcal{Y}_2|\}$ and for all $y\in\bar{\mathcal{M}}\setminus\mathcal{Y}_2^j$ such that $c(y)+c(\mathcal{Y}_2^j)>B$. Noting our assumption that $|\hat{f}_{Pa}(\mathcal{Y})-f_{Pa}(\mathcal{Y})|\le\varepsilon/2$ for all $\mathcal{Y}\subseteq\bar{\mathcal{M}}$ (with $\hat{f}_{Pa}(\emptyset)=0$), one can see that $\hat{\gamma}_2$ given by \eqref{eqn:gamma_2 lower bound} also satisfies \eqref{eqn:sbmr 2}, which implies $\hat{\gamma}_2\le\gamma_2$. Given $\mathcal{Y}_2^j$ for all $j\in\{0,\dots,|\mathcal{Y}_2|\}$ from Algorithm~\ref{alg:greedy}, $\hat{\gamma}_2$ can now be obtained via $O(|\bar{\mathcal{M}}|^2)$ function calls of $\hat{f}_{Pa}(\cdot)$. 

Based on Definition~$\ref{def:submodularity ratios}$, the following result extends the analysis in \cite{khuller1999budgeted,krause2005note}, and characterizes the performance guarantees of Algorithm~\ref{alg:greedy} for Problem~\eqref{eqn:equivalent obj PEMS} with $f_P(\cdot)=f_{Pa}(\cdot)$. 
\begin{theorem}
\label{thm:greedy guarantee f_Pa}
Consider Problem \eqref{eqn:equivalent obj PEMS} with the objective function $f_{Pa}:2^{\bar{\mathcal{M}}}\to\mathbb{R}_{\ge0}$ given by \eqref{eqn:def of f_Pa}. Then Algorithm~$\ref{alg:greedy}$ yields a solution, denoted as $\mathcal{Y}_a^g$, to Problem \eqref{eqn:equivalent obj PEMS} that satisfies
\begin{equation}
\label{eqn:guarantee for greedy f_Pa}
f_{Pa}(\mathcal{Y}_a^g)\ge\frac{\mathop{\min}\{\gamma_2,1\}}{2}(1-e^{-\gamma_1}) f_{Pa}(\mathcal{Y}_a^{\star})-(\frac{B+c_{\mathop{\max}}}{c_{\mathop{\min}}}+1)\varepsilon,
\end{equation}
where $\mathcal{Y}_a^{\star}\subseteq\bar{\mathcal{M}}$ is an optimal solution to Problem ~\eqref{eqn:equivalent obj PEMS}, $\gamma_1\in\mathbb{R}_{\ge0}$ and $\gamma_2\in\mathbb{R}_{\ge0}$ are defined in Definition~$\ref{def:submodularity ratios}$, $c_{\mathop{\min}}=\mathop{\min}_{y\in\bar{\mathcal{M}}}c(y)$, $c_{\mathop{\max}}=\mathop{\max}_{y\in\bar{\mathcal{M}}}c(y)$, and $\varepsilon\in\mathbb{R}_{\ge0}$ satisfies $|\hat{f}_{Pa}(\mathcal{Y})-f_{Pa}(\mathcal{Y})|\le\varepsilon/2$ for all $\mathcal{Y}\subseteq\bar{\mathcal{M}}$.
\end{theorem}
\begin{proof}
Noting that \eqref{eqn:guarantee for greedy f_Pa} holds trivially if $\gamma_1=0$ or $\gamma_2=0$, we assume that $\gamma_1>0$ and $\gamma_2>0$. In this proof, we drop the subscript of $f_{Pa}(\cdot)$ (resp., $\hat{f}_{Pa}(\cdot)$) and denote $f(\cdot)$ (resp., $\hat{f}(\cdot)$) for notational simplicity. First, recall that for all $j\in\{1,\dots,|\mathcal{Y}_2|\}$, we let $\mathcal{Y}_2^j=\{y_1,\dots,y_j\}$ denote the set that contains the first $j$ elements added to set $\mathcal{Y}_2$ in Algorithm~$\ref{alg:greedy}$, and let $\mathcal{Y}_2^0=\emptyset$. Now, let $j_l$ be the first index in $\{1,\dots,|\mathcal{Y}_2|\}$ such that a candidate element $y^{\star}\in\mathop{\arg}{\max}_{y\in\mathcal{C}}\frac{\hat{f}(\{y\}\cup\mathcal{Y}_2^{j_l})-\hat{f}(\mathcal{Y}_2^{j_l})}{c(y)}$ for $\mathcal{Y}_2$ (given in line $6$ of Algorithm~$\ref{alg:greedy}$) cannot be added to $\mathcal{Y}_2$ due to $c(y^{\star})+c(\mathcal{Y}_2^{j_l})>B$. In other words, for all $j\in\{0,\dots,j_l-1\}$, any candidate element $y^{\star}\in\mathop{\arg}{\max}_{y\in\mathcal{C}}\frac{\hat{f}(\{y\}\cup\mathcal{Y}_2^j)-\hat{f}(\mathcal{Y}_2^j)}{c(y)}$ for $\mathcal{Y}_2$ satisfies $c(y^{\star})+c(\mathcal{Y}_2^j)\le B$ and can be added to $\mathcal{Y}_2$ in Algorithm~\ref{alg:greedy}. Noting that $|\hat{f}_P(\mathcal{Y})-f_P(\mathcal{Y})|\le\varepsilon/2$ for all $\mathcal{Y}\subseteq\bar{\mathcal{M}}$, one can then show that the following hold for all $j\in\{0,\dots,j_l-1\}$:
\begin{equation}
\label{eqn:greedy choice}
\frac{f(\mathcal{Y}_2^{j+1})-f(\mathcal{Y}_2^j)+\varepsilon}{c(y_{j+1})}\ge\frac{f(\{y\}\cup\mathcal{Y}_2^j)-f(\mathcal{Y}_2^j)-\varepsilon}{c(y)}\quad \forall y\in\bar{\mathcal{M}}\setminus\mathcal{Y}_2^j.
\end{equation}
Now, considering any $j\in\{0,\dots,j_l-1\}$, we have the following:
\begin{align}
&\quad f(\mathcal{Y}^{\star}_a\cup \mathcal{Y}_2^j)-f(\mathcal{Y}_2^j)\le\frac{1}{\gamma_1}\sum_{y\in \mathcal{Y}^{\star}_a\setminus \mathcal{Y}^j_2}c(y)\cdot\frac{f(\{y\}\cup \mathcal{Y}^j_2)-f(\mathcal{Y}^j_2)}{c(y)}\label{eqn:ineq from sbmr}\\
&\le\frac{1}{\gamma_1}\sum_{y\in \mathcal{Y}^{\star}_a\setminus \mathcal{Y}_2^j}c(y)(\frac{f(\mathcal{Y}^{j+1}_2)-f(\mathcal{Y}^j_2)+\varepsilon}{c(y_{j+1})}+\frac{\varepsilon}{c(y)})\label{eqn:greedy choice use}\\
&\le\frac{B}{\gamma_1}\cdot\frac{f(\mathcal{Y}^{j+1}_2)-f(\mathcal{Y}^j_2)}{c(y_{j+1})}+\frac{\varepsilon}{\gamma_1}\sum_{y\in\mathcal{Y}_a^{\star}\setminus\mathcal{Y}_j^2}(\frac{c(y)}{c(y_{j+1})}+1)\label{eqn:Y star feasible}\\
&\le\frac{B}{\gamma_1}\cdot\frac{f(\mathcal{Y}^{j+1}_2)-f(\mathcal{Y}^j_2)}{c(y_{j+1})}+\frac{\varepsilon}{\gamma_1}(\frac{B}{c(y_{j+1})}+|\mathcal{Y}_a^{\star}|),\label{eqn:Y star feasible 1}
\end{align}
where \eqref{eqn:ineq from sbmr} follows from the definition of $\gamma_1$ in \eqref{eqn:sbmr 1}, and \eqref{eqn:greedy choice use} follows from \eqref{eqn:greedy choice}. To obtain \eqref{eqn:Y star feasible}, we use the fact $c(\mathcal{Y}^{\star}_a)\le B$. Similarly, we obtain \eqref{eqn:Y star feasible 1}. Noting that $f(\cdot)$ is monotone nondecreasing, one can further obtain from \eqref{eqn:Y star feasible 1} that
\begin{equation}
\label{eqn:function value diff}
f(\mathcal{Y}_2^{j+1})-f(\mathcal{Y}_2^j)\ge\frac{\gamma_1c(y_{j+1})}{B}\big(f(\mathcal{Y}_a^{\star})-f(\mathcal{Y}_2^j)\big)-\varepsilon(1+|\mathcal{Y}_a^{\star}|\frac{c(y_{j+1})}{B}).
\end{equation}
To proceed,  let $y^{\prime}\in\mathop{\arg}{\max}_{y\in\mathcal{C}}\frac{\hat{f}(\{y\}\cup\mathcal{Y}_2^{j_l})-\hat{f}(\mathcal{Y}_2^{j_l})}{c(y)}$ be the (first) candidate element for $\mathcal{Y}_2$ that cannot be added to $\mathcal{Y}_2$ due to $c(y^{\prime})+c(\mathcal{Y}_2^{j_l})>B$, as we argued above. Similarly to \eqref{eqn:greedy choice}, one can see that $\frac{f(\{y^{\prime}\}\cup\mathcal{Y}_2^{j_l})-f(\mathcal{Y}_2^{j_l})+\varepsilon}{c(y^{\prime})}\ge\frac{f(\{y\}\cup\mathcal{Y}_2^{j_l})-f(\mathcal{Y}_2^{j_l})-\varepsilon}{c(y)}$ holds for all $y\in\bar{\mathcal{M}}\setminus\mathcal{Y}_2^{j_l}$. Letting $\bar{\mathcal{Y}}_2^{j_l+1}\triangleq\{y^{\prime}\}\cup\mathcal{Y}_2^{j_l}$ and following similar arguments to those leading up to \eqref{eqn:function value diff}, we have
\begin{equation}
\label{eqn:function value diff 1}
f(\bar{\mathcal{Y}}_2^{j_l+1})-f(\mathcal{Y}_2^{j_l})\ge\frac{\gamma_1c(y^{\prime})}{B}\big(f(\mathcal{Y}_a^{\star})-f(\mathcal{Y}_2^{j_l})\big)-\varepsilon(1+|\mathcal{Y}_a^{\star}|\frac{c(y^{\prime})}{B}).
\end{equation}
Denoting $\Delta_j\triangleq f(\mathcal{Y}^{\star}_a)-f(\mathcal{Y}^j_2)$ for all $j\in\{0,\dots,j_l\}$ and $\Delta_{j_l+1}\triangleq f(\mathcal{Y}_a^{\star})-f(\bar{\mathcal{Y}}_2^{j_l+1})$, we obtain from \eqref{eqn:function value diff} and \eqref{eqn:function value diff 1} the following:
\begin{equation}
\label{eqn:recursion for D_j}
\Delta_j\le \Delta_{j-1}(1-\frac{c(y_j)\gamma_1}{B})+\varepsilon+\frac{c(y_j)|\mathcal{Y}_a^{\star}|}{B}\varepsilon\quad \forall j\in[j_l+1].
\end{equation}
Unrolling \eqref{eqn:recursion for D_j} yields
\begin{align}
&\Delta_{j_l+1}\le \Delta_0\big(\prod_{j=1}^{j_l}(1-\frac{c(y_j)\gamma_1}{B})\big)(1-\frac{c(y^{\prime})\gamma_1}{B})+(j_l+1+\frac{c(\bar{\mathcal{Y}}_2^{j_l+1})|\mathcal{Y}_a^{\star}|}{B})\varepsilon,\label{eqn:recursion for D_j from D_0}\\
\implies&\Delta_{j_l+1}\le \Delta_0\big(\prod_{j=1}^{j_l}(1-\frac{c(y_j)\gamma_1}{B})\big)(1-\frac{c(y^{\prime})\gamma_1}{B})+\frac{2(B+c_{\mathop{\max}})}{c_{\mathop{\min}}}\varepsilon.\label{eqn:recursion for D_j from D_0 1}
\end{align}
To obtain \eqref{eqn:recursion for D_j from D_0}, we use the facts that $(1-\frac{c(y_j)\gamma_1}{B})\le 1$ for all $j\in[j_l+1]$ and $(1-\frac{c(y^{\prime})\gamma_1}{B})\le 1$, since $\gamma_1\in(0,1]$ as we argued in Remark~\ref{remark:submodularity ratios}. To obtain \eqref{eqn:recursion for D_j from D_0 1}, we first note from the way we defined $j_l$ that $j_l+1\le c(\bar{\mathcal{Y}}^{j_l+1}_2)/c_{\mathop{\min}}\le (B+c_{\mathop{\max}})/c_{\mathop{\min}}$. Also noting that $|\mathcal{Y}_a^{\star}|\le B/c_{\mathop{\min}}$, we then obtain \eqref{eqn:recursion for D_j from D_0 1}.

Now, one can show that $\big(\prod_{j=1}^{j_l}(1-\frac{c(y_j)\gamma_1}{B})\big)(1-\frac{c(y^{\prime})\gamma_1}{B})\le\prod_{j=1}^{j_l+1}(1-\frac{c(\bar{\mathcal{Y}}_2^{j_l+1})\gamma_1}{(j_l+1) B})\le e^{-\gamma_1\frac{c(\bar{\mathcal{Y}}_2^{j_l+1})}{B}}$ (e.g., \cite{khuller1999budgeted}). We then have from \eqref{eqn:recursion for D_j from D_0 1} the following:
\begin{align}\nonumber
&f(\mathcal{Y}_a^{\star})-f(\bar{\mathcal{Y}}_2^{j_l+1})\le f(\mathcal{Y}^{\star}_a)e^{-\gamma_1\frac{c(\bar{\mathcal{Y}}_2^{j_l+1})}{B}}+\frac{2(B+c_{\mathop{\max}})}{c_{\mathop{\min}}}\varepsilon\\
\implies& f(\bar{\mathcal{Y}}_2^{j_l+1})\ge(1-e^{-\gamma_1})f(\mathcal{Y}_a^{\star})-\frac{2(B+c_{\mathop{\max}})}{c_{\mathop{\min}}}\varepsilon,\label{eqn:gamma_1 ratio}
\end{align}
where \eqref{eqn:gamma_1 ratio} follows from $c(\bar{\mathcal{Y}}_2^{j_l+1})>B$.

To proceed with the proof of the theorem, we note from the definition of $\gamma_2$ in Definition~$\ref{def:submodularity ratios}$ that $f(\{y^{\prime}\}\cup\mathcal{Y}_2^{j_l})-f(\mathcal{Y}_2^{j_l})\le\frac{1}{\gamma_2}f(\mathcal{Y}_1)$ with $\gamma_2>0$, which together with \eqref{eqn:gamma_1 ratio} imply that 
\begin{equation}
\label{eqn:combine two terms}
f(\mathcal{Y}_2^{j_l})+\frac{1}{\gamma_2}f(\mathcal{Y}_1)\ge f(\bar{\mathcal{Y}}_2^{j_l+1})\ge (1-e^{\gamma_1})f(\mathcal{Y}_a^{\star})-\frac{2\bar{B}}{c_{\mathop{\min}}}\varepsilon,
\end{equation}
where $\bar{B}\triangleq B+c_{\mathop{\max}}$. Since $f(\cdot)$ is monotone nondecreasing, we obtain from \eqref{eqn:combine two terms}
\begin{equation}
\label{eqn:combine two terms 1}
f(\mathcal{Y}_2)+\frac{1}{\gamma_2}f(\mathcal{Y}_1)\ge (1-e^{\gamma_1})f(\mathcal{Y}_a^{\star})-\frac{2\bar{B}}{c_{\mathop{\min}}}\varepsilon.
\end{equation}
We will then split our analysis into two cases. First, supposing that $\gamma_2\ge1$, we see from \eqref{eqn:combine two terms 1} that at least one of $f(\mathcal{Y}_2)\ge\frac{1}{2}(1-e^{-\gamma_1})f(\mathcal{Y}_a^{\star})-\frac{\bar{B}}{c_{\mathop{\min}}}\varepsilon$ and $f(\mathcal{Y}_1)\ge\frac{1}{2}(1-e^{-\gamma_1})f(\mathcal{Y}_a^{\star})-\frac{\bar{B}}{c_{\mathop{\min}}}\varepsilon$ holds. Recalling that $|\hat{f}(\mathcal{Y})-f(\mathcal{Y})|\le\varepsilon/2$ for all $\mathcal{Y}\subseteq\bar{\mathcal{M}}$, it follows that at least one of $\hat{f}(\mathcal{Y}_2)\ge\frac{1}{2}(1-e^{-\gamma_1})f(\mathcal{Y}_a^{\star})-\frac{\bar{B}}{c_{\mathop{\min}}}\varepsilon-\frac{\varepsilon}{2}$ and $\hat{f}(\mathcal{Y}_1)\ge\frac{1}{2}(1-e^{-\gamma_1})f(\mathcal{Y}_a^{\star})-\frac{\bar{B}}{c_{\mathop{\min}}}\varepsilon-\frac{\varepsilon}{2}$ holds. Second, supposing $\gamma_2<1$ and using similar arguments, we have that at least one of $\hat{f}(\mathcal{Y}_2)\ge\frac{1}{2}(1-e^{-\gamma_1})f(\mathcal{Y}_a^{\star})-\frac{\bar{B}}{c_{\mathop{\min}}}\varepsilon-\frac{\varepsilon}{2}$ and $\hat{f}(\mathcal{Y}_1)\ge\frac{\gamma_2}{2}(1-e^{-\gamma_1})f(\mathcal{Y}_a^{\star})-\frac{\gamma_2\bar{B}}{c_{\mathop{\min}}}\varepsilon-\frac{\varepsilon}{2}$ holds. Now, we note from line~12 of Algorithm~$\ref{alg:greedy}$ that $\hat{f}(\mathcal{Y}_a^g)\ge\mathop{\max}\{\hat{f}(\mathcal{Y}_1),\hat{f}(\mathcal{Y}_2)\}$, which implies $f(\mathcal{Y}_a^g)\ge\mathop{\max}\{\hat{f}(\mathcal{Y}_1),\hat{f}(\mathcal{Y}_2)\}-\frac{\varepsilon}{2}$. Combining the above arguments together, we obtain \eqref{eqn:guarantee for greedy f_Pa}.
\end{proof}

\begin{remark}
Note that \eqref{eqn:guarantee for greedy f_Pa} becomes $f_{Pa}(\mathcal{Y}_a^g)\ge\frac{1}{2}(1-e^{-\gamma_1}) f_{Pa}(\mathcal{Y}_a^{\star})-(\frac{B+c_{\mathop{\max}}}{c_{\mathop{\min}}}+1)\varepsilon$ if $\gamma_2\ge1$. Also note that $\gamma_2\ge1$ can hold when the objective function $f_{Pa}(\cdot)$ is not submodular, as we will see later in our numerical examples. 
\end{remark}

\begin{remark}
The authors in \cite{Tzoumas2020LQG} also extended the analysis of Algorithm~\ref{alg:greedy} to nonsubmodular settings, when the objective function can be exactly evaluated (i.e., $\varepsilon=0$). They obtained a performance guarantee for Algorithm~\ref{alg:greedy} that depends on a submodularity ratio defined in a different manner. One can show that the submodularity ratios defined in Definition~$\ref{def:submodularity ratios}$ are lower bounded by the one defined in \cite{Tzoumas2020LQG}, which further implies that the performance bound (when $\varepsilon$ is $0$) for Algorithm~$\ref{alg:greedy}$ given in Theorem~$\ref{thm:greedy guarantee f_Pa}$ is tighter than that provided in  \cite{Tzoumas2020LQG}.
\end{remark}

Finally, we aim to provide a lower bound on $\gamma_1$ that can be computed in polynomial time. The lower bounds on $\gamma_1$ and $\gamma_2$ together with Theorem~\ref{thm:greedy guarantee f_Pa} will also provide performance guarantees for the greedy algorithm. 
\begin{lemma}(\cite{horn2012matrix})
\label{lemma:eigenvalue ineqs}
For any positive semidefinite matrices $P,Q\in\mathbb{R}^{n\times n}$, $\lambda_1(P)\le\lambda_1(P+Q)\le\lambda_1(P)+\lambda_1(Q)$, and $\lambda_n(P+Q)\ge\lambda_n(P)+\lambda_n(Q)$.
\end{lemma}

We have the following result; the proof is included in Section~$\ref{sec:proof of lower bound}$ in the Appendix.
\begin{lemma}
\label{lemma:lower bound on gamma_1}
Consider the set function $f_{Pa}:2^{\bar{\mathcal{M}}}\to\mathbb{R}_{\ge0}$ defined in \eqref{eqn:def of f_Pa}. The type-1 greedy submodularity ratio of $f_{Pa}(\cdot)$ given by Definition~$\ref{def:submodularity ratios}$ satisfies
\begin{equation}
\label{eqn:lower bound on gamma_1}
\gamma_1\ge\mathop{\min}_{j\in\{0,\dots,|\mathcal{Y}_2|\}}\frac{\lambda_2(F_p+H(\mathcal{Y}_2^j))\lambda_2(F_p+H(\{z_j\}\cup\mathcal{Y}_2^j))}{\lambda_1(F_p+H(\mathcal{Y}_2^j))\lambda_1(F_p+H(\{z_j\}\cup\mathcal{Y}_2^j))},
\end{equation}
where $\mathcal{Y}_2^j$ contains the first $j$ elements added to $\mathcal{Y}_2$ in Algorithm~$\ref{alg:greedy}$ $\forall j\in\{1,\dots,|\mathcal{Y}_2|\}$ with $\mathcal{Y}_2^0=\emptyset$, $F_p$ is given by \eqref{eqn:def of F_p}, $H(\mathcal{Y})=\sum_{y\in\mathcal{Y}}H_y$ $\forall \mathcal{Y}\subseteq\bar{\mathcal{M}}$ with $H_y\succeq\mathbf{0}$ defined in \eqref{eqn:def of H_y}, and $z_j\in\mathop{\arg}{\min}_{y\in\bar{\mathcal{M}}\setminus\mathcal{Y}_2^j}\frac{\lambda_2(F_p+H(\{y\}\cup\mathcal{Y}_2^j))}{\lambda_1(F_p+H(\{y\}\cup\mathcal{Y}_2^j))}$ $\forall j\in\{1,\dots,|\mathcal{Y}_2|\}$. 
\end{lemma}

Recalling our arguments at the beginning of this section, we may only obtain approximations of the entries in the ($2$ by $2$) matrix $F_p+H(\mathcal{Y})$ for $\mathcal{Y}\subseteq\bar{\mathcal{M}}$ using, e.g., numerical integration, where $H_y$ (resp., $F_p$) is defined in \eqref{eqn:def of H_y} (resp.,  \eqref{eqn:def of F_p}). Specifically, for all $\mathcal{Y}\subseteq\bar{\mathcal{M}}$, let $\hat{H}(\mathcal{Y})=(F_p+H(\mathcal{Y}))+E(\mathcal{Y})$ be the approximation of $F_p+H(\mathcal{Y})$, where each entry of $E(\mathcal{Y})\in\mathbb{R}^{2\times 2}$ represents the approximation error of the corresponding entry in $F_p+H(\mathcal{Y})$. Since $F_p$ and $H(\mathcal{Y})$ are positive semidefinite matrices, $E(\mathcal{Y})$ is a symmetric matrix. Now, using a standard eigenvalue perturbation result, e.g., Corollary~6.3.8 in  \cite{horn2012matrix}, one can obtain that $\sum_{i=1}^2|\lambda_i(F_p+H(\mathcal{Y}))-\lambda_i(\hat{H}(\mathcal{Y})|^2\le\|E(\mathcal{Y})\|_F^2$ for all $\mathcal{Y}\subseteq\bar{\mathcal{M}}$, where $\|E(\mathcal{Y})\|_F\triangleq\sqrt{\Tr(E(\mathcal{Y})^{\top}E(\mathcal{Y}))}$ denotes the Frobenius norm of a matrix. It then follows that
\begin{equation*}
\frac{\lambda_2(F_p+H(\mathcal{Y}))}{\lambda_1(F_p+H(\mathcal{Y}))}\ge\frac{\lambda_2(\hat{H}(\mathcal{Y}))-\|E(\mathcal{Y})\|_F}{\lambda_1(\hat{H}(\mathcal{Y}))+\|E(\mathcal{Y})\|_F}\ge\frac{\lambda_2(\hat{H}(\mathcal{Y}))-\varepsilon^{\prime}}{\lambda_1(\hat{H}(\mathcal{Y}))+\varepsilon^{\prime}}\quad\forall \mathcal{Y}\subseteq\bar{\mathcal{M}},
\end{equation*}
where $\varepsilon^{\prime}\in\mathbb{R}_{\ge0}$ and satisfies $\|E(\mathcal{Y})\|_F\le\varepsilon^{\prime}$ for all $\mathcal{Y}\subseteq\bar{\mathcal{M}}$. Combining the above arguments with \eqref{eqn:lower bound on gamma_1} in Lemma~\ref{lemma:lower bound on gamma_1}, we obtain a lower bound on $\gamma_1$ that can be computed using $O(|\bar{\mathcal{M}}|^2)$ function calls of $\hat{H}(\cdot)$.
 
\subsubsection{Illustrations}
\label{sec:illustrations}
Note that one can further obtain from \eqref{eqn:lower bound on gamma_1}:
\begin{equation}
\label{eqn:lower bound on gamma_1 2nd}
\gamma_1\ge\mathop{\min}_{j\in\{0,\dots,|\mathcal{Y}_2|\}}\frac{\lambda_2(F_p)+\lambda_2(H(\mathcal{Y}_2^j))}{\lambda_1(F_p)+\lambda_1(H(\mathcal{Y}_2^j))}\cdot\frac{\lambda_2(F_p)+\lambda_2(H(z_j))+\lambda_2(H(\mathcal{Y}_2^j))}{\lambda_1(F_p)+\lambda_1(H(z_j))+\lambda_1(H(\mathcal{Y}_2^j))},
\end{equation}
where $z_j\in\mathop{\arg}{\min}_{y\in\bar{\mathcal{M}}\setminus\mathcal{Y}_2^j}\frac{\lambda_2(F_p+H(\{y\}\cup\mathcal{Y}_2^j))}{\lambda_1(F_p+H(\{y\}\cup\mathcal{Y}_2^j))}$. Supposing $F_p$ is fixed, we see from \eqref{eqn:lower bound on gamma_1 2nd} that the lower bound on $\gamma_1$ would potentially increase if $\lambda_2(H(z_j))/\lambda_1(H(z_j))$ and $\lambda_2(H(\mathcal{Y}_2^j))/\lambda_1(H(\mathcal{Y}_2^j))$ increase. Recall that $F_p$ given by \eqref{eqn:def of F_p} encodes the prior knowledge that we have about $\theta=[\beta\ \delta]^T$. Moreover, recall from \eqref{eqn:def of H_y} that $H(y)$ depends on the prior pdf $p(\theta)$ and the dynamics of the SIR model in \eqref{eqn:SIR single node}. Thus, the lower bound given by Lemma~$\ref{lemma:lower bound on gamma_1}$ and thus the corresponding performance bound of Algorithm~$\ref{alg:greedy}$ given in Theorem~$\ref{thm:greedy guarantee f_Pa}$ depend on the prior knowledge that we have on $\theta=[\beta\ \delta]^T$ and the dynamics of the SIR model.  Also note that the performance bounds given in Theorem~$\ref{thm:greedy guarantee f_Pa}$ are {\it worst-case} performance bounds for Algorithm~$\ref{alg:greedy}$. Thus, in practice the ratio between a solution returned by the algorithm and an optimal solution can be smaller than the ratio predicted by Theorem~$\ref{thm:greedy guarantee f_Pa}$, as we will see in our simulations in next section. Moreover, instances with tighter performance bounds potentially imply better performance of the algorithm when applied to those instances. Similar arguments hold for the performance bounds provided in Theorem~$\ref{thm:greedy guarantee f_Pd}$.

\subsubsection{Simulations}
\label{sec:simulations}
To validate the theoretical results in Theorems~$\ref{thm:greedy guarantee f_Pd}$ and $\ref{thm:greedy guarantee f_Pa}$, and Lemma~$\ref{lemma:lower bound on gamma_1}$, we consider various PEMS instances.\footnote{In our simulations, we neglect the approximation error corresponding to the numerical integrations discussed in Section~\ref{sec:greedy algorithm for PEMS}, since the error terms are found to be sufficiently small.} The directed network $\mathcal{G}=\{\mathcal{V},\mathcal{E}\}$ is given by Fig.~$\ref{fig:parameters}(a)$. According to the existing literature about the estimated infection and recovery rates for the COVID-19 pandemic (e.g., \cite{prem2020effect}), we assume that the infection rate $\beta$ and the recovery rate $\delta$ lie in the intervals $[3,7]$ and $[1,4]$, respectively. Let the prior pdf of $\beta$ (resp., $\delta$) be a (linearly transformed) Beta distribution with parameters $\alpha_1=6$ and $\alpha_2=3$ (resp., $\alpha_1=3$ and $\alpha_2=4$), where $\beta$ and $\delta$ are also assumed to be independent. The prior pdfs of $\beta$ and $\delta$ are then plotted in Fig.~$\ref{fig:parameters}(b)$ and Fig.~$\ref{fig:parameters}(c)$, respectively. We set the sampling parameter $h=0.1$. We then randomly generate the weight matrix $A\in\mathbb{R}^{5\times 5}$ such that Assumptions~$\ref{ass:initial condition}$-$\ref{ass:model parameters}$ are satisfied, where each entry of $A$ is drawn (independently) from certain uniform distributions. The initial condition is set to be $s_1[0]=0.95$, $x_1[0]=0.05$ and $r_1[0]=0$, and $s_i[0]=0.99$, $x_i[0]=0.01$ and $r_i[0]=0$ for all $i\in\{2,\dots,5\}$. In the  pmfs of measurements $\hat{x}_i[k]$ and $\hat{r}_i[k]$ given in Eq.~\eqref{eqn:pmf of xi} and Eq.~\eqref{eqn:pmf of ri}, respectively, we set $N_i^x=N_i^r=100$ and $N_i=1000$ $\forall i\in\mathcal{V}$, where $N_i$ is the total population at node $i$.

\setkeys{Gin}{width=0.3\textwidth}
 \begin{figure}[htbp]
 \center 
 \subfloat[a][Network]{
 \includegraphics[width=0.25\linewidth]{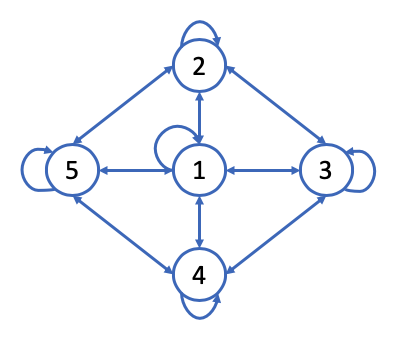}
 }
 \subfloat[b][Prior pdf of $\beta$]{
 \includegraphics[width=0.30\linewidth]{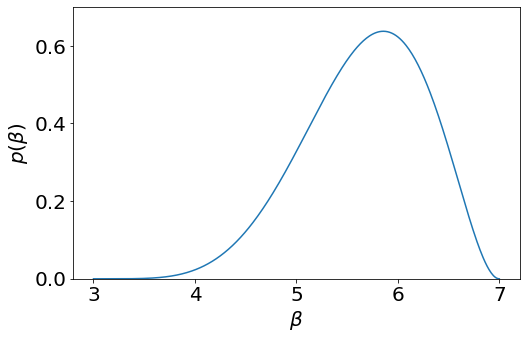}}
 \subfloat[c][Prior pdf of $\delta$]{
 \includegraphics[width=0.30\linewidth]{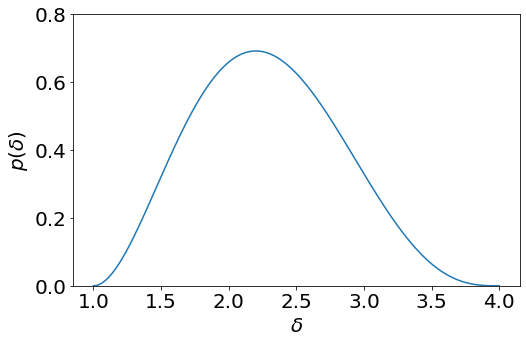}}\\
 \caption{Network structure and prior pdfs of $\beta$ and $\delta$.}
 \label{fig:parameters}
 \end{figure}
  
First, let us consider PEMS instances with a relatively smaller size. In such instances, we set the time steps $t_1=t_2=5$, i.e., we only consider collecting measurements at time step $t=5$. In the sets $\mathcal{C}_{5,i}=\{\zeta c_{5,i}:\zeta\in(\{0\}\cup[\zeta_i])\}$ and 
$\mathcal{B}_{5,i}=\{\eta b_{5,i}:\eta\in(\{0\}\cup[\eta_i])\}$, we let $c_{5,i}=b_{5,i}$ and $\zeta_i=\eta_i=2$ for all $i\in\mathcal{V}$, and draw $c_{5,i}$ and $b_{5,i}$ uniformly randomly from $\{1,2,3\}$. Here, we can choose to perform $0$, $100$, or $200$ virus (or antibody) tests at a node $i\in\mathcal{V}$ and at $k=5$. In Fig.~$\ref{fig:small instances}(a)$, we consider the objective function $f_{Pd}(\cdot)$, given by Eq.~\eqref{eqn:def of f_Pd}, in the PEMS instances constructed above, and plot the greedy solutions and the optimal solutions (found by brute force) to the PEMS instances under different values of budget $B$. Note that for all the simulation results in this section, we obtain the averaged results from $50$ randomly generated $A$ matrices as described above, for each value of $B$. As shown in Theorem~$\ref{thm:greedy guarantee f_Pd}$, the greedy algorithm yields a $\frac{1}{2}(1-e^{-1})\approx0.31$ approximation for $f_{Pd}(\cdot)$ (in the worst case), and the results in Fig.~$\ref{fig:small instances}(a)$ show that the greedy algorithm performs near optimally for the PEMS instances generated above. Similarly, in Fig.~$\ref{fig:small instances}(b)$, we plot the greedy solutions and the optimal solutions to the PEMS instances constructed above under different values of $B$, when the objective function is $f_{Pa}(\cdot)$ given in Eq.~\eqref{eqn:def of f_Pa}. Again, the results in Fig.~$\ref{fig:small instances}(b)$ show that the greedy algorithm performs well for the constructed PEMS instances. Moreover, according to Lemma~$\ref{lemma:lower bound on gamma_1}$, we plot the lower bound on the submodularity ratio $\gamma_1$ of $f_{Pa}(\cdot)$ in Fig.~$\ref{fig:small instances}(c)$. Here, we note that the submodularity ratio $\gamma_2$ of $f_{Pa}(\cdot)$ is always greater than one in the PEMS instances constructed above. Hence, Theorem~$\ref{thm:greedy guarantee f_Pa}$ yields a $\frac{1}{2}(1-e^{-\gamma_1})$ worst-case approximation guarantee for the greedy algorithm, where $\frac{1}{2}(1-e^{-0.3})\approx0.13$.

 \begin{figure}[htbp]
 \center 
 \subfloat[a][OPT vs. Greedy for $f_{Pd}(\cdot)$]{
 \includegraphics[width=0.3\linewidth]{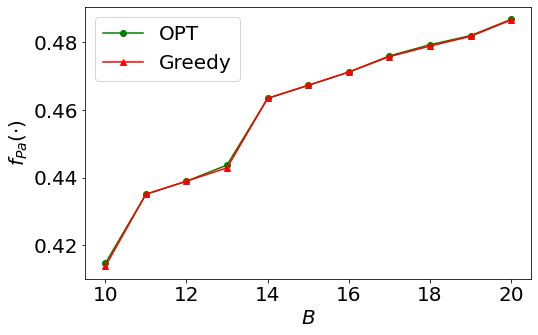}
 }
 \subfloat[b][OPT vs. Greedy for $f_{Pa}(\cdot)$]{
 \includegraphics[width=0.3\linewidth]{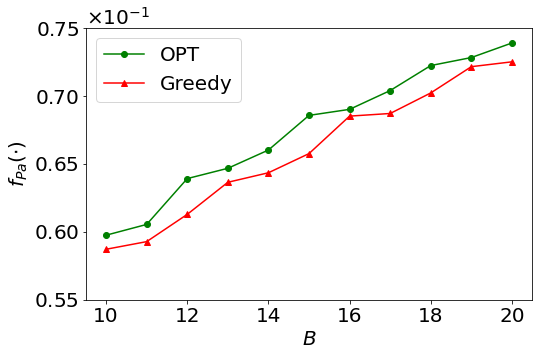}}
 \subfloat[c][Bound on $\gamma_1$]{
 \includegraphics[width=0.3\linewidth]{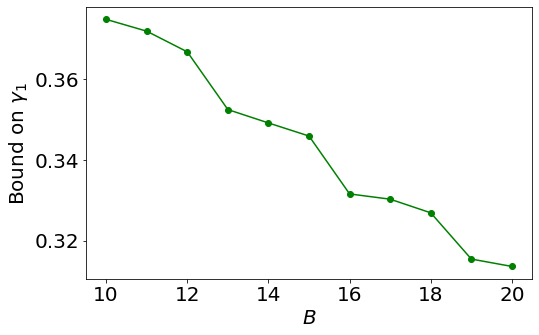}}\\
 \caption{Results for PEMS instances of medium size.}
 \label{fig:small instances}
 \end{figure}
  
We then investigate the performance of the greedy algorithm for PEMS instances of a larger size. Since the optimal solution to the PEMS instances cannot be efficiently obtained when the sizes of the instances become large, we obtain the lower bound on the submodularity ratio $\gamma_1$ of $f_{Pa}(\cdot)$ provided in Lemma~$\ref{lemma:lower bound on gamma_1}$, which can be computed in polynomial time. Different from the smaller instances constructed above, we set $t_1=1$ and $t_2=5$. We let $\zeta_i=\eta_i=10$ for all $i\in\mathcal{V}$ in $\mathcal{C}_{k,i}=\{\zeta c_{k,i}:\zeta\in(\{0\}\cup[\zeta_i])\}$ and $\mathcal{B}_{k,i}=\{\eta b_{k,i}:\eta\in(\{0\}\cup[\eta_i])\}$, where we also set $c_{k,i}=b_{k,i}$ and draw $c_{k,i}$ and $b_{k,i}$ uniformly randomly from $\{1,2,3\}$, for all $k\in[5]$ and for all $i\in\mathcal{V}$. Moreover, we modify the parameter of the Beta distribution corresponding to the pdf of $\beta$ to be $\alpha_1=8$ and $\alpha_2=3$. Here, we can choose to perform $0$, $100$, $200$, ..., or $1000$ virus (or antibody) tests at a node $i\in\mathcal{V}$ and at $k\in[5]$. In Fig.~$\ref{fig:large instances}(a)$, we plot the lower bound on $\gamma_1$ obtained from the PEMS instances constructed above.  We note that the submodularity ratio $\gamma_2$ of $f_{Pa}(\cdot)$ is also always greater than one. Hence, Theorem~$\ref{thm:greedy guarantee f_Pa}$ yields a $\frac{1}{2}(1-e^{-\gamma_1})$ worst-case approximation guarantee for the greedy algorithm. We plot in Fig.~$\ref{fig:large instances}(b)$ the approximation guarantee using the lower bound that we obtained on $\gamma_1$.

 \begin{figure}[htbp]
 \center 
 \subfloat[a][Bound on $\gamma_1$]{
 \includegraphics[width=0.37\linewidth]{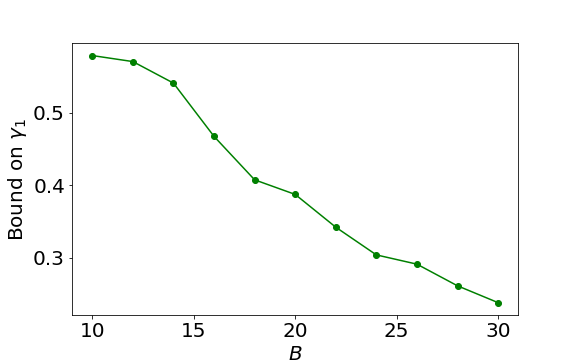}
 }
 \subfloat[b][Approx. guarantee of greedy]{
 \includegraphics[width=0.33\linewidth]{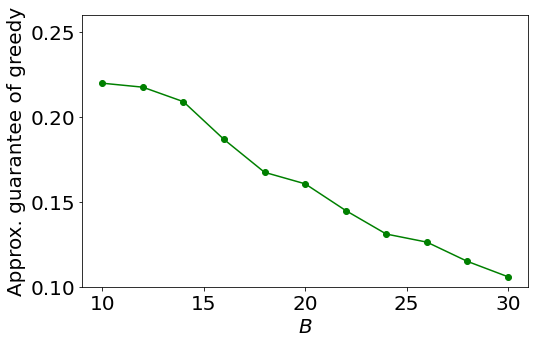}}\\
 \caption{Results for PEMS instances of large size.}
 \label{fig:large instances}
 \end{figure}

\section{Conclusion}
We first considered the PIMS problem under the exact measurement setting, and showed that the problem is NP-hard. We then proposed an approximation algorithm that returns a solution to the PIMS problem that is within a certain factor of the optimal one. Next, we studied the PEMS problem under the noisy measurement setting. Again, we showed that the problem is NP-hard. We applied a greedy algorithm to solve the PEMS problem, and provided performance guarantees on the greedy algorithm. We presented numerical examples to validate the obtained performance bounds of the greedy algorithm, and showed that the greedy algorithm performs well in practice. 

\section{Appendix} 
\subsection{Proof of Lemma~$\ref{lemma:propagate of initial condition}$}\label{sec:proof of results about SIR}
We first prove part~$(a)$. Considering any $i\in\mathcal{V}$ and any $k\in\mathbb{Z}_{\ge0}$, we note from Eq.~\eqref{eqn:SIR S} that
\begin{equation}
\label{eqn:positive of s}
s_i[k+1]=s_i[k](1-h\beta\sum_{j\in\bar{\mathcal{N}}_i}a_{ij}x_j[k]).
\end{equation}
Under Assumptions~$\ref{ass:initial condition}$-$\ref{ass:model parameters}$, we have $x_i[k]\in[0,1]$ for all $i\in\mathcal{V}$ as argued in Section~$\ref{sec:preliminaries}$, and $h\beta\sum_{j\in\bar{\mathcal{N}}_i}a_{ij}<1$ for all $i\in\mathcal{V}$, which implies $1-h\beta\sum_{j\in\bar{\mathcal{N}}_i}a_{ij}x_j[k]\ge 1-h\beta\sum_{j\in\bar{\mathcal{N}}_i}a_{ij}>0$. Supposing $s_i[k]>0$, we have from Eq.~\eqref{eqn:positive of s} $s_i[k+1]>0$. Combining the above arguments with the fact $s_i[0]\in(0,1]$ from Assumption~$\ref{ass:initial condition}$, we see that $s_i[k]>0$ for all $k\in\mathbb{Z}_{\ge0}$. Noting that $s_i[k],x_i[k],r_i[k]\in[0,1]$ with $s_i[k]+x_i[k]+r_i[k]=1$ for all $i\in\mathcal{V}$ and for all $k\in\mathbb{Z}_{\ge0}$ as argued in Section~$\ref{sec:preliminaries}$ and that $x_i[0]\in[0,1)$ and $r_i[0]=0$ for all $i\in\mathcal{V}$, the result in part~$(a)$ also implies $x_i[k],r_i[k]\in[0,1)$ for all $i\in\mathcal{V}$ and for all $k\in\mathbb{Z}_{\ge0}$.

One can then observe that in order to prove parts~$(b)$-$(d)$, it is sufficient to prove the following facts.
\begin{fact}
\label{fact:nonzero state x}
Consider any $i\in\mathcal{V}$ and any $k_1\in\mathbb{Z}_{\ge0}$. If $x_i[k_1]>0$, then $x_i[k_2]>0$ for all $k_2\in\mathbb{Z}_{\ge0}$ with $k_2\ge k_1$.
\end{fact}
\begin{fact}
\label{fact:switch point of zero state}
Consider any $i\in\mathcal{V}$ and any $k\in\mathbb{Z}_{\ge0}$ such that $x_i[k]=0$. If there exists $j\in\mathcal{N}_i$ such that $x_j[k]>0$, then $x_i[k+1]>0$. If $x_j[k]=0$ for all $j\in\mathcal{N}_i$, then $x_i[k+1]=0$.
\end{fact}
\begin{fact}
\label{fact:nonzero state r}
Consider any $i\in\mathcal{V}$ and any $k_1\in\mathbb{Z}_{\ge0}$. If $x_i[k_1]>0$, then $r_i[k_1+1]>0$. If $x_i[k_1]=0$, then $r_i[k_1+1]=0$.
\end{fact}

Let us first prove Fact~$\ref{fact:nonzero state x}$. Consider any $i\in\mathcal{V}$ and any $k\in\mathbb{Z}_{\ge0}$. Supposing $x_i[k]>0$, we have from Eq.~\eqref{eqn:SIR single node}
\begin{equation}
\label{eqn:positive time step k+1}
x_i[k+1]=(1-h\delta)x_i[k]+hs_i[k]\beta\sum_{j\in\bar{\mathcal{N}}_i}a_{ij}x_j[k],
\end{equation}
where the first term on the right-hand side of the above equation is positive, since $1-h\delta>0$ from Assumption~$\ref{ass:model parameters}$, and the second term on the right-hand side of the above equation is nonnegative. It then follows that $x_i[k+1]>0$. Repeating the above argument proves Fact~$\ref{fact:nonzero state x}$.

We next prove Fact~$\ref{fact:switch point of zero state}$. Considering any $i\in\mathcal{V}$ and any $k\in\mathbb{Z}_{\ge0}$ such that $x_i[k]=0$, we note from Eq.~\eqref{eqn:SIR single node} that 
\begin{equation}
\label{eqn:state at time k+1}
x_i[k+1]=hs_i[k]\beta\sum_{j\in\mathcal{N}_i}a_{ij}x_j[k],
\end{equation}
where $s_i[k]>0$ as shown in part~$(a)$. Suppose there exists $j\in\mathcal{N}_i$ such that $x_j[k]>0$. Since $h,\beta\in\mathbb{R}_{>0}$ and $a_{ij}>0$ for all $j\in\mathcal{N}_i$ from Assumption~$\ref{ass:model parameters}$, we have from Eq.~\eqref{eqn:state at time k+1} $x_i[k+1]>0$. Next, supposing $x_j[k]=0$ for all $j\in\mathcal{N}_i$, we obtain from Eq.~\eqref{eqn:state at time k+1} $x_i[k+1]=0$. This proves Fact~$\ref{fact:switch point of zero state}$.

Finally, we prove Fact~$\ref{fact:nonzero state r}$. Let us consider any $i\in\mathcal{V}$ and any $k_1\in\mathbb{Z}_{\ge0}$. Suppose $x_i[k_1]>0$. Since $h,\delta\in\mathbb{R}_{>0}$ from Assumption~$\ref{ass:model parameters}$, we have from Eq.~\eqref{eqn:SIR R} $r_i[k_1+1]=r_i[k]+h\delta x_i[k_1]>0$. Next, supposing $x_i[k_1]=0$, we note from Fact~$\ref{fact:nonzero state x}$ that $x_i[k_1^{\prime}]=0$ for all $k_1^{\prime}\le k_1$. It then follows from Eq.~\eqref{eqn:SIR R} and Assumption~$\ref{ass:initial condition}$ that $r_i[k_1+1]=r_i[k_1]=\cdots=r_i[0]=0$, completing the proof of Fact~$\ref{fact:nonzero state r}$.$\hfill\square$

\subsection{Proof of Theorem~$\ref{thm:PIMS NP-hard}$}\label{sec:proof of PIMS NP-hard}
We show that PIMS is NP-hard via a polynomial reduction from the exact cover by 3-sets (X3C) problem which is known to be NP-complete \cite{garey1979computers}. 
\begin{problem}
\label{pro:X3C instance}
Consider $\mathcal{X}=\{1,2,\dots,3m\}$ and a collection $\mathcal{Z}=\{z_1,z_2,\dots,z_{\tau}\}$ of 3-element subsets of $\mathcal{X}$, where $\tau\ge m$. The X3C problem is to determine if there is an exact cover for $\mathcal{X}$, i.e., a subcollection $\mathcal{Z}^{\prime}\subseteq\mathcal{Z}$ such that every element of $\mathcal{X}$ occurs in exactly one member of $\mathcal{Z}^{\prime}$.
\end{problem}

\begin{figure}[hp]
\centering
\captionsetup{justification=centering}
\includegraphics[width=0.47\linewidth]{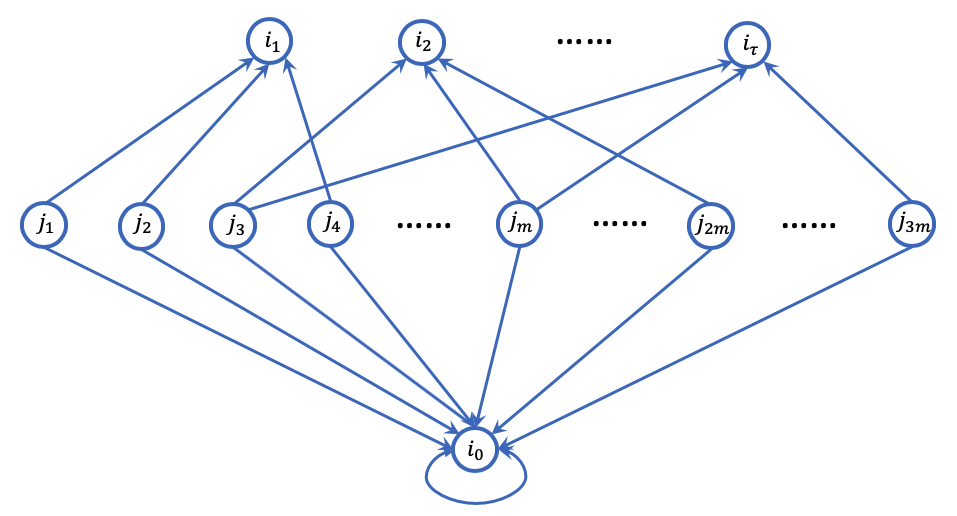}
\center
\caption{Graph $\mathcal{G}=\{\mathcal{V},\mathcal{E}\}$ constructed in the proof of Theorem~$\ref{thm:PIMS NP-hard}$.}
\label{fig:PIMS plot}
\end{figure}

Consider an instance of the X3C problem given by a set $\mathcal{X}=\{1,\dots,3m\}$ and a collection $\mathcal{Z}=\{z_1,\dots,z_{\tau}\}$ of 3-element subsets of $\mathcal{X}$, where $\tau\ge m$. We then construct an instance of the PIMS problem as follows. The node set of the graph $\mathcal{G}=\{\mathcal{V},\mathcal{E}\}$ is set to be $\mathcal{V}=\{i_0,i_1,\dots,i_{\tau}\}\cup\{j_1,j_2,\dots,j_{3m}\}$. The edge set of $\mathcal{G}=\{\mathcal{V},\mathcal{E}\}$ is set to satisfy that $(j_q,i_l)\in\mathcal{E}$ if $q\in\mathcal{X}$ is contained in $z_l\in\mathcal{Z}$, $(j_q,i_0)\in\mathcal{E}$ for all $q\in\mathcal{X}$, and $(i_0,i_0)\in\mathcal{E}$. Note that based on the construction of $\mathcal{G}=\{\mathcal{V},\mathcal{E}\}$, each node $i\in\{i_1,\dots,i_{\tau}\}$ represents a subset from $\mathcal{Z}$ in the X3C instance, and each node $j\in\{j_1,\dots,j_{3m}\}$ represents an element from $\mathcal{X}$ in the X3C instance, where the edges between $\{i_1,\dots,i_{\tau}\}$ and $\{j_1,\dots,j_{3m}\}$ indicate how the elements in $\mathcal{X}$ are included in the subsets in $\mathcal{Z}$. A plot of $\mathcal{G}=\{\mathcal{V},\mathcal{E}\}$ is given in Fig.~\ref{fig:PIMS plot}. Accordingly, the weight matrix $A\in\mathbb{R}^{(3m+\tau+1)\times(3m+\tau+1)}$ is set to satisfy that $a_{i_lj_q}=1$ if $q\in\mathcal{X}$ is contained in $z_l\in\mathcal{Z}$, $a_{i_0j_q}=1$ for all $q\in\mathcal{X}$, and $a_{i_0i_0}=1$. We set the sampling parameter to be $h=1/(3m+1)$. The set $\mathcal{S}_I\subseteq\mathcal{V}$ is set to be $\mathcal{S}_I=\mathcal{V}$, i.e., $x_i[0]>0$ for all $i\in\mathcal{V}$. We set time steps $t_1=2$ and $t_2=3$. Finally, we set $b_{2,i}=b_{3,i}=0$ for all $i\in\mathcal{V}$, $c_{2,i_l}=1$ and $c_{3,i_l}=0$ for all $l\in\{1,\dots,\tau\}$, $c_{2,j_q}=c_{3,j_q}=m+1$ for all $q\in\mathcal{X}$, and $c_{2,i_0}=c_{3,i_0}=0$.  Since $x_i[0]>0$ for all $i\in\mathcal{V}$, we see from Lemma~$\ref{lemma:propagate of initial condition}$ that $x_i[k]>0$ and $r_i[k]>0$ for all $i\in\mathcal{V}$ and for all $k\in\{2,3\}$. Therefore, Lemma~$\ref{lemma:propagate of initial condition}$ is no longer useful in determining the coefficients in the equations from Eq.~\eqref{eqn:rewrite SIR matrix}.

We claim that an optimal solution, denoted as $\mathcal{I}^{\star}$, to the constructed PIMS instance satisfies $c(\mathcal{I}^{\star})\le m$ if and only if the solution to the X3C instance is ``yes''.

First, suppose the solution to the X3C instance is ``yes''. Denote an exact cover as $\mathcal{Z}^{\prime}=\{z_{q_1},\dots,z_{q_m}\}\subseteq\mathcal{Z}$, where $\{q_1,\dots,q_m\}\subseteq\{1,\dots,\tau\}$. Let us consider a measurement selection strategy $\mathcal{I}_0\subseteq\mathcal{I}_{t_1:t_2}$ given by
\begin{equation*}
\mathcal{I}_0=\big(\bigcup_{l\in\{1,\dots,m\}}\{x_{i_{q_l}}[2],x_{i_{q_l}}[3],r_{i_{q_l}}[2]\}\big)\cup\{x_{i_0}[2],x_{i_0}[3],r_{i_0}[2],r_{i_0}[3]\}.
\end{equation*}
We then have from Eq.~\eqref{eqn:def of I bar} $\bar{\mathcal{I}}_0=\{(2,i_0,r),(2,i_0,x)\}\cup\{(2,i_{q_l},x):l\in\{1,\dots,m\}\}$. Noting that $s_i[k]>0$ for all $i\in\mathcal{V}$ and for all $k\in\mathbb{Z}_{\ge0}$ from Lemma~$\ref{lemma:propagate of initial condition}(a)$, we consider the following $(m+1)$ equations from Eq.~\eqref{eqn:rewrite SIR matrix} whose indices are contained in $\bar{\mathcal{I}}_0$:
\begin{align}
&\frac{1}{s_{i_0}[2]}(x_{i_0}[3]-x_{i_0}[2])=h\begin{bmatrix}x_{i_0}[2]+\sum_{w\in\mathcal{N}_{i_0}}x_w[2] & -\frac{x_{i_0}[2]}{s_{i_0}[1]} \end{bmatrix}\begin{bmatrix}\beta\\ \delta\end{bmatrix},\label{eqn:base eq}\\
&\frac{1}{s_{i_{q_l}}[2]}(x_{i_{q_l}}[3]-x_{i_{q_l}}[2])=h\begin{bmatrix}\sum_{w\in\mathcal{N}_{i_{q_l}}}x_w[2] & -\frac{x_{i_{q_l}}[2]}{s_{i_{q_l}}[2]} \end{bmatrix}\begin{bmatrix}\beta\\ \delta\end{bmatrix},\ \forall l\in\{1,\dots,m\},\label{eqn:eqs for exact cover}
\end{align}
where we note $\mathcal{N}_{i_0}=\{j_1,\dots,j_{3m}\}$ from the way we constructed $\mathcal{G}=\{\mathcal{V},\mathcal{E}\}$.  Since $\mathcal{Z}^{\prime}=\{z_{q_1},\dots,z_{q_m}\}$ is an exact cover for $\mathcal{X}$, we see from the construction of $\mathcal{G}=\{\mathcal{V},\mathcal{E}\}$ that $\bigcup_{l\in\{1,\dots,m\}}\mathcal{N}_{i_{q_l}}$ is a union of mutually disjoint (3-element) sets such that $\bigcup_{l\in\{1,\dots,m\}}\mathcal{N}_{i_{q_l}}=\{j_1,\dots,j_{3m}\}$. Thus, subtracting the equations in \eqref{eqn:eqs for exact cover} from Eq.~\eqref{eqn:base eq}, we obtain
\begin{multline}
\label{eqn:base eq 2nd}
\frac{1}{s_{i_0}[2]}(x_{i_0}[3]-x_{i_0}[2])-\sum_{l\in\{1,\dots,m\}}\frac{1}{s_{i_{q_l}}[2]}(x_{i_{q_l}}[3]-x_{i_{q_l}}[2])\\=
h\begin{bmatrix}x_{i_0}[2] & -\frac{x_{i_0}[2]}{s_{i_0}[2]}+\sum_{l\in\{1,\dots,m\}}\frac{x_{i_{q_l}}[2]}{s_{i_{q_l}}[2]} \end{bmatrix}\begin{bmatrix}\beta\\ \delta\end{bmatrix},
\end{multline}
where we note $x_{i_0}[2]>0$ as argued above. Following Definition~$\ref{def:Phi I}$, we stack coefficient matrices $\Phi_{2,i_0}^r\in\mathbb{R}^{1\times2}$, $\Phi_{2,i_0}^x\in\mathbb{R}^{1\times2}$ and $\Phi_{2,i_{q_l}}^x\in\mathbb{R}^{1\times2}$ for all $l\in\{1,\dots,m\}$ into a matrix $\Phi(\mathcal{I}_0)\in\mathbb{R}^{(m+2)\times2}$, where $\Phi_{k,i}^r$ and $\Phi_{k,i}^x$ are defined in \eqref{eqn:def of rows}. Now, considering the matrix:
\begin{equation}
\Phi_0=\begin{bmatrix}
x_{i_0}[2] & -\frac{x_{i_0}[2]}{s_{i_0}[2]}+\sum_{l\in\{1,\dots,m\}}\frac{x_{i_{q_l}}[2]}{s_{i_{q_l}}[2]}\\
0 & x_{i_0}[2]
\end{bmatrix},
\end{equation}
we see from the above arguments that $(\Phi_0)_1$ and $(\Phi_0)_2$ can be be obtained via algebraic operations among the rows in $\Phi(\mathcal{I}_0)$, and the elements in $(\Phi_0)_1$ and $(\Phi_0)_2$ can be determined using the measurements from $\mathcal{I}_0$. Therefore, we have $\Phi_0\in\tilde{\Phi}(\mathcal{I}_0)$, where $\tilde{\Phi}(\mathcal{I}_0)$ is defined in Definition~$\ref{def:Phi I}$. Noting that $x_{i_0}[2]>0$, we have $\rank(\Phi_0)=2$, which implies $r_{\mathop{\max}}(\mathcal{I}_0)=2$, where $r_{\mathop{\max}}(\mathcal{I}_0)$ is given by Eq.~\eqref{eqn:max rank of Phi}. Thus, $\mathcal{I}_0\subseteq\mathcal{I}_{t_1:t_2}$ satisfies the constraint in \eqref{eqn:PIMS obj}. Since $c(\mathcal{I}_0)=m$ from the way we set the costs of collecting measurements in the PIMS instance, we have $c(\mathcal{I}^{\star})\le m$.

Conversely, suppose the solution to the X3C instance is ``no'', i.e., for any subcollection $\mathcal{Z}^{\prime}\subseteq\mathcal{Z}$ that contains $m$ subsets, there exists at least one element in $\mathcal{X}$ that is not contained in any subset in $\mathcal{Z}^{\prime}$. We will show that for any measurement selection strategy $\mathcal{I}\subseteq\mathcal{I}_{t_1:t_2}$ that satisfies $r_{\mathop{\max}}(\mathcal{I})=2$, $c(\mathcal{I})>m$ holds. Equivalently, we will show that for any $\mathcal{I}\subseteq\mathcal{I}_{t_1:t_2}$ with $c(\mathcal{I})\le m$, $r_{\mathop{\max}}(\mathcal{I})=2$ does not hold. Consider any $\mathcal{I}\subseteq\mathcal{I}_{t_1:t_2}$ such that $c(\mathcal{I})\le m$. Noting that $c_{2,j_q}=c_{3,j_q}=m+1$ for all $q\in\mathcal{X}$ in the constructed PIMS instance, we have $x_{j_q}[2]\notin\mathcal{I}$ and $x_{j_q}[3]\notin\mathcal{I}$ for all $q\in\mathcal{X}$. Moreover, we see that $\mathcal{I}$ contains at most $m$ measurements from $\{x_{i_1}[2],\dots,x_{i_{\tau}}[2]\}$. To proceed, let us consider any $\mathcal{I}_1\subseteq\mathcal{I}_{t_1:t_2}$ such that
\begin{equation}
\label{eqn:I if X3C no}
\mathcal{I}_1=\{x_{i_0}[2],x_{i_{v_1}}[2],\dots,x_{i_{v_m}}[2]\}\cup\big(\bigcup_{l\in\{0,\dots,\tau\}}\{x_{i_l}[3]\}\big)\cup\big(\bigcup_{i\in\mathcal{V}}\{r_i[2],r_i[3]\}\big),
\end{equation}
where $\{v_1,\dots,v_m\}\subseteq\{1,\dots,\tau\}$. In other words, $\mathcal{I}_1\subseteq\mathcal{I}_{t_1:t_2}$ contains $m$ measurements from $\{x_{i_1}[2],\dots,x_{i_{\tau}}[2]\}$ and all the other measurements from $\mathcal{I}_{t_1:t_2}$ that have zero costs. It follows that $c(\mathcal{I}_1)=m$. Similarly to \eqref{eqn:base eq} and \eqref{eqn:eqs for exact cover}, we have the following $(m+1)$ equations from Eq.~\eqref{eqn:rewrite SIR matrix} whose indices are contained in $\bar{\mathcal{I}}_1$ (given by Eq.~\eqref{eqn:def of I bar}):
\begin{align}
&\frac{1}{s_{i_0}[2]}(x_{i_0}[3]-x_{i_0}[2])=h\begin{bmatrix}x_{i_0}[2]+\sum_{w\in\mathcal{N}_{i_0}}x_w[2] & -\frac{x_{i_0}[2]}{s_{i_0}[1]} \end{bmatrix}\begin{bmatrix}\beta\\ \delta\end{bmatrix},\label{eqn:base eq X3C no}\\
&\frac{1}{s_{i_{v_l}}[2]}(x_{i_{v_l}}[3]-x_{i_{v_l}}[2])=h\begin{bmatrix}\sum_{w\in\mathcal{N}_{i_{v_l}}}x_w[2] & -\frac{x_{i_{v_l}}[2]}{s_{i_{v_l}}[2]} \end{bmatrix}\begin{bmatrix}\beta\\ \delta\end{bmatrix}, \forall l\in\{1,\dots,m\}.\label{eqn:eqs for X3C no}
\end{align}
Noting that for any subcollection $\mathcal{Z}^{\prime}\subseteq\mathcal{Z}$ that contains $m$ subsets, there exists at least one element in $\mathcal{X}$ that is not contained in any subset in $\mathcal{Z}^{\prime}$ as we argued above, we see that there exists at least one element in $\mathcal{X}$ that is not contained in any subset in $\{z_{v_1},\dots,z_{v_m}\}$. It then follows from the way we constructed $\mathcal{G}=\{\mathcal{V},\mathcal{E}\}$ that there exists  $w^{\prime}\in\mathcal{N}_{i_0}$ such that $w^{\prime}\notin\mathcal{N}_{i_{v_l}}$ for all $l\in\{1,\dots,m\}$. Thus, by subtracting the equations in \eqref{eqn:eqs for X3C no} (multiplied by any constants resp.) from Eq.~\eqref{eqn:base eq X3C no}, $x_{w^{\prime}}[2]$ will remain on the right-hand side of the equation in \eqref{eqn:base eq X3C no}. Similarly, consider any equation from \eqref{eqn:eqs for X3C no} indexed by $(2,i_{v_l},x)\in\bar{\mathcal{I}}_1$, where $l\in\{1,\dots,m\}$. First, suppose we subtract Eq.~\eqref{eqn:base eq X3C no} multiplied by some positive constant and any equations in \eqref{eqn:eqs for X3C no} other than equation $(2,i_{v_l},x)$ (multiplied by any constants resp.) from equation $(2,i_{v_l},x)$. Since there exists  $w^{\prime}\in\mathcal{N}_{i_0}$ such that $w^{\prime}\notin\mathcal{N}_{i_{v_l}}$ for all $l\in\{1,\dots,m\}$ as argued above, we see that $x_{w^{\prime}}[2]$ will appear on the right-hand side of equation $(2,i_{v_l},x)$. Next, suppose we subtract any equations in \eqref{eqn:eqs for X3C no} other than equation $(2,i_{v_l},x)$ (multiplied by any constants resp.) from equation $(2,i_{v_l},x)$. One can check that either of the following two cases holds for the resulting equation $(2,i_{v_l},x)$: (a) the coefficients on the right-hand side of equation $(2,i_{v_l},x)$ contain $x_{j_q}[2]\notin\mathcal{I}_1$, where $q\in\mathcal{X}$; or (b) the coefficient matrix on the right-hand side of equation $(2,i_{v_l},x)$ is of the form $\begin{bmatrix}0 & \star\end{bmatrix}$.  Again, we stack $\Phi_{k,i}^r\in\mathbb{R}^{1\times2}$ for all $(k,i,r)\in\bar{\mathcal{I}}_1$ and $\Phi_{k,i}^x\in\mathbb{R}^{1\times 2}$ for all $(k,i,x)\in\bar{\mathcal{I}}_1$ into a matrix $\Phi(\mathcal{I}_1)$, where we note that $\Phi_{k,i}^r$ is of the form $\begin{bmatrix}0 & \star\end{bmatrix}$ for all $(k,i,r)\in\bar{\mathcal{I}}_1$. One can then see from the above arguments that for all $\Phi\in\mathbb{R}^{2\times2}$ (if they exist) such that $(\Phi)_1$ and $(\Phi)_2$ can be obtained from algebraic operations among the rows in $\Phi(\mathcal{I}_1)$, and the elements in $(\Phi)_1$ and $(\Phi)_2$ can be determined using the measurements from $\mathcal{I}_1$, $\rank(\Phi)\le1$ holds. It follows that $r_{\mathop{\max}}(\mathcal{I}_1)<2$, i.e., constraint $r_{\mathop{\max}}(\mathcal{I}_1)=2$ in \eqref{eqn:PIMS obj} does not hold. Using similar arguments to those above, one can further show that $r_{\mathop{\max}}(\mathcal{I})<2$ holds for all $c(\mathcal{I})\le m$, completing the proof of the converse direction of the above claim.

Hence, it follows directly from the above arguments that an algorithm for the PIMS problem can also be used to solve the X3C problem. Since X3C is NP-complete, we conclude that the PIMS problem is NP-hard.$\hfill\square$

\subsection{Proof of Theorem~$\ref{thm:PEMS NP-hard}$}\label{sec:proof of PEMS NP-hard}
We prove the NP-hardness of the PEMS problem via a polynomial reduction from the knapsack problem which is known to be NP-hard \cite{garey1979computers}. An instance of the knapsack problem is given by a set $D=\{d_1,\dots,d_{\tau}\}$, a size $s(d)\in\mathbb{Z}_{>0}$ and a value $v(d)\in\mathbb{Z}_{>0}$ for each $d\in D$, and $K\in\mathbb{Z}_{>0}$. The knapsack problem is to find $D^{\prime}\subseteq D$ such that $\sum_{d\in D^{\prime}}v(d)$ is maximized while satisfying $\sum_{d\in D^{\prime}}s(d)\le K$.

Given any knapsack instance, we construct an instance of the PEMS problem as follows. Let $\mathcal{G}=\{\mathcal{V},\mathcal{E}\}$ be a graph that contains a set of $n$ isolated nodes with $n=\tau$ and $\mathcal{V}=[n]$. Set the weight matrix to be $A=\mathbf{0}_{n\times n}$, and set the sampling parameter as $h=1$. The time steps $t_1$ and $t_2$ are set to be $t_1=t_2=1$, i.e., only the measurements of $x_i[1]$ and $r_i[1]$ for all $i\in\mathcal{V}$ will be considered. The initial condition is set to satisfy $s_i[0]=0.5$, $x_i[0]=0.5$ and $r_i[0]=0$ for all $i\in\mathcal{V}$. The budget constraint is set as $B=K$. Let $\mathcal{C}_{1,i}=\{0,B+1\}$ and $\mathcal{B}_{1,i}=\{0,s(d_i)\}$ for all $i\in\mathcal{V}$. The pmfs of measurements $\hat{x}_i[1]$ and $\hat{r}_i[1]$ are given by Eqs.~\eqref{eqn:pmf of xi} and \eqref{eqn:pmf of ri}, respectively, with $N_i^x=N_i^r=v(d_i)$ and $N_i=\mathop{\max}_{i\in\mathcal{V}}v(d_i)$ for all $i\in\mathcal{V}$, where Assumption~$\ref{ass:white and independent noise}$ is assumed to hold. Finally, let the prior pdf of $\beta\in(0,1)$ be a Beta distribution with parameters $\alpha_1=3$ and $\alpha_2=3$, and let the prior pdf of $\delta\in(0,1)$ also be a Beta distribution with parameters $\alpha_1=3$ and $\alpha_2=3$, where we take $\beta$ and $\delta$ to be independent. Noting that $\mathcal{C}_{1,i}=\{0,B+1\}$ in the PEMS instance constructed above, i.e., $\hat{x}_i[k]$ incurs a cost of $B+1>B$, we only need to consider measurements $\hat{r}_i[1]$ for all $i\in\mathcal{V}$. Moreover, since $\mathcal{B}_{1,i}=\{0,s(d_i)\}$, a corresponding measurement selection is then given by $\mu\in\{0,1\}^{\mathcal{V}}$. In other words, $\mu(i)=1$ if measurement $\hat{r}_i[1]$ is collected (with cost $s(d_i)$), and $\mu(i)=0$ if measurement $\hat{r}_i[1]$ is not collected. We will see that there is a one to one correspondence between a measurement $\hat{r}_i[1]$ in the PEMS instance and an element $d_i\in D$ in the knapsack instance.

Consider a measurement selection $\mu\in\{0,1\}^{\mathcal{V}}$. Since $r_i[0]=0$ and $x_i[0]=0.5$ for all $i\in\mathcal{V}$, Eq.~\eqref{eqn:SIR R} implies $r_i[1]=0.5h\delta$ for all $i\in\mathcal{V}$, where $h=1$. One then obtain from Eq.~\eqref{eqn:FIM 2nd exp} and Eq.~\eqref{eqn:ln pmf xi 2} the following:
\begin{equation}
F_{\theta}(\mu)=\frac{1}{0.5\delta(1-0.5\delta)}\begin{bmatrix}0 & 0\\ 0 & 0.25\end{bmatrix}\sum_{i\in\text{supp}(\mu)}N_i^r\mu(i).\label{eqn:FIM for hard proof 2}
\end{equation}
Next, noting that $\beta$ and $\delta$ are independent, one can show via Eq.~\eqref{eqn:def of F_p} that $F_p\in\mathbb{R}^{2\times 2}$ is diagonal, where one can further show that $(F_p)_{11}=(F_p)_{22}>0$ using the fact that the pdfs of $\beta$ and $\delta$ are Beta distributions with parameters $\alpha_1=3$ and $\alpha_2=3$. Similarly, one can obtain $\mathbb{E}_{\theta}[1/0.5\delta(1-0.5\delta)]>0$. It now follows from Eq.~\eqref{eqn:FIM for hard proof 2} that
\begin{equation}
\label{eqn:CRLB hard proof}
\mathbb{E}_{\theta}[F_{\theta}(\mu)]+F_p=\begin{bmatrix}z_1 & 0\\ 0 & z_1+z_2\sum_{i\in\text{supp}(\mu)}N_i^r\mu(i)\end{bmatrix},
\end{equation}
where $z_1,z_2\in\mathbb{R}_{>0}$ are some constants (independent of $\mu$). Note that the objective in the PEMS instance is given by $\mathop{\min}_{\mu\in\{0,1\}^{\mathcal{V}}}f(\mu)$, where $f(\cdot)\in\{f_a(\cdot),f_d(\cdot)\}$. First, considering the objective function $f_a(\mu)=\Tr(\bar{C}(\mu))$, where $\bar{C}(\mu)=(\mathbb{E}_{\theta}[F_{\theta}(\mu)]+F_p)^{-1}$, we see from Eq.~\eqref{eqn:CRLB hard proof} that $\Tr(\bar{C}(\mu))$ is minimized (over $\mu\in\{0,1\}^{\mathcal{V}}$) if and only if $\sum_{i\in\text{supp}(\mu)}N_i^r\mu(i)$ is maximized. Similar arguments hold for the objective function $f_d(\mu)=\ln\det(\bar{C}(\mu))$. It then follows directly from the above arguments that a measurement selection $\mu^{\star}\in\{0,1\}^{\mathcal{V}}$ is an optimal solution to the PEMS instance if and only if $D^{\star}\triangleq\{d_i:i\in\text{supp}(\mu^{\star})\}$ is an optimal solution to the knapsack instance. Since the knapsack problem is NP-hard, the PEMS problem is NP-hard.$\hfill\square$

\subsection{Proof of Lemma~$\ref{lemma:lower bound on gamma_1}$}\label{sec:proof of lower bound}
Noting the definition of $\gamma_1$ in Definition~$\ref{def:submodularity ratios}$, we provide a lower bound on $\frac{\sum_{y\in \mathcal{A}\setminus\mathcal{Y}_2^j}(f_{Pa}(\{y\}\cup\mathcal{Y}_2^j)-f_{Pa}(\mathcal{Y}_2^j))}{f_{Pa}(\mathcal{A}\cup\mathcal{Y}_2^j)-f_{Pa}(\mathcal{Y}_2^j)}$ for all $\mathcal{A}\subseteq\bar{\mathcal{M}}$ and for all $\mathcal{Y}_2^j$, where we assume that $\mathcal{A}\setminus\mathcal{Y}_2^j\neq\emptyset$, otherwise \eqref{eqn:sbmr 1} would be satisfied for all $\gamma_1\in\mathbb{R}$. Recalling the definition of $f_{Pa}(\cdot)$ in \eqref{eqn:def of f_Pa}, we lower bound $LHS\triangleq\sum_{y\in \mathcal{A}\setminus\mathcal{Y}_2^j}\big(f_{Pa}(\{y\}\cup\mathcal{Y}_2^j)-f_{Pa}(\mathcal{Y}_2^j)\big)$ in the following manner:
\begin{align}\nonumber
\qquad\qquad LHS&=\sum_{y\in\mathcal{A}\setminus\mathcal{Y}_2^j}\sum_{i=1}^2\frac{\lambda_i(F_p+H(\{y\}\cup\mathcal{Y}_2^j))-\lambda_i(F_p+H(\mathcal{Y}_2^j))}{\lambda_i(F_p+H(\mathcal{Y}_2^j))\lambda_i(F_p+H(\{y\}\cup\mathcal{Y}_2^j))}\\
&\ge\sum_{y\in\mathcal{A}\setminus\mathcal{Y}_2^j}\frac{\sum_{i=1}^2(\lambda_i(F_p+H(\{y\}\cup\mathcal{Y}_2^j))-\lambda_i(F_p+H(\mathcal{Y}_2^j)))}{\lambda_1(F_p+H(\mathcal{Y}_2^j))\lambda_1(F_p+H(\{z^{\prime}\}\cup\mathcal{Y}_2^j))}\label{eqn:gamma_1 derive 0}\\
&=\frac{\sum_{y\in\mathcal{A}\setminus\mathcal{Y}_2^j}\Tr(H_y)}{\lambda_1(F_p+H(\mathcal{Y}_2^j))\lambda_1(F_p+H(\{z^{\prime}\}\cup\mathcal{Y}_2^j))}.\label{eqn:gamma_1 derive 1}
\end{align}
To obtain \eqref{eqn:gamma_1 derive 0}, we let $z^{\prime}\in\mathop{\arg}{\max}_{y\in\mathcal{A}\setminus\mathcal{Y}_2^j}\lambda_1(F_p+H(\{y\}\cup\mathcal{Y}_2^j))$ and note that $\lambda_1(F_p+H(\{z^{\prime}\}\cup\mathcal{Y}_2^j))\ge\lambda_i(F_p+H(\{y\}\cup\mathcal{Y}_2^j))$ for all $i\in\{1,2\}$ and for all $y\in\mathcal{A}\setminus\mathcal{Y}_2^j$. Next, we upper bound $f_{Pa}(\mathcal{A}\cup\mathcal{Y}_2^j)-f_{Pa}(\mathcal{Y}_2^j)$ in the following manner:
\begin{align}\nonumber
f_{Pa}(\mathcal{A}\cup\mathcal{Y}_2^j)-f_{Pa}(\mathcal{Y}_2^j)&=\sum_{i=1}^2\frac{\lambda_i(F_p+H(\mathcal{A}\cup\mathcal{Y}_2^j))-\lambda_i(F_p+H(\mathcal{Y}_2^j))}{\lambda_i(F_p+H(\mathcal{Y}_2^j))\lambda_i(F_p+H(\mathcal{A}\cup\mathcal{Y}_2^j))}\\
&\le\frac{\sum_{i=1}^2\big(\lambda_i(F_p+H(\mathcal{A}\cup\mathcal{Y}_2^j))-\lambda_i(F_p+H(\mathcal{Y}_2^j)\big)}{\lambda_2(F_p+H(\mathcal{Y}_2^j))\lambda_2(F_p+H(\{z^{\prime}\}\cup\mathcal{Y}_2^j))}\label{eqn:gamma_1 derive 2}\\
&=\frac{\sum_{y\in\mathcal{A}\setminus\mathcal{Y}_2^j}\Tr(H_y)}{\lambda_2(F_p+H(\mathcal{Y}_2^j))\lambda_2(F_p+H(\{z^{\prime}\}\cup\mathcal{Y}_2^j))}.\label{eqn:gamma_1 derive 3}
\end{align}
To obtain \eqref{eqn:gamma_1 derive 2}, we note that $\lambda_i(F_p+H(\mathcal{A}\cup\mathcal{Y}_2^j))\ge\lambda_2(F_p+H(\mathcal{A}\cup\mathcal{Y}_2^j))\ge\lambda_2(F_p+\{z^{\prime}\}\cup\mathcal{Y}_2^j)$ for all $i\in\{1,2\}$, where the second inequality follows from Lemma~$\ref{lemma:eigenvalue ineqs}$ with the fact $H(\mathcal{A}\cup\mathcal{Y}_2^j)-H(\{z^{\prime}\}\cup\mathcal{Y}_2^j)\succeq\mathbf{0}$, and $z^{\prime}$ is defined above. Combining \eqref{eqn:gamma_1 derive 1} and \eqref{eqn:gamma_1 derive 3}, and noting $z_j\in\mathop{\arg}{\min}_{y\in\bar{\mathcal{M}}\setminus\mathcal{Y}_2^j}\frac{\lambda_2(F_p+H(\{y\}\cup\mathcal{Y}_2^j))}{\lambda_1(F_p+H(\{y\}\cup\mathcal{Y}_2^j))}$, we have
\begin{equation}
\frac{\sum_{y\in \mathcal{A}\setminus\mathcal{Y}_2^j}(f_{Pa}(\{y\}\cup\mathcal{Y}_2^j)-f_{Pa}(\mathcal{Y}_2^j))}{f_{Pa}(\mathcal{A}\cup\mathcal{Y}_2^j)-f_{Pa}(\mathcal{Y}_2^j)}\ge\frac{\lambda_2(F_p+H(\mathcal{Y}_2^j))\lambda_2(F_p+H(\{z_j\}\cup\mathcal{Y}_2^j))}{\lambda_1(F_p+H(\mathcal{Y}_2^j))\lambda_1(F_p+H(\{z_j\}\cup\mathcal{Y}_2^j))},\label{eqn:gamma_1 derive 4}
\end{equation}
which implies \eqref{eqn:lower bound on gamma_1}.$\hfill\square$

\bibliographystyle{siamplain}
\bibliography{references}

\begin{thebibliography}{37}
\providecommand{\natexlab}[1]{#1}
\providecommand{\url}[1]{\texttt{#1}}
\expandafter\ifx\csname urlstyle\endcsname\relax
  \providecommand{\doi}[1]{doi: #1}\else
  \providecommand{\doi}{doi: \begingroup \urlstyle{rm}\Url}\fi

\bibitem[cdc()]{cdcviraltest}
{Centers for Disease Control and Prevention} {Test for Current Infection}.
\newblock
  \url{https://www.cdc.gov/coronavirus/2019-ncov/testing/diagnostic-testing.html}.
\newblock 2020.

\bibitem[pro()]{protect}
{Protect Purdue-Purdue University's response to COVID-19}.
\newblock \url{https://protect.purdue.edu}.
\newblock 2020.

\bibitem[Ahn and Hassibi(2013)]{ahn2013global}
H.~J. Ahn and B.~Hassibi.
\newblock Global dynamics of epidemic spread over complex networks.
\newblock In \emph{Proc. Conference on Decision and Control}, pages 4579--4585.
  IEEE, 2013.

\bibitem[Bendavid et~al.(2020)Bendavid, Mulaney, Sood, Shah, Ling,
  Bromley-Dulfano, Lai, Weissberg, Saavedra, Tedrow, et~al.]{bendavid2020covid}
E.~Bendavid, B.~Mulaney, N.~Sood, S.~Shah, E.~Ling, R.~Bromley-Dulfano, C.~Lai,
  Z.~Weissberg, R.~Saavedra, J.~Tedrow, et~al.
\newblock {COVID-19} antibody seroprevalence in {Santa Clara County},
  {California}.
\newblock \emph{MedRxiv}, 2020.

\bibitem[Bian et~al.(2017)Bian, Buhmann, Krause, and
  Tschiatschek]{bian2017guarantees}
A.~A. Bian, J.~M. Buhmann, A.~Krause, and S.~Tschiatschek.
\newblock Guarantees for greedy maximization of non-submodular functions with
  applications.
\newblock In \emph{Proc. International Conference on Machine Learning}, pages
  498--507, 2017.

\bibitem[Chakrabarti et~al.(2008)Chakrabarti, Wang, Wang, Leskovec, and
  Faloutsos]{chakrabarti2008epidemic}
D.~Chakrabarti, Y.~Wang, C.~Wang, J.~Leskovec, and C.~Faloutsos.
\newblock Epidemic thresholds in real networks.
\newblock \emph{ACM Transactions on Information and System Security},
  10\penalty0 (4):\penalty0 1--26, 2008.

\bibitem[Chepuri and Leus(2014)]{chepuri2014sparsity}
S.~P. Chepuri and G.~Leus.
\newblock Sparsity-promoting sensor selection for non-linear measurement
  models.
\newblock \emph{IEEE Transactions on Signal Processing}, 63\penalty0
  (3):\penalty0 684--698, 2014.

\bibitem[Cormen et~al.(2009)Cormen, Leiserson, Rivest, and
  Stein]{cormen2009introduction}
T.~H. Cormen, C.~E. Leiserson, R.~L. Rivest, and C.~Stein.
\newblock \emph{Introduction to algorithms}.
\newblock MIT press, 2009.

\bibitem[Garey and Johnson(1979)]{garey1979computers}
M.~R. Garey and D.~S. Johnson.
\newblock \emph{Computers and intractability}, volume 174.
\newblock Freeman San Francisco, 1979.

\bibitem[Horn and Johnson(2012)]{horn2012matrix}
R.~A. Horn and C.~R. Johnson.
\newblock \emph{Matrix analysis}.
\newblock Cambridge University Press, 2012.

\bibitem[Hota et~al.(2020)Hota, Godbole, Bhariya, and Par{\'e}]{hota2020closed}
A.~R. Hota, J.~Godbole, P.~Bhariya, and P.~E. Par{\'e}.
\newblock A closed-loop framework for inference, prediction and control of
  {SIR} epidemics on networks.
\newblock \emph{arXiv preprint arXiv:2006.16185}, 2020.

\bibitem[Joshi and Boyd(2008)]{joshi2008sensor}
S.~Joshi and S.~Boyd.
\newblock Sensor selection via convex optimization.
\newblock \emph{IEEE Transactions on Signal Processing}, 57\penalty0
  (2):\penalty0 451--462, 2008.

\bibitem[Kay(1993)]{kay1993fundamentals}
S.~M. Kay.
\newblock \emph{Fundamentals of statistical signal processing: Estimation
  theory}.
\newblock Prentice Hall PTR, 1993.

\bibitem[Kempe et~al.(2003)Kempe, Kleinberg, and Tardos]{kempe2003maximizing}
D.~Kempe, J.~Kleinberg, and {\'E}.~Tardos.
\newblock Maximizing the spread of influence through a social network.
\newblock In \emph{Proc. international conference on Knowledge Discovery and
  Data mining}, pages 137--146, 2003.

\bibitem[Khuller et~al.(1999)Khuller, Moss, and Naor]{khuller1999budgeted}
S.~Khuller, A.~Moss, and J.~S. Naor.
\newblock The budgeted maximum coverage problem.
\newblock \emph{Information Processing Letters}, 70\penalty0 (1):\penalty0
  39--45, 1999.

\bibitem[Krause and Guestrin(2005)]{krause2005note}
A.~Krause and C.~Guestrin.
\newblock \emph{A note on the budgeted maximization of submodular functions}.
\newblock Carnegie Mellon University. Center for Automated Learning and
  Discovery, 2005.

\bibitem[Krause et~al.(2008)Krause, Singh, and Guestrin]{krause2008near}
A.~Krause, A.~Singh, and C.~Guestrin.
\newblock Near-optimal sensor placements in {Gaussian} processes: Theory,
  efficient algorithms and empirical studies.
\newblock \emph{Journal of Machine Learning Research}, 9\penalty0
  (Feb):\penalty0 235--284, 2008.

\bibitem[Mei et~al.(2017)Mei, Mohagheghi, Zampieri, and Bullo]{mei2017dynamics}
W.~Mei, S.~Mohagheghi, S.~Zampieri, and F.~Bullo.
\newblock On the dynamics of deterministic epidemic propagation over networks.
\newblock \emph{Annual Reviews in Control}, 44:\penalty0 116--128, 2017.

\bibitem[Mo et~al.(2011)Mo, Ambrosino, and Sinopoli]{mo2011sensor}
Y.~Mo, R.~Ambrosino, and B.~Sinopoli.
\newblock Sensor selection strategies for state estimation in energy
  constrained wireless sensor networks.
\newblock \emph{Automatica}, 47\penalty0 (7):\penalty0 1330--1338, 2011.

\bibitem[Newman(2002)]{newman2002spread}
M.~E. Newman.
\newblock Spread of epidemic disease on networks.
\newblock \emph{Physical Review E}, 66\penalty0 (1):\penalty0 016128, 2002.

\bibitem[Nowzari et~al.(2016)Nowzari, Preciado, and
  Pappas]{nowzari2016analysis}
C.~Nowzari, V.~M. Preciado, and G.~J. Pappas.
\newblock Analysis and control of epidemics: A survey of spreading processes on
  complex networks.
\newblock \emph{IEEE Control Systems Magazine}, 36\penalty0 (1):\penalty0
  26--46, 2016.

\bibitem[Par{\'e} et~al.(2018)Par{\'e}, Liu, Beck, Kirwan, and
  Ba{\c{s}}ar]{pare2018analysis}
P.~E. Par{\'e}, J.~Liu, C.~L. Beck, B.~E. Kirwan, and T.~Ba{\c{s}}ar.
\newblock Analysis, estimation, and validation of discrete-time epidemic
  processes.
\newblock \emph{IEEE Transactions on Control Systems Technology}, 2018.

\bibitem[Pastor-Satorras et~al.(2015)Pastor-Satorras, Castellano, Van~Mieghem,
  and Vespignani]{pastor2015epidemic}
R.~Pastor-Satorras, C.~Castellano, P.~Van~Mieghem, and A.~Vespignani.
\newblock Epidemic processes in complex networks.
\newblock \emph{Reviews of Modern Physics}, 87\penalty0 (3):\penalty0 925,
  2015.

\bibitem[Pezzutto et~al.(2020)Pezzutto, Rossello, Schenato, and
  Garone]{pezzutto2020smart}
M.~Pezzutto, N.~B. Rossello, L.~Schenato, and E.~Garone.
\newblock Smart testing and selective quarantine for the control of epidemics.
\newblock \emph{arXiv preprint arXiv:2007.15412}, 2020.

\bibitem[Preciado et~al.(2014)Preciado, Zargham, Enyioha, Jadbabaie, and
  Pappas]{preciado2014optimal}
V.~M. Preciado, M.~Zargham, C.~Enyioha, A.~Jadbabaie, and G.~J. Pappas.
\newblock Optimal resource allocation for network protection against spreading
  processes.
\newblock \emph{IEEE Transactions on Control of Network Systems}, 1\penalty0
  (1):\penalty0 99--108, 2014.

\bibitem[Prem et~al.(2020)Prem, Liu, Russell, Kucharski, Eggo, Davies, Flasche,
  Clifford, Pearson, Munday, et~al.]{prem2020effect}
K.~Prem, Y.~Liu, T.~W. Russell, A.~J. Kucharski, R.~M. Eggo, N.~Davies,
  S.~Flasche, S.~Clifford, C.~A. Pearson, J.~D. Munday, et~al.
\newblock The effect of control strategies to reduce social mixing on outcomes
  of the {COVID-19} epidemic in {Wuhan}, {China}: a modelling study.
\newblock \emph{The Lancet Public Health}, 2020.

\bibitem[Pukelsheim(2006)]{pukelsheim2006optimal}
F.~Pukelsheim.
\newblock \emph{Optimal design of experiments}.
\newblock SIAM, 2006.

\bibitem[Stoer and Bulirsch(2013)]{stoer2013introduction}
J.~Stoer and R.~Bulirsch.
\newblock \emph{Introduction to numerical analysis}, volume~12.
\newblock Springer Science \& Business Media, 2013.

\bibitem[Streeter and Golovin(2009)]{streeter2009online}
M.~Streeter and D.~Golovin.
\newblock An online algorithm for maximizing submodular functions.
\newblock In \emph{Proc. Advances in Neural Information Processing Systems},
  pages 1577--1584, 2009.

\bibitem[Summers et~al.(2015)Summers, Cortesi, and
  Lygeros]{summers2015submodularity}
T.~H. Summers, F.~L. Cortesi, and J.~Lygeros.
\newblock On submodularity and controllability in complex dynamical networks.
\newblock \emph{IEEE Transactions on Control of Network Systems}, 3\penalty0
  (1):\penalty0 91--101, 2015.

\bibitem[{Tzoumas} et~al.(2020){Tzoumas}, {Carlone}, {Pappas}, and
  {Jadbabaie}]{Tzoumas2020LQG}
V.~{Tzoumas}, L.~{Carlone}, G.~J. {Pappas}, and A.~{Jadbabaie}.
\newblock Lqg control and sensing co-design.
\newblock \emph{IEEE Transactions on Automatic Control}, pages 1--1, 2020.
\newblock \doi{10.1109/TAC.2020.2997661}.

\bibitem[Van~Trees(2004{\natexlab{a}})]{van2004detection}
H.~L. Van~Trees.
\newblock \emph{Detection, estimation, and modulation theory, part I}.
\newblock John Wiley \& Sons, 2004{\natexlab{a}}.

\bibitem[Van~Trees(2004{\natexlab{b}})]{van2004optimum}
H.~L. Van~Trees.
\newblock \emph{Detection, estimation, and modulation theory, part IV: Optimum
  array processing}.
\newblock John Wiley \& Sons, 2004{\natexlab{b}}.

\bibitem[Vrabac et~al.(2020)Vrabac, Par{\'e}, Sandberg, and
  Johansson]{vrabac2020overcoming}
D.~Vrabac, P.~E. Par{\'e}, H.~Sandberg, and K.~H. Johansson.
\newblock Overcoming challenges for estimating virus spread dynamics from data.
\newblock In \emph{Proc. Conference on Information Sciences and Systems}, pages
  1--6. IEEE, 2020.

\bibitem[Ye and Sundaram(2019)]{ye2019sensor}
L.~Ye and S.~Sundaram.
\newblock Sensor selection for hypothesis testing: Complexity and greedy
  algorithms.
\newblock In \emph{Proc. Conference on Decision and Control}, pages 7844--7849.
  IEEE, 2019.

\bibitem[Ye et~al.(2020)Ye, Roy, and Sundaram]{ye2020resilient}
L.~Ye, S.~Roy, and S.~Sundaram.
\newblock Resilient sensor placement for {Kalman} filtering in networked
  systems: Complexity and algorithms.
\newblock \emph{IEEE Transactions on Control of Network Systems}, 7\penalty0
  (4):\penalty0 1870--1881, 2020.

\bibitem[{Ye} et~al.(2021){Ye}, {Woodford}, {Roy}, and
  {Sundaram}]{ye2018complexity}
L.~{Ye}, N.~{Woodford}, S.~{Roy}, and S.~{Sundaram}.
\newblock On the complexity and approximability of optimal sensor selection and
  attack for {Kalman} filtering.
\newblock \emph{IEEE Transactions on Automatic Control}, 66\penalty0
  (5):\penalty0 2146--2161, 2021.

\end{thebibliography}
\end{document}